\documentclass[a4paper,11pt,reqno]{amsart}

\usepackage[margin=35mm]{geometry}
\usepackage[utf8]{inputenc}
\usepackage{amsmath,amssymb,amsthm,amsfonts,amstext,amsopn,amsxtra,dsfont,mathrsfs,esint,enumitem,bookmark,mathtools}
\usepackage[capitalize]{cleveref}
\usepackage{braket}
\usepackage{hyperref}

\theoremstyle{plain}
\newtheorem{theorem}{Theorem}
\newtheorem{lemma}[theorem]{Lemma}
\newtheorem{corollary}[theorem]{Corollary}

	\newcommand\xqed[1]{%
		\leavevmode\unskip\penalty9999 \hbox{}\nobreak\hfill\quad\hbox{#1}%
	}
	\newcommand\remarkend{\xqed{$\triangle$}}
	
	\makeatletter
		\def\@endtheorem{\remarkend\endtrivlist\@endpefalse }
	\makeatother
\theoremstyle{remark}
\newtheorem{remark}[theorem]{Remark}
\makeatletter
	\def\@endtheorem{\endtrivlist\@endpefalse }
\makeatother

\newcommand\C{{\ensuremath {\mathbb C} }}
\newcommand\N{{\ensuremath {\mathbb N} }}
\newcommand\R{{\ensuremath {\mathbb R} }}
\newcommand\1{{\ensuremath {\mathds 1} }}

\newcommand{\bS}{\mathbb{S}}
\newcommand{\gH}{\mathfrak{H}}
\newcommand{\gS}{\mathfrak{S}}
\newcommand{\cB}{\mathcal{B}}
\newcommand{\cE}{\mathcal{E}}
\newcommand{\cF}{\mathcal{F}}
\newcommand{\cG}{\mathcal{G}}
\newcommand{\cK}{\mathcal{K}}
\newcommand{\cL}{\mathcal{L}}
\newcommand{\cM}{\mathcal{M}}
\newcommand{\cN}{\mathcal{N}}
\newcommand{\cO}{\mathcal{O}}
\newcommand{\eps}{\varepsilon}
\newcommand{\bx}{\mathbf{x}}
\def\d{\,\mathrm{d}}
\DeclareMathOperator{\Tr}{\mathrm{Tr}}
\newcommand{\supp}{\mathrm{Supp}\,}
\renewcommand{\ge}{\geqslant}
\renewcommand{\le}{\leqslant}
\renewcommand{\geq}{\geqslant}
\renewcommand{\leq}{\leqslant}
\renewcommand{\hat}{\widehat}
\renewcommand{\tilde}{\widetilde}

\newcommand{\Egp}{\cE_{\mathrm{GP}}}
\newcommand{\GPnrg}{e_{\mathrm{GP}}}
\newcommand{\Vext}{V_{\mathrm{ext}}}
\newcommand{\GPmin}{\cM_{\mathrm{GP}}}
\newcommand{\mf}{\mathrm{mf}}

\newcommand{\wto}{\rightharpoonup}

\newcommand{\<}{\langle}
\renewcommand{\>}{\rangle}
\newcommand{\pscal}[1]{\ensuremath{\left\langle #1 \right\rangle}}
\newcommand{\norm}[1]{\left\vert\kern-0.25ex\left\vert #1 \right\vert\kern-0.25ex\right\vert}
\newcommand{\normt}[1]{\left\vert\kern-0.25ex\left\vert\kern-0.25ex\left\vert #1 \right\vert\kern-0.25ex\right\vert\kern-0.25ex\right\vert}
\newcommand{\ketbra}[2]{\ket{#1}\!\bra{#2}}
\newcommand{\KetBra}[2]{\Ket{#1}\!\Bra{#2}}

\newcommand\nn{\nonumber}

\allowdisplaybreaks

\begin{document}

\title[The condensation of a Bose gas with three\nobreakdash-body interactions]{The condensation of a trapped dilute Bose gas with three\nobreakdash-body interactions}

\author[P.T. Nam]{Phan Th\`anh Nam}
\address{Department of Mathematics, LMU Munich, Theresienstrasse 39, 80333 Munich, and Munich Center for Quantum Science and Technology, Schellingstr. 4, 80799 Munich, Germany}
\email{nam@math.lmu.de}

\author[J. Ricaud]{Julien Ricaud}
\address{Department of Mathematics, LMU Munich, Theresienstrasse 39, 80333 Munich, and Munich Center for Quantum Science and Technology, Schellingstr. 4, 80799 Munich, Germany}
\email{julien.ricaud@polytechnique.edu}

\author[A. Triay]{Arnaud Triay}
\address{Department of Mathematics, LMU Munich, Theresienstrasse 39, 80333 Munich, and Munich Center for Quantum Science and Technology, Schellingstr. 4, 80799 Munich, Germany}
\email{triay@math.lmu.de}

\begin{abstract}
	We consider a trapped dilute gas of $N$ bosons in $\R^3$ interacting via a three\nobreakdash-body interaction potential of the form $N\, V(N^{1/2}(x-y,x-z))$. In the limit $N\to \infty$, we prove that every approximate ground state of the system is a convex superposition of minimizers of a 3D energy-critical nonlinear Schr{\"o}\-dinger functional where the nonlinear coupling constant is proportional to the scattering energy of the interaction potential. In particular, the $N$-body ground state exhibits complete Bose--Einstein condensation if the nonlinear Schr{\"o}\-dinger minimizer is unique up to a complex phase.
\end{abstract}

\maketitle

\tableofcontents

\section{Introduction}
	Bose--Einstein condensation was predicted in 1924~\cite{Bose-24,Einstein-24}, but it was not until 1995 that it was first realized experimentally in gases of alkali atoms~\cite{AndEnsMatWieCor-95,DavMewAndvDrDurKurKet-95}, leading to the 2001 Nobel prize in Physics of Cornell, Wieman, and Ketterle. While the theoretical works of Bose and Einstein~\cite{Bose-24,Einstein-24,Einstein-25} focused only on the ideal gas, interactions among particles do not only make the analysis more involved but also lead to new phenomena, one of the most famous being superfluidity~\cite{Landau-41,Bogoliubov-47b}. In dilute regimes, where collisions occur rarely, the interaction is most often described by an effective two\nobreakdash-body potential $v_{\mathrm{2B}}(x_1-x_2)$ which is computed by integrating out the possible internal degrees of freedom of two particles placed at $x_1$ and $x_2$, in the ideal situation where all the others are infinitely far away. This approximation, enough for most applications, might break down if a third particle nearby affects strongly the internal structure of the two others. In this case, one can add a correction by using another effective interaction potential $v_{\mathrm{3B}}(x_1-x_2,x_1-x_3)$ in the ideal situation where three-particles are close and the others infinitely far away, etc. ---namely,
\[
	U(x_1,\dots,x_N) \simeq \sum_{1\leq i < j \leq N} v_{\mathrm{2B}}(x_i-x_j) + \sum_{1\leq i < j < k \leq N} v_{\mathrm{3B}}(x_i-x_j,x_i-x_k) + \dots\,.
\]
In this many\nobreakdash-body expansion, the three- and higher-order corrections, although small, are not always negligible: three\nobreakdash-body interactions account for $2\%$ of the binding energy of liquid $\text{He}^4$~\cite{MurBar-71} and even $14\%$ for water~\cite{MasBukSza-03}, preventing the two\nobreakdash-body approximation from explaining certain of their physical properties~\cite{PieTaiSki-11,UjeVit-07}. In the realm of condensed matter, the Bose--Hubbard model with two\nobreakdash-body and three\nobreakdash-body interactions was derived from cold polar molecules, where the two\nobreakdash-body interaction can independently be tuned and even switched off~\cite{BucMicZol-07}. Finally, in the study of ultracold gases, three\nobreakdash-body interactions have received a strong interest with the hope of observing exotic states like self-trapped droplets or Pfaffian states~\cite{Petrov-14}. In particular, a repulsive three\nobreakdash-body interaction can stabilize a condensate against collapse due to an attractive two\nobreakdash-body interaction, the competing contributions can lead to crystallization and is believed to be a good candidate for observing super-solid states~\cite{BisBla-15,Blakie-16}.

	In this work, we consider a system of $N$ bosons in $\R^{3}$ trapped by a confining potential and interacting via three\nobreakdash-body interactions. It would be also interesting to include two\nobreakdash-body interactions in the model, but we do not do so here in order to simplify the problem. We study the Gross--Pitaevskii-like regime where the interaction potential scales like
\[
	V_N(x-y,x-z) = N V\!\left( N^{1/2}(x-y,x-z) \right).
\]
We prove that the ground state energy and approximate ground states of the system are effectively described by the 3D energy-critical nonlinear Schr{\"o}\-dinger functional where the nonlinear term is proportional to the scattering energy of the interaction potential. To our knowledge, it is the first time that the ground state problem of a dilute Bose gas with three\nobreakdash-body interactions is investigated rigorously. In fact, a simpler result with the weaker interaction potential of the form
\[
	N^{6\beta-2} V\!\left(N^{\beta}(x-y,x-z)\right)
\]
with $\beta \ge 0$ small can be handled by following the method in~\cite{LewNamRou-16} (see also the subsequent~\cite{NamRou-20,Rougerie-21,Triay-18} for further developments of relevant techniques). In this case, the usual mean-field approximation is correct to the leading order. The critical case $\beta=1/2$ that we consider here is more difficult since the three\nobreakdash-body correlation between particles yields a leading order correction to the ground state energy, which is similar to the Gross--Pitaevskii regime of the two\nobreakdash-body interaction studied in~\cite{LieSeiYng-00,LieSei-02,LieSei-06,NamRouSei-16}. We refer to~\cite{LieSeiSolYng-05,Rougerie-15-arxiv,Rougerie-20b} for pedagogical reviews on relevant results in the two-body interaction case. 

	In the time-dependent problem, the potential $N^{6\beta-2} V\!\left( N^{\beta}(x-y,x-z) \right) $ with $\beta \ge 0$ small already received some attention~\cite{ChePav-11,Chen-12,Yuan-15,CheHol-19,Lee-20,NamSal-20,LiYao-21}. The critical regime $\beta=1/2$ is again more difficult and the corresponding dynamical problem remains open. In the two-body interaction case, the Gross--Pitaevskii regime was first studied in~\cite{ErdSchYau-10,ErdSchYau-09}, see also the reviews~\cite{Schlein-13,Golse-16,BenPorSch-16} for further results.

\subsection{Model}
	Let us now explain the mathematical setting in detail. We consider a system of $N\geq3$ identical bosons in $\R^3$ described by the Hamiltonian
\begin{equation}\label{eq:HN}
	H_N=\sum_{i=1}^N h_i + \sum_{1\le i<j<k \le N} V_N(x_i-x_j,x_i-x_k )\,,
\end{equation}
acting on the symmetric space
\[
	\gH^N = L^2_s \left( (\R^3)^N \right) = \bigotimes^N_s L^2(\R^3)\,,
\]
where, following the most mathematically demanding setting in the two\nobreakdash-body interaction case~\cite{LieSei-06,NamRouSei-16}, we consider the one-body operator of the form
\begin{equation} \label{Def_h}
	h:= (-{\bf i}\nabla_{x} + A(x))^2 +\Vext(x)\quad \text{on }L^2(\R^3)\,.
\end{equation}
Here, we assume that the external trapping potential satisfies
\begin{equation}\label{eq:Vext}
	\Vext \in L^{\infty}_{\mathrm{loc}}(\R^3,\R) \quad \text{and} \quad \Vext(x) \ge C |x|^\alpha +1 \quad \text{for a.e. } x\in \R^3\,,
\end{equation}
for some constant $\alpha>0$, while the vector potential (accounting for a magnetic field or a rotation) satisfies
\begin{equation}
	\label{eq:A}
	A \in L^3_{\mathrm{loc}}(\mathbb{R}^{3}) \quad \text{and} \quad \lim_{|x| \to \infty} |A(x)|^2 e^{-C|x|} = 0\,,
\end{equation}
for some constant $C>0$. Moreover, the interaction potential is chosen of the form
\begin{equation}\label{eq:VN}
	V_N(x,y)=N V ( N^{1/2} (x,y)) \quad \text{for all } x,y\in \R^3\,,
\end{equation}
where $V:\R^3\times \R^3\to \R$ is nonnegative, bounded, compactly supported, and has the symmetry properties
\begin{equation}\label{eq:sym}
	V(x,y) = V(y,x) \quad \text{and} \quad V(x-y,x-z) = V(y-x,y-z) = V(z-y,z-x)
\end{equation}
which ensure that $H_N$ leaves the bosonic space $\gH^N$ invariant. Thus, $H_N$ models a trapped dilute Bose gas where the range of the interaction potential is much smaller than the average distance between particles: $N^{-1/2}\ll N^{-1/3}$.

	Under our conditions~\eqref{eq:Vext}--\eqref{eq:sym}, the Hamiltonian $H_N$ is well-defined and bounded from below with core domain $\bigotimes^N_s C^\infty_c(\R^3)$. Hence, it can be extended to a self-adjoint operator on $\gH^N$ by Friedrichs' method. The extension, still denoted by $H_N$, is bounded from below and has compact resolvent since $\Vext(x)\to \infty$ as $|x|\to \infty$. In particular, a ground state exists.

	In the present paper, we are interested in the ground state energy
\begin{equation}\label{Def_E_N}
	E_N = \inf\limits_{\norm{\Psi}_{\gH^N}=1} \pscal{\Psi, H_N \Psi}
\end{equation}
and the corresponding ground states in the limit $N\to \infty$. When $N$ becomes large, it is in general impossible to compute, both analytically and numerically, the ground state energy $E_N$ and the ground states from the full many\nobreakdash-body description~\eqref{eq:HN}. Hence, for practical computations it is important to derive an effective theory which relies on less variables.

	The usual mean-field approximation suggests to restrict to the complete condensation ansatz of $N$ particles
\begin{equation} \label{eq:formal-Psi}
	\Psi_N \approx u^{\otimes N}(x_1,\dots ,x_N)= \prod_{j=1}^N u(x_j) \quad \text{with } u \in L^2(\R^3)\,.
\end{equation}
However, the correlation due to the strong interaction between particles at short distances will also play a role to leading order. Since the gas is dilute, the correlation structure is encoded in the scattering problem associated to the three\nobreakdash-body interaction potential $V_N(x-y,x-z)$. Let us consider the three\nobreakdash-body operator
\[
	-\Delta_{x_1} - \Delta_{x_2} - \Delta_{x_3} + V_N(x_1-x_2,x_1-x_3) \quad \text{on }L^2\!\left( (\R^3)^3 \right).
\]
We can remove the center of mass by change of variables:
\begin{equation}\label{eq:removing-center-of-mass}
	r_{1} = \frac{1}{3}(x_1+x_2+x_3), \quad r_2 = x_1 -x_2, \quad \text{and} \quad r_3 =x_1-x_3\,.
\end{equation}
Denoting the momentum variable $p_x= -{\bf i}\nabla_x$, with $x\in \R^3$ and ${\bf i}^2=-1$, we have
\[
	p_{x_1} = \frac{1}{3} p_{r_1} + p_{r_2} + p_{r_3}, \quad p_{x_2} = \frac{1}{3} p_{r_1} - p_{r_2} ,\quad p_{x_3} = \frac{1}{3} p_{r_1}-p_{r_3}\,,
\]
and, consequently,
\begin{multline}\label{eq:remove-center}
	-\Delta_{x_1} - \Delta_{x_2} - \Delta_{x_3} + V_N(x_1-x_2,x_1-x_3)\\
	= \left( \frac{1}{3} p_{r_1} + p_{r_2} + p_{r_3} \right)^2 + \left(\frac{1}{3} p_{r_1} - p_{r_2} \right)^2 + \left( \frac{1}{3} p_{r_1}-p_{r_3} \right)^2 + V_N(r_1,r_2) \\
	= \frac{1}{3}p_{r_1}^2 + 2(p_{r_2}^2 + p_{r_3}^2 + p_{r_2} p_{r_3} ) + V_N(r_2,r_3)\,.
\end{multline}
Thus, after removing the center of mass, we are left with the two\nobreakdash-body operator
\[
	2(p_{x}^2 + p_{y}^2 + p_{x} p_{y} ) + V_N(x,y) = -2\Delta_{\cM} + V_N(x,y)\quad \text{on } L^2\!\left( (\R^3)^2 \right),
\]
where $-\Delta_{\cM}=|\cM \nabla_{\R^6}|^2 = \mathrm{div}_{\R^6}(\cM^2 \nabla_{\R^6})$ and the matrix $\cM: \R^3\times \R^3 \to \R^3\times \R^3$ is given by
\[
	\cM:= \left( \frac{1}{2}
	\begin{pmatrix}
		2 & 1 \\
		1 & 2
	\end{pmatrix}
	\right)^{1/2} = \frac{1}{2\sqrt{2}}
	\begin{pmatrix}
		\sqrt{3}+1 & \sqrt{3}-1 \\
		\sqrt{3}-1 & \sqrt{3}+1
	\end{pmatrix}.
\]
The operator $-2\Delta_{\cM} + V_N(x,y)$ is associated to the (modified) {\em scattering energy}
\[
	b_{\cM}(V_N) : = \inf_{\omega \in \dot H^1(\R^6)} \int_{\R^6} \left( 2| \cM \nabla \omega( \bx)|^2 + V_N(\bx) |1-\omega(\bx)|^2 \right) \d \bx\,.
\]
Recall the standard scattering energy of a potential $W: \R^d \to \R$
\begin{equation} \label{eq:intro-def-scat-energy}
	b(W) : = \inf_{\omega \in \dot H^1(\R^d)} \int_{\R^d} \left( 2| \nabla \omega ( \bx)|^2 + W(\bx) |1-\omega(\bx)|^2 \right) \d \bx\,.
\end{equation}
By change of variables, the modified scattering energy can be written as
\[
	b_\cM(V)=b(V( \cM\cdot )) \det \cM \,.
\]
Here, $\det \cM = 3\sqrt 3/8$. Moreover, the specific choice in~\eqref{eq:VN} ensures that
\[
	b_\cM(V_N)= \frac{b_\cM(V)}{N^2}\,.
\]
In summary, to the leading order we expect that
\begin{equation} \label{eq:formal-V}
	V_N(x,y) \approx \frac{b_{\cM}(V)}{N^2} \delta_{x=y=0}\,.
\end{equation}
Putting the formal approximations~\eqref{eq:formal-Psi} and~\eqref{eq:formal-V} together, we arrive at
\begin{equation}\label{eq:NLS-GSE}
	\frac{E_N}{N} \approx \GPnrg := \inf_{\norm{u}_{L^2(\R^3)}=1} \Egp(u)\,,
\end{equation}
where $\Egp(u)$ is the 3D energy-critical nonlinear Schr{\"o}\-dinger (NLS) functional
\begin{equation}\label{eq:NLS}
	\Egp(u) := \int_{\R^3} \left( |(-{\bf i}\nabla + A(x)) u(x)|^2 + \Vext(x) |u(x)|^2 + \frac{b_{\cM}(V)}{6} |u(x)|^6 \right) \d x\,.
\end{equation}

	Let us denote by $\GPmin$ the set of minimizers of $\Egp$. The existence of a minimizer $u_0$ of $\Egp$ follows straightforwardly from standard variational methods, where the compactness in $L^2(\R^3)$ of minimizing sequences is guaranteed by the trapping condition $\Vext(x)\to \infty$ as $|x|\to \infty$. The solution solves the nonlinear equation
\[
	\left(h + b_{\cM} (V) |u_0(x)|^4/2 - \eps_0 \right) u_0(x)=0 \quad \text{for all } x\in \R^3\,,
\]
for some chemical potential $\eps_0\in \R^{3}$ (the Lagrange multiplier associated with the mass constraint $\norm{u}_{L^2(\R^3)}=1$). Moreover, with our assumptions on $\Vext $ and $A$, the absolute value of the resolvent kernel of $h + b_{\cM} (V) |u_0|^4/2+1$ is bounded by the one of $-\Delta+1$ (see for instance~\cite[Sect.~15]{Simon-79b}), which implies $u_0 \in L^\infty (\R^{3})$ by the Sobolev embedding $H^2(\mathbb{R}^{3}) \subset L^\infty(\mathbb{R}^{3})$.

	In the absence of the magnetic field, i.e.~$A\equiv 0$, the minimizer $u_0$ of $\Egp$ is unique up to a complex phase (this can be seen by a standard convexity argument). However, in the general case $A\not\equiv 0$, $\Egp$ may have several minimizers which indicates the presence quantized vortices. We refer to~\cite{LieSeiYng-00,LieSei-06,NamRouSei-16} and the reviews~\cite{LieSeiSolYng-05,Rougerie-15-arxiv,Rougerie-20b} for the related discussions in the two\nobreakdash-body interaction case where the nonlinear term $|u|^6$ in~\eqref{eq:NLS} is replaced by $|u|^4$.

\subsection{Main results}
	Our first result is a rigorous justification of~\eqref{eq:NLS-GSE}.
\begin{theorem}[Ground state energy]\label{thm:main1} Let $\Vext $, $A$, and $V_N$ satisfy~\eqref{eq:Vext}--\eqref{eq:sym}. Then, the ground state energy of $H_N$ in~\eqref{eq:HN} satisfies
	\begin{equation} \label{eq:Energy-CV}
		\lim_{N\to \infty} \frac{E_N}{N} = \GPnrg = \inf_{\norm{u}_{L^2(\R^3)}=1} \Egp(u)\,,
	\end{equation}
	where the effective functional $\Egp(u)$ is given in~\eqref{eq:NLS}.
\end{theorem}

	As a by-product of our proof of Theorem~\ref{thm:main1}, we also obtain that every approximate ground state of $H_N$ behaves as a convex superposition of the pure tensor products as in~\eqref{eq:formal-Psi}. Note that the approximation~\eqref{eq:formal-Psi} is expected to hold not in the norm topology of $\gH^N$, but rather in a weaker topology defined by reduced density matrices. Recall that for every $1\le k\le N$, the $k$\nobreakdash-body density matrix $\gamma_{\Psi_N}^{(k)}$ of a normalized wave function $\Psi_N \in \gH^N$ is a nonnegative trace class operator on $\gH^k=L_s^2((\R^3)^k)$ with kernel\footnote{In our convention, inner products in (complex) Hilbert spaces are linear in the second argument and anti-linear in the first.}
\[
	\gamma_{\Psi_N}^{(k)}({\bf z};{\bf z}') = \int_{(\R^3)^{N-k}} \Psi_N ({\bf z},x_{k+1}, \dots , x_N) \overline{\Psi_N({\bf z}',x_{k+1},\dots,x_N)} \d x_{k+1}\cdots \d x_{N}\,.
\]
Equivalently, we can also write $\gamma_{\Psi_N}^{(k)}=\Tr_{k+1\to N} \ketbra{\Psi_N}{\Psi_N}$, where the partial trace is taken over all but the first $k$ variables. The proper meaning of~\eqref{eq:formal-Psi} is
\begin{equation}\label{eq:1pdm-formal}
	\gamma_{\Psi_N}^{(k)} \approx \ketbra{ u^{\otimes k} }{ u^{\otimes k} }, \quad k= 1, 2, \dots\,,
\end{equation}
which is often referred to as the complete Bose--Einstein condensation.

	Our second result is a rigorous justification of~\eqref{eq:1pdm-formal} for ground states, or more generally for approximate ground states, of $H_N$.
\begin{theorem}[Condensation of approximate ground states]\label{thm:main2}
	Let $\Vext $, $A$, and $V_N$ be as in Theorem~\ref{thm:main1}. Assume that $\Psi_N$ is an approximate ground state for $H_N$ ---namely,
	\[
		\norm{\Psi_N}_{\gH^N}=1 \quad \text{and} \quad \lim_{N\to \infty} \frac{ \pscal{ \Psi_N, H_N \Psi_N} }{ N} = \GPnrg\,.
	\]
	Then, up to a subsequence as $N\to \infty$, there exists a Borel probability measure supported on $\GPmin$, the set of minimizers of $\Egp$, such that
	\[
		\lim_{N\to \infty} \gamma_{\Psi_{N}}^{(k)} = \int_{\GPmin} \ketbra{ u^{\otimes k} }{ u^{\otimes k} } \d \mu (u)\,, \quad k=1,2,\dots\,,
	\]
	in trace norm.
	
	In particular, if $\Egp$ has a unique minimizer $u_0$ (modulo a complex phase), then the whole sequence $\{\gamma^{(k)}_{\Psi_{N}}\}_N$ converges towards $\ketbra{ u_0^{\otimes k} }{ u_0^{\otimes k} }$ for all $k\in \N\setminus\{0\}$.
\end{theorem}

	Let us give some quick remarks about our results.
\begin{remark}[Less singular interactions] From our approach, we also find that for every $0<\beta<1/2$, the ground state energy of the Hamiltonian
	\begin{equation}\label{eq:HN-mf}
		H_{N,\mf}=\sum_{i=1}^N h_i + \sum_{1\le i<j<k \le N} N^{6\beta-2}V(N^{\beta}(x_i-x_j, x_i-x_k))
	\end{equation}
	satisfies
	\begin{equation} \label{eq:Energy-CV-MF}
		\lim_{N\to \infty} \frac{E_N}{N} = \inf_{\norm{u}_{L^2(\R^3)}=1} \left\{ \pscal{ u, h u }_{L^2(\mathbb{R}^{3})} + \frac{\hat V(0)}{6} \int_{\R^3} |u(x)|^6 \d x \right\}.
	\end{equation}
	This is an analogue of~\eqref{eq:Energy-CV} where the scattering energy $b_{\cM}(V)$ in the nonlinear functional is replaced by its first Born approximation $\hat V(0)=\int_{\R^6} V$. Consequently, we also obtain the convergence of approximate ground states of $H_{N,\mf}$ to the minimizer of the right side of~\eqref{eq:Energy-CV-MF}.

	When $\beta>0$ is small (depending on the growth of $\Vext$),~\eqref{eq:Energy-CV-MF} can be proved by a standard mean-field technique, for example using a quantitative quantum de Finetti theorem as proposed by Lewin, Nam, and Rougerie~\cite{LewNamRou-16}. The proof of~\eqref{eq:Energy-CV-MF} for the whole range $0<\beta<1/2$ is more difficult and can be obtained from a simplification of our method. The critical case $\beta=1/2$ in Theorem~\ref{thm:main1} is the hardest one since the strong correlation yields a leading order correction to the mean-field approximation.
\end{remark}

\begin{remark}[Dynamical problem]
	In the context of quantum dynamics, it was proved by Chen and Holmer~\cite{CheHol-19} for $\beta < 1/9$, then by Nam and Salzmann~\cite{NamSal-20} for $\beta < 1/6$, that the Schr{\"o}\-dinger evolution $\Psi_N(t)=e^{itH^0_{N,\mf}} \Psi_N(0)$ with initial condition $\gamma_{\Psi_N(0)}^{(1)} \approx \ketbra{u(0)}{u(0)}$ exhibits the complete condensation $\gamma_{\Psi_N(t)}^{(1)} \approx \ketbra{u(t)}{u(t)}$ for any time $t>0$, where the condensate $u(t)$ is determined by the time-dependent equation
	\begin{equation}\label{eq:NLS-eq-time}
		{\bf i} \partial_t u(t,x) = \left(-\Delta + \frac{1}{2} \hat V(0) |u(t,x)|^4 \right) u(t,x)\,.
	\end{equation}
	Here, $ H^0_{N,\mf}$ is the Hamiltonian as in~\eqref{eq:HN-mf} but without the external potential and the magnetic field. We expect that the same result holds for all $0<\beta<1/2$ and that the factor $\hat V(0)$ in~\eqref{eq:NLS-eq-time} is replaced in the critical case $\beta=1/2$ by the scattering energy $b_{\cM}(V)$. This is still an open question. The results in the present paper justify the initial condition of the dynamical problem.
\end{remark}

\begin{remark}[Comparison to the two\nobreakdash-body interaction case]
	The derivation of the Gross--Pitaevskii functional, the analogue of~\eqref{eq:NLS} with $|u|^6$ replaced by $|u|^4$, from the many\nobreakdash-body problem with a two\nobreakdash-body interaction potential of the form $N^2 V(N(x-y))$ goes back to the seminal papers of Lieb, Seiringer, and Yngvason~\cite{LieYng-98,LieSeiYng-00,LieSei-02,LieSei-06}. There are four levels of difficulty. First, for the homogeneous system of $N$ bosons trapped in a unit torus ($A\equiv 0$ and $\Vext \equiv 0$), the convergence of the ground state energy (similar to Theorem~\ref{thm:main1}) follows from the analysis in~\cite{LieYng-98}. Second, the extension of the energy convergence to the inhomogeneous trapped case in $\R^3$ ($A\equiv 0$ and $\Vext \not\equiv 0$) was solved in~\cite{LieSeiYng-00}. Third, the proof of the Bose--Einstein condensation (similar to Theorem~\ref{thm:main2}), which is harder, was first achieved in~\cite{LieSei-02}. Finally, in the most mathematically demanding setting of a general trapped case with a magnetic field ($A\not\equiv 0$ and $\Vext \not\equiv 0$), the convergences of both the energy and the states were settled in~\cite{LieSei-06}. See also~\cite{NamRouSei-16} for an alternative proof.

	Here, we aim at extending the results from the two\nobreakdash-body interaction case, in the most difficult setting from~\cite{LieSei-06,NamRouSei-16}, to the three\nobreakdash-body interaction case. For the homogeneous gas, a simpler version of our analysis can be found in~\cite{NamRicTri-22a} where we combine a variant of Dyson's lemma of the present paper with the argument in~\cite{LieYng-98} in order to quickly derive the energy convergence (the analysis can even be done in the thermodynamic limit). On the other hand, the proof in the present paper is much more complicated than that in~\cite{NamRicTri-22a} since we handle the full generality of the one-body operator and we prove Bose--Einstein condensation, for which the argument in~\cite{NamRicTri-22a} is insufficient.

	In comparison to the existing works in the two\nobreakdash-body interaction case~\cite{LieYng-98,LieSeiYng-00,LieSei-02,LieSei-06,NamRouSei-16}, the analysis in the three\nobreakdash-body interaction case requires three new ingredients. First, we are unable to derive the energy upper bound from a Jastrow--Dyson type state and a cubic (in annihilation/creation operators) transformation is hence needed in order to create the correct correlation structure. Second, for the energy lower bound, we have to extend Dyson's lemma to the new scattering problem, which is in particular relevant for non-radial potentials. Third, the core novel technique of our proof lies in a bootstrap argument where we repeatedly apply Dyson's lemma in order to implement the mean-field approximation. This new technique is not only crucial to handle the three\nobreakdash-body interaction case, which is energy-critical, but is also helpful to simplify the proof in the two\nobreakdash-body interaction case (see Appendix~\ref{app}). More details are given below.
\end{remark}

\subsection{Ingredients of the proof}
	Now let us explain the main ingredients of the proof.

\medskip
\noindent {\em Upper bound.}
	The uncorrelated trial state as in~\eqref{eq:formal-Psi} is insufficient to get the leading order upper bound $E_N\le N\GPnrg+ o(N)$. In fact, the energy per particle over the Hartree states $u^{\otimes N}$ is given by a functional similar to $\Egp(u)$ but where the scattering energy $b_{\cM}(V)$ in front of the nonlinear term $|u|^6$ is replaced by $\hat V(0)=\int_{\R^6} V$. Hence, it is important to take some correlation into account. Heuristically, we can think of the Jastrow--Dyson type state
\begin{equation} \label{eq:Jastrow--Dyson}
	\Psi_N(x_1,\dots ,x_N)= \prod_{j=1}^N u_0(x_j) \prod_{p<k<\ell}^N f_N(x_p-x_k,x_p-x_\ell)\,,
\end{equation}
where $u_0$ is a minimizer of $\GPnrg$ and $f_N: \R^3\times \R^3\to \R$ is a function satisfying the symmetry~\eqref{eq:sym} and solving, for almost every $\bx\in \R^6$, the scattering equation
\[
	-2\Delta_\cM f_N (\bx) + V_N(\bx) f_N(\bx) =0\,,
\]
which is equivalent to solve, for almost every $x,y,z\in \R^3$, the equation
\begin{equation} \label{eq:def-fN}
	(-\Delta_x - \Delta_y - \Delta_z) f_N(x-y,x-z) + (V_N f_N) (x-y,x-z) =0\,.
\end{equation}
The existence of such a function is proved in Theorem~\ref{thm:scattering-M}. In particular, the scattering energy is encoded in $f_N$ as
\[
	\int_{\R^6} V_N f_N = b_{\cM}(V_N) = N^{-2}b_{\cM}(V)\,.
\]
Unfortunately, we are not able to compute the ground state energy per particle of the trial state~\eqref{eq:Jastrow--Dyson} to the leading order, even if we replace $f_{N}$ by a modified version $f_{N,\ell}$ satisfying $f_{N,\ell}(x)=1$ for $|x|\ge \ell$: the computation is significantly more complicated than for the two\nobreakdash-body interaction case.
	
	Here, we follow an alternative approach. First, for ease of computation, we extend~$H_N$ to the operator~$\mathbb{H}_N = 0\oplus \bigoplus_{M=1}^\infty H_{M,N}$ acting on the bosonic Fock space~$\cF(\gH) :=\C \oplus \bigoplus_{M=1}^\infty \gH^M$
\[
	H_{M,N} = \sum_{j=1}^M h_j + \sum_{1\leq i<j<k \leq M} N V(N^{1/2}(x_i-x_j, x_i-x_k))\,.
\]
This extension can be written conveniently as
\[
	\mathbb{H}_N =\int_{\R^{3}} a^*_x h_x a_x \d x + \frac{1}{6} \int_{(\R^{3})^3} V_N(x-y,x-z) a^*_x a^*_y a^*_z a_x a_y a_z \d x \d y \d z
\]
using the standard creation and annihilation operators $a_x^*, a_x$. To capture the condensation, we define the Weyl operator
\[
	W(f) = \exp (a^*(f) - a(f)) \quad \text{for all } f\in L^2(\R^{3})\,,
\]
which is a unitary operator on the Fock space and satisfies
\[
	W(f)^* a(g) W(f) = a(g) + \pscal{ g, f } \quad \text{for all } f,g\in L^2(\R^{3})\,.
\]
In order to create the desired correlation structure encoded in the scattering solution $f_N$ in~\eqref{eq:def-fN}, we introduce another unitary transformation
\[
	U_N := \exp \left[ \1(\cN \le N^{1/2}) B_1^* - B_1 \1(\cN \le N^{1/2}) \right],
\]
where $\cN = \int a^*_x a^{\phantom*}_x dx$ is the number operator on the Fock space and
\begin{equation}\label{intro-Def_B1}
	B_1^* = -\frac{1}{6} N^{\frac32} \int_{(\R^{3})^3} (1-f_N)(x-y,x-z) u_0 (x) u_0 (y) u_0 (z) a^*_x a^*_y a^*_z \d x \d y \d z \,.
\end{equation}
We choose the cut\nobreakdash-off on the particle number to be $N^{1/2}$ such that $U_N$ does not create too many excited particles. We will prove (see Theorem~\ref{theo:upper_bound_fock_space} below) that
\[
	\< \Omega, U_N^* W(\sqrt N \varphi)^* \mathbb{H}_N W(\sqrt N \varphi) U_N \Omega \> \leq N \GPnrg + \cO(N^{2/3})\,,
\]
where $\Omega$ is the vacuum.
Note that the trial state $W(\sqrt N \varphi) U_N \Omega$ does not belong to the $N$\nobreakdash-body space $\gH^N$, but it essentially lives on sectors of $N+ O(\sqrt{N})$ particle number. Following an idea of Solovej~\cite{Solovej-06}, by controlling the fluctuations of the particle number of $W(\sqrt N \varphi) U_N \Omega$, we are able to construct a trial state in the $N$\nobreakdash-body Hilbert space and to obtain the desired upper bound
\[
	E_N \leq N \GPnrg +o(N)\,.
\]

	In the two\nobreakdash-body interaction case, a similar trial state has been used~\cite{BenPorSch-16,NamNapRicTri-21} where equation~\eqref{intro-Def_B1} is replaced by a kernel that is quadratic in terms of creation and annihilation operators, which simplifies the computation greatly since $U_N \Omega$ is a quasi-free state. In our case, the computation with the cubic kernel in equation~\eqref{intro-Def_B1} is more complicated, but technically manageable. The details will be explained in Section~\ref{sec:upper}.

\medskip
\noindent {\em Lower bound.}
	We will follow the overall strategy from the two-body interaction case~\cite{LieSei-06,NamRouSei-16}, namely we replace the singular potential $V_N=N V(N^{1/2} \cdot)$ by a softer potential using a Dyson lemma and then we apply the mean-field approximation. However, to handle the three-body interaction case, we have to use the Dyson lemma several times (instead of only one time as in the two-body interaction case) and this iteration procedure requires new ideas which eventually lead to a substantial improvement over the overall strategy.

Let us quickly explain  our approach. The general idea of the Dyson lemma is that for any $0<\nu\le 1$ and $0<\beta'<\beta \le 1/2$, we have the operator inequality 
\begin{equation} \label{eq:Dyson-easy}
	-2\nu \Delta_{\cM} + V_{N,\beta}({\bf x}) \geq U_{N,\beta'}({\bf x}) \left( 1- CN^{\beta'-\beta}\right)
\end{equation}
on  $L^2(\R^6)$, 	with the scaling convention  
\[
	V_{N,\beta}({\bf x})=N^{6\beta-2} V(N^\beta {\bf x})\,, \quad U_{N,\beta'}({\bf x})= N^{6\beta'-2} U(N^{\beta'}{\bf x})
\]
where $V,U \in C_c^\infty(\R^6)$ are essentially fixed and 
\[
	\int_{\R^6} U = \nu N^{2} b_{\cM}(\nu^{-1} V_{N,\beta})\,.
\]
In our first use of the Dyson lemma, by taking $\nu=1$ we can replace the original Gross--Pitaevskii scaling $\beta=1/2$ by a simpler scaling $0<\beta'<1/2$ with  
\[
	\int_{\R^6} U_{N,\beta'} =  b_{\cM}(V_{N,\beta=1/2})= N^{-2} b_{\cM}(V)=\text{the desired scattering energy}.
\]
While the above estimate holds for all $0<\beta<1/2$, lifting it to the many-body level requires the additional condition $\beta>1/3$ which is technically needed to control several error terms. On the other hand, the mean-field techniques in~\cite{LewNamRou-16} only work for a smaller $\beta$. 

To reduce further $\beta$, we will apply the Dyson lemma again. Note that thanks to the sub-critical scaling $\beta'<1/2$, the equality  
\[
	\nu  b_{\cM} (\nu^{-1} U_{N,\beta'}) = \int_{\R^6} U_{N,\beta'} (1+o(1)) = N^{-2} b_{\cM}(V) (1+ o(1)) 
\]
holds for all $N^{\beta'-1/2} \ll \nu \le 1$, namely the scattering energy of $\nu^{-1} U_{N,\beta'}$ is well approximated by its first Born approximation $\int_{\R^6} \nu^{-1} U_{N,\beta'}$. 
 This allows us to apply the Dyson lemma  with some $0<\nu \ll 1$, namely we sacrifice very little kinetic energy and still get 
\begin{equation} \label{eq:Dyson-easy-subcritical}
	-2\nu \Delta_{\cM} + U_{N,\beta'}({\bf x}) \geq U_{N,\beta''}({\bf x}) (1+ o(1))
\end{equation}
with $0<\beta'' < \beta'$. Repeating this step finitely many times, we end up with a soft potential which can be handled by the techniques in~\cite{LewNamRou-16}. 

On the technical side, we will derive a many-body version of~\eqref{eq:Dyson-easy} with suitable cut-offs in the configuration and momentum spaces. In order to control various error terms and make the iteration procedure work, we will use the bosonic symmetry to adjust the relevant number of particles in each step of the bootstrap argument. 

Now let us go for a more detailed explanation of our lower bound proof. 

\medskip
{\bf Step 1: Dyson's lemma.}	We will prove in Theorem~\ref{dyson_lemma_Rd_nonradial-M} that, given the potential $V_N \ge 0$ supported in $|{\bf x}| \le O(N^{-1/2})$ and $R\gg N^{-1/2}$, we can find a function $U\ge 0$ in $L^1(\R^6)$ supported in $\{|{\bf x}|\le R\}$ such that $U_R$ satisfies the symmetry~\eqref{eq:sym},
\[
	U_R(x,y)\le C R^{-6} \1_{\{|x|\le R\}} \1_{\{|y|\le R\}}\,,\quad \int_{\R^6} U_R = b_{\cM}(V) (1+ o(1)_{N\to \infty})\,,
\]
and
\begin{equation} \label{eq:intro-Dyson}
	-2 \Delta_{\cM}  + V_N (\bx) \geq N^{-2} {U_R} (\bx)\quad \text{on }L^2\!\left( \R^6 \right).
\end{equation}
Actually we will derive need an improved version of~\eqref{eq:intro-Dyson}, with $\cM \nabla_{\bf x} \1_{\{|{\bf x}| \le \sqrt{2}R\}} \cM$ instead of $\Delta_{\cM}$, but let us ignore the technical cut-off in the introductory discussion.    

Note that all existing proofs of the Dyson lemma and its generalizations rely on the radial symmetry of the potential (see~\cite{LieSeiYng-00,LieSeiSol-05}), which is not satisfied by our potential $V:\R^6\to \R$. We will derive~\eqref{eq:intro-Dyson} from a general result on the standard scattering energy (see Theorem~\ref{dyson_lemma_Rd_nonradial}) which holds for a large class of potentials and could be of independent interest.

Now, coming back to the Hamiltonian $H_N$ and using~\eqref{eq:intro-Dyson}, we obtain in Lemma~\ref{lem:gen_dyson_lemma} the following lower bound for all $1>\eps>0$,
\begin{align} \label{eq:intro-many-body-Dyson}
	(1-&\eps)^{-1} H_N + C_\eps R^3 N^2 \\
	&\ge \sum_{i=1}^N \widetilde{h}_i + \frac{1}{6 N^2} \sum_{\substack{ 1\le i,j,k \leq N \\ i\neq j \neq k \neq i } } {U_R}(x_i-x_j, x_i-x_k) \prod_{\ell \neq i,j,k} \theta_{2R}\left(\frac{x_i+x_j+x_k}{3}-x_\ell \right), \nn
\end{align}
where
\[
	\widetilde{h} = h - (1-\eps) p^2 \1_{\{|p| > \eps^{-1}\}} \quad \text{and} \quad \theta_R:=\1_{\{|x| > R\}}\,.
\]
Note that~\eqref{eq:intro-many-body-Dyson} implicitly contains an improved version of~\eqref{eq:intro-Dyson}, where only the high-momentum part $|p| \ge \eps^{-1}$ of the kinetic energy is needed to replace $V_N$ by $N^{-2}U_R$, and the low-momentum part $|p|\le \eps^{-1}$ is kept in order to recover the full nonlinear functional in~\eqref{eq:NLS}. The same idea of saving the the low-momentum part has been also used in the two-body interaction case~\cite{LieSei-06,NamRouSei-16}.

\medskip
{\bf Removing the cut\nobreakdash-off and estimating four-body error terms.}
	The cut\nobreakdash-off $\theta_{2R}(x)=\1_{\{|x| > 2R\}}$ appears in~\eqref{eq:intro-many-body-Dyson} due to the fact that we exclude the event of having four particles within a distance $O(R)$. This is a disadvantage of the use of the Dyson lemma and the four\nobreakdash-body problem here is similar to the three\nobreakdash-body one in the two\nobreakdash-body interaction case~\cite{LieSei-06,NamRouSei-16}. The standard way to remove the cut\nobreakdash-off $\theta_{2R}$ is to use Bernoulli's inequality
\begin{align*}
	\prod_{\ell \neq i,j,k} \theta_{2R}\left(\frac{x_i+x_j+x_k}{3}-x_\ell \right) &= \prod_{\ell \neq i,j,k} \left( 1- \chi_{2R}\left(\frac{x_i+x_j+x_k}{3}-x_\ell \right) \right) \\
		&\ge 1- \sum_{\ell \neq i,j,k} \chi_{2R}\left(\frac{x_i+x_j+x_k}{3}-x_\ell \right),
\end{align*}
where $\chi_R(x):=\1_{\{|x| \leq R\}} = 1 - \theta_R(x)$. Consequently, the interaction in~\eqref{eq:intro-many-body-Dyson} can be bounded from below as
\begin{multline} \label{eq:intro-many-body-Dyson-2}
	\frac{1}{6 N^2} \sum_{\substack{ 1\le i,j,k \leq N \\ i\neq j \neq k \neq i } } {U_R}(x_i-x_j, x_i-x_k) \prod_{\ell \neq i,j,k} \theta_{2R}\left(\frac{x_i+x_j+x_k}{3}-x_\ell \right) \\
	\ge \begin{multlined}[t]
		\frac{1}{6 N^2} \sum_{\substack{ 1\le i,j,k \leq N \\ i\neq j \neq k \neq i } } {U_R}(x_i-x_j, x_i-x_k) \\
		- \frac{C}{N^2R^6} \sum_{\substack{ 1\le i,j,k \leq N \\ i\neq j \neq k \neq i } } \sum_{\ell \neq i,j,k} \chi_{4R}(x_i-x_j) \chi_{4R}(x_i-x_k) \chi_{4R}(x_i-x_\ell)\,.
	\end{multlined}
\end{multline}
Due to the energy-critical nature of the problem, we are unable to control the four-body error term using the second-moment argument as in~\cite{NamRouSei-16} (and its variants, e.g.~a third-moment argument, seem also insufficient). Nevertheless, following the approach in~\cite{LieSei-06}, we can show (see Lemma~\ref{lem:4_body_collision}) that, up to a replacement of $N$ by $M\approx N$ if necessary, the zero-temperature limit of the bosonic Gibbs state $\Gamma_{N}$ of $H_N$ satisfies the four\nobreakdash-body collision estimate
\[
	\Tr \prod_{\ell=2}^{4} \chi_R(x_1-x_\ell) \Gamma_{N} \leq C R^9\,.
\]
Therefore, the expectation against $\Gamma_{N}$ of the error term in~\eqref{eq:intro-many-body-Dyson-2} is bounded by $CN^2 R^3$, which coincides with the error in~\eqref{eq:intro-many-body-Dyson}. In summary, from~\eqref{eq:intro-many-body-Dyson} and~\eqref{eq:intro-many-body-Dyson-2} we deduce that
\begin{equation} \label{eq:intro-many-body-Dyson-3}
	\frac{E_N}{N} \ge \frac{1-\eps}{N} \Tr \bigg( \sum_{i=1}^N \widetilde{h}_i + \frac{1}{6 N^2} \sum_{\substack{ 1\le i,j,k \leq N \\ i\neq j \neq k \neq i } } {U_R}(x_i-x_j, x_i-x_k) \bigg) \Gamma_{N} - C_\eps N R^3\,,
\end{equation}
where $\Gamma_{N}$ is the zero-temperature limit of the bosonic Gibbs state of $H_N$. In order to keep the error of order $o(1)$, we need to take
\begin{equation} \label{eq:cond-R}
	N^{-1/3} \gg R \gg N^{-1/2}\,.
\end{equation}
These constraints on $R$ are optimal in order to make the Dyson lemma useful: the condition $R\gg N^{-1/2}$ ensures that we replace $V_N$ by a less singular potential, while the condition $R \ll N^{-1/3}$ keeps us in the dilute regime where there are essentially no four\nobreakdash-body collisions. Under these conditions, applying the Dyson lemma does not change the energy to the leading order. 

\medskip
{\bf Step 2: Mean-field approximation.}
	So far we follow closely the existing analysis in the two\nobreakdash-body interaction case~\cite{LieSei-06,NamRouSei-16}. Let us now explain a crucial new difficulty in the three\nobreakdash-body interaction case, which implicitly relies on the fact that we are dealing with an energy-critical problem here. In the two\nobreakdash-body interaction case, as soon as we arrive at an analogue of~\eqref{eq:intro-many-body-Dyson-3}, the right-hand side can be treated using now standard mean-field techniques, e.g.~using the coherent state method as in~\cite{LieSei-06} or using the quantum de Finetti theorem as in~\cite{NamRouSei-16}. A key ingredient needed in~\cite{LieSei-06,NamRouSei-16} is the two\nobreakdash-body inequality
\begin{equation} \label{eq:intro-2-body-operator-inequality}
	|W(x-y)| \le C \norm{W}_{L^1(\R^3)} (1-\Delta_x) (1-\Delta_y) \quad \text{on } L^2\!\left( (\R^3)^2 \right).
\end{equation}
Together with the so-called ``second moment estimates", this inequality allows to control the interaction potential efficiently by the kinetic operator. The bound~\eqref{eq:intro-2-body-operator-inequality} also plays an essential role in the study of the Gross--Pitaevskii dynamics in~\cite{ErdYau-01}. In fact, as proved in~\cite{NamRouSei-16}, the refinement
\begin{equation} \label{eq:intro-2-body-operator-inequality-imp}
	|W(x-y)| \le C_\eta \norm{W}_{L^1(\R^3)} (1-\Delta_x)^{3/4+\eta} (1-\Delta_y)^{3/4+\eta} \quad \text{on } L^2\!\left( (\R^3)^2 \right)
\end{equation}
of~\eqref{eq:intro-2-body-operator-inequality} holds for every $\eta>0$ and is useful in combination with the so-called second moment estimates (which ones are out of reach for the three\nobreakdash-body case in the regime we consider). Roughly speaking,~\eqref{eq:intro-2-body-operator-inequality-imp} can be interpreted as a variant of the Sobolev embedding theorem $L^\infty(\R^3)\subset H^{3/2+2\eta}(\R^3)$, where the total $(3/2+2\eta)$ derivatives on $x-y$ are divided equally between the variables $x$ and $y$. In the three\nobreakdash-body interaction case, the analogue of~\eqref{eq:intro-2-body-operator-inequality-imp} is the following operator inequality on $L^2((\R^3)^3)$:
\begin{equation} \label{eq:intro-3-body-operator-inequality-imp}
	|W(x-y,x-z)| \le C_\eta \norm{W}_{L^1(\R^6)} (1-\Delta_x)^{1+\eta} (1-\Delta_y)^{1+\eta} (1-\Delta_z)^{1+\eta}
\end{equation}
for every $\eta>0$, which should be compared with the Sobolev embedding theorem $L^\infty(\R^6)\subset H^{3+3\eta}(\R^6)$. In particular, there is no extension of~\eqref{eq:intro-2-body-operator-inequality} to the three\nobreakdash-body case, namely one cannot take $\eta=0$ in~\eqref{eq:intro-3-body-operator-inequality-imp}. For that reason, we are not able to apply directly the mean-field techniques as in ~\cite{LieSei-06,NamRouSei-16} in order to handle the right-hand side of~\eqref{eq:intro-many-body-Dyson-3}. More precisely, one could try to replace~\eqref{eq:intro-3-body-operator-inequality-imp} by the bound
\[
	|W(x-y,x-z)| \le C_p \norm{W}_{L^p(\R^6)} (1-\Delta_x) (1-\Delta_y) (1-\Delta_z)\,,
\]
for $p>1$, but in our application $\norm{U_R}_{L^p(\R^6)} \sim R^{-6 (p-1)}$ is then too large due to the constraint $R\ll N^{-1/3}$ in~\eqref{eq:cond-R}. Therefore, to overcome this difficulty, we have to relax the condition $R\ll N^{-1/3}$ before applying the mean-field techniques, and this requires new ideas. 

\medskip
{\bf Repeated use of the Dyson lemma.} We will replace $U_R$ by softer potentials by applying the Dyson lemma again, in the spirit of~\eqref{eq:Dyson-easy-subcritical}  where $N^{-2}U_R$ plays the role of $U_{N,\beta}$ and $R$ plays the role of $N^{-\beta}$. As we already mentioned before, in our second use of the Dyson lemma, we only use a very small fraction of the kinetic energy. More precisely, for every $0<R\ll R_1\ll 1$  we can find $0< \nu \ll 1$ such that 
\begin{equation}\label{eq:intro-Dyson-2a}
	- \nu \Delta_{\cM}  + N^{-2} U_R (\bx) \geq N^{-2} {U_{R_1}} (\bx) (1+ o(1))\,.
\end{equation}
However, the latter bound is not very helpful since if we use it to deal with the $N$-body Hamiltonian, then we have to impose the additional condition $R_1\ll N^{-1/3}$, which is similar to~\eqref{eq:cond-R},  in order to control the corresponding four-body error terms. Therefore, to proceed further we have to introduce a new technique to  relax the condition on $R_1$. 

\medskip
{\bf Adjustment of the number of particles.} To relax the condition on $R_1$, we will replace the $N$-body Hamiltonian by a $N_1$-body Hamiltonian with $N_1\ll N$. Heuristically, if we can replace $N$ by $N_1\ll N$, then the constraint on $R_1$ becomes $R_1\ll N_1^{-1/3}$ which is much better than the previous condition $R_1\ll N^{-1/3}$. This can be done rigorously using the bosonic symmetry of $\Psi_{N}$, namely we can rewrite the main term on the right-hand side of~\eqref{eq:intro-many-body-Dyson-3} as
\begin{multline}\label{eq:intro-soft-potential-main-2a}
	\frac{1}{N}\Tr \bigg( \sum_{i=1}^N \widetilde{h}_i + \frac{1}{6 N^2} \sum_{\substack{ 1\le i,j,k \leq N \\ i\neq j \neq k \neq i } } {U_R}(x_i-x_j, x_i-x_k) \bigg) \Gamma_{N}  \\
	\approx \frac{1}{N_1} \Tr \bigg( \sum_{i=1}^{N_1} \widetilde{h}_i + \frac{1}{6 N_1^2} \sum_{\substack{ 1\le i,j,k \leq N_1 \\ i\neq j \neq k \neq i } } {U_R}(x_i-x_j, x_i-x_k) \bigg)\Gamma_{N} 
\end{multline}
for $N \gg N_1 \gg 1$. For the $N_1$-body Hamiltonian, the following bound 
\begin{equation*}
	- \nu \Delta_\cM + N_1^{-2} U_R (\bx) \geq N_1^{-2} {U_{R_1}} (\bx) (1+ o(1)), 
\end{equation*}
can be used instead of~\eqref{eq:intro-Dyson-2a}, and we can replace  $U_R$ on the right-hand side of~\eqref{eq:intro-soft-potential-main-2a} by $U_{R_1}$ for a lower bound as soon as
\begin{equation*}
	N_1^{-1/3} \gg R_1 \gg R \gg N_1^{-1/2}\,.
\end{equation*}
The latter constraints  are comparable to~\eqref{eq:cond-R}. By choosing $N_1$ suitably, we can fulfill these conditions provided that  
\begin{equation*}
	R^{2/3} \gg R_1 \gg R\,,
\end{equation*}
and obtain 
\begin{multline*} 
	\frac{1}{N}\Tr \bigg( \sum_{i=1}^N \widetilde{h}_i + \frac{1}{6 N^2} \sum_{\substack{ 1\le i,j,k \leq N \\ i\neq j \neq k \neq i } } {U_R}(x_i-x_j, x_i-x_k) \bigg) \Gamma_{N}  \\
	\ge \frac{1-\eps}{N} \Tr \bigg( \sum_{i=1}^{N} \widetilde{h}_i + \frac{1}{6 N^2} \sum_{\substack{ 1\le i,j,k \leq N \\ i\neq j \neq k \neq i } } {U_{R_1}}(x_i-x_j, x_i-x_k) \bigg)\Gamma_{N}  + o(1)\,.
\end{multline*}
Here, we already used the bosonic symmetry again to replace $N_1$ by $N$ on the right-hand side. Repeating this procedure, we can replace $U_R$ by $U_{R_\ell}$ for every fixed $\ell\in \N$ provided that
\[
	R_{\ell-1}^{2/3} \gg R_{\ell} \gg R_{\ell-1}\,.
\]
Thus, for every $\eta>0$ small arbitrarily, we can choose  $R_\ell= N^{-\eta}$ with $\ell=\ell(\eta)$ sufficiently large.
Putting it all together, we deduce from~\eqref{eq:intro-many-body-Dyson-3} that
\begin{multline*}
	\frac{E_N}{N} \ge \frac{(1-\eps)^{\ell+1}}{N} \Tr \bigg( \sum_{i=1}^N \widetilde{h}_i + \frac{1}{6 N^2} \sum_{\substack{ 1\le i,j,k \leq N \\ i\neq j \neq k \neq i } } {U_{R_\ell}}(x_i-x_j, x_i-x_k) \bigg) \Gamma_{N} + o(1)
\end{multline*}
with a soft potential $U_{R_\ell}$ that can be handled by the mean-field techniques from~\cite{LewNamRou-16}. The details will be discussed in Section~\ref{sec:lower}. 
Taking $\eps\to 0^+$ at the end, we obtain the desired lower bound
\[
	\frac{E_N}{N} \ge \inf_{\norm{u}_{L^2(\R^3)=1}} \left( \pscal{ u, h u }_{L^2(\mathbb{R}^{3})} + \frac{b_{\cM}(V)}{6} \int_{\R^3} |u(x)|^6 \d x \right) + o(1)\,.
\]
This completes our sketch of the proof of Theorem~\ref{thm:main1}.

\medskip
\noindent {\em Proof of the BEC.}
	Theorem~\ref{thm:main2} follows from a Hellmann--Feynman argument as in~\cite{NamRouSei-16}, where we will derive the energy convergence of Theorem~\ref{thm:main1} for a perturbed problem. This will be discussed in Section~\ref{sec:state}. 
	
\medskip
\noindent {\bf Organization of the paper.}
	In Section~\ref{sec:scattering-energy}, we discuss basic facts on the scattering energy in~\eqref{eq:intro-def-scat-energy}. Then, we derive several versions of the Dyson lemma in Section~\ref{sec:Dyson}, which will be used to replace the potential $V_N$ by softer ones in Section~\ref{sec:reduction_soft}. In Section~\ref{sec:lower}, we conclude the energy lower bound in Theorem~\ref{thm:main1}. The matching energy upper bound is proved in Section~\ref{sec:upper}. Finally, the convergence of states in Theorem~\ref{thm:main2} is obtained in Section~\ref{sec:state}.

\medskip
\noindent {\bf Notation.}
	From now on and for shortness, we will denote $\norm{\cdot}_p:=\norm{\cdot}_{L^p(\R^d)}$ when there is no possible confusion on the dimension $d$.

\medskip
\noindent {\bf Acknowledgments.}
	Arnaud Triay thanks Jonas Lampart for helpful discussions. We would like to thank the referees for their useful comments and suggestions, leading to a significant improvement of the paper. We received funding from the Deutsche Forschungsgemeinschaft (DFG, German Research Foundation) under Germany's Excellence Strategy (EXC-2111-390814868).

\section{Scattering energy}\label{sec:scattering-energy}

\subsection{General setting} In this section we discuss the zero-scattering problem of nonnegative potentials which are not necessarily radial. We refer to \cite[Appendix C]{LieSeiSolYng-05} for a related discussion in the case of radial potentials. 

	Let $d\ge 3$ and $0\le v \in L^\infty(\R^d)$ be compactly supported. We define the {\em zero--scattering energy} of $v$ by
\begin{equation} \label{eq:def-scat-energy}
	b(v):= \inf_{\varphi \in \dot H^1(\R^d)} \int_{\R^d} \left( 2|\nabla \varphi( \bx)|^2 + v(\bx) |1-\varphi(\bx)|^2 \right) \d \bx\,.
\end{equation}

	Here, $\dot H^1(\R^d)$ is the space of functions $g:\R^d \to \C$ vanishing at infinity with $|\nabla g|\in L^2(\R^d)$, denoted by $D^1(\R^d)$ in Lieb--Loss~\cite[Section 8.3]{LieLos-01}.
\begin{theorem}[General scattering solution]\label{thm:scattering}
	Let $d\ge 3$ and $0\le v \in L^\infty(\R^d)$ be compactly supported. Then, the variational problem~\eqref{eq:def-scat-energy} has a unique minimizer $\omega=(-2\Delta +v)^{-1}v$.
	It solves, for almost every $\bx\in \R^d$, the scattering equation
	\[
		-2\Delta \omega (\bx) + v(\bx) (\omega(\bx)-1) =0
	\]
	and satisfies, for all $\bx\in \R^d$, the pointwise estimates
	\begin{equation}\label{eq:bounds_on_w_general}
		0\le \omega(\bx)<1\,, \quad \omega(\bx) \le \frac{C_{d,v}}{|\bx|^{d-2}+1}\,, \quad \text{and} \quad |\nabla \omega(\bx)| \leq \frac{C_{d,v}}{|\bx|^{d-1}+1}\,.
	\end{equation}
	Moreover, the scattering energy satisfies
	\begin{equation} \label{eq:bounds_on_w_general_2}
		b(v)= \int_{\R^d} v(\bx) (1-\omega(\bx)) \d \bx \quad \text{and} \quad 0 \le \int_{\R^{d}} v - b( v) \le C_d \norm{v}_{\frac{2d}{d+2}}^2\,.
	\end{equation}
\end{theorem}
\begin{proof}
	Let $\{\varphi_n\}_{n=1}^\infty\subset \dot{H}^1(\R^d)$ be a minimizing sequence for the functional
	\[
		\mathscr{E}[\varphi] = \int_{\R^d}\left( 2 \left|\nabla\varphi(\bx)\right|^2 + v(\bx)\left|1-\varphi(\bx)\right|^2\right)\d\bx\,.
	\]
	We can assume that $\varphi_n$ is real-valued, since we can ignore the imaginary part of $\varphi_n$ without increasing $ \mathscr{E}[\varphi_n]$, and that $0\le \varphi_n\le 1$, since we can replace $\varphi_n$ by $\min(\max(\varphi_n,0),1)$ without increasing $\mathscr{E}[\varphi_n]$. Given that $ \mathscr{E}[\varphi_n]$ is bounded and $v$ nonnegative, $\varphi_n$ is bounded in $\dot H^1(\R^d)$ and $\sqrt{v}(1-\varphi_n)$ is bounded in $L^2(\R^d)$. By the Banach--Alaoglu theorem, we can assume up to a subsequence that $ \varphi_n \wto \omega$ weakly in ${\dot H}^1(\R^d)$ and $\sqrt{v}(1-\varphi_n) \wto \sqrt{v}(1-\omega)$ weakly in $L^2(\R^d)$. By Fatou's lemma, we conclude that $\omega$ is a minimizer. The minimizer $\omega$ is unique since the functional $\varphi \mapsto \mathscr{E}[\varphi]$ is strictly convex.
	
	The above proof also gives $0\le \omega \le 1$. Moreover, $\mathscr{E}[\omega] \le \mathscr{E}[\omega+ t \varphi]$ for $t\ge 0$ and any function $\varphi \in C_c^\infty(\R^d)$. Hence,
	\[
		0 \le \frac{d}{dt}_{|t=0} \mathscr{E}[\omega+ t \varphi]= 2\int_{\R^d} \big( 2\nabla \omega \cdot \nabla \varphi + v (\omega-1) \varphi \big).
	\]
	Thus\,
	\begin{equation} \label{eq:equation-scattering}
		-2\Delta \omega+ v (\omega-1) =0
	\end{equation}
	in the distributional sense. Since $0\le v (1-\omega)\le v \in L^1(\R^d)\cap L^\infty(\R^d)$, we get $\omega\in H^2(\R^{d})\cap C^1(\R^d)$ by the standard elliptic regularity~\cite[Theorem 10.2]{LieLos-01}. Thus, equation~\eqref{eq:equation-scattering} holds in the pointwise sense (almost everywhere).
	
	The scattering equation can be written as
	\begin{equation} \label{eq:om-con}
		\omega(\bx)= (-2\Delta)^{-1} [v (1-\omega)] (\bx)=\frac{1}{2 |\bS^{d-1}| (d-2)} \int_{\R^d} \frac{v ({\bf y}) (1-\omega({\bf y})) \d {\bf y}}{|\bx-{\bf y}|^{d-2}}
	\end{equation}
	where $|\bS^{d-1}|$ is the surface area of the $(d-1)$ dimensional sphere $\bS^{d-1}$. Since $v (1-\omega)\in L^1(\R^d)$ and it has compact support, we deduce from~\eqref{eq:om-con} that $\omega(\bx)$ decays as $\cO(|\bx|^{2-d})$ as $|\bx|\to \infty$. Since $0\le \omega\le 1$ everywhere, we conclude that
	\[
		\omega(\bx) \le \frac{C_{d,v}}{|\bx|^{d-2}+1} \quad \text{for all } \bx\in \R^d\,.
	\]
	Moreover, note that~\eqref{eq:equation-scattering} is equivalent to $-2\Delta f + v f =0$ pointwise with $f=1-w$. Therefore, from $v \in L^\infty(\mathbb{R}^{d})$, $f\ge 0$ everywhere, and $f$ is not identically zero (since it does not vanish at infinity), we find that $f>0$ everywhere by~\cite[Theorem 9.10]{LieLos-01}. Thus, $\omega<1$ everywhere.

	From~\eqref{eq:om-con}, we also obtain
	\begin{equation} \label{eq:D-om-con}
		\nabla \omega(\bx)= - \frac{1}{2 |\bS^{d-1}|} \int_{\R^d} \frac{v ({\bf y}) (1-\omega({\bf y})) (\bx-{\bf y}) \d {\bf y}}{|\bx-{\bf y}|^{d}} \,.
	\end{equation}
	This implies that $|\nabla \omega(\bx)|$ decays as $\cO(|\bx|^{1-d})$ as $|\bx|\to \infty$. Moreover, since $\omega\in C^1(\R^d)$, we conclude that
	\[
		|\nabla \omega(\bx)| \le \frac{C_{d,v}}{|\bx|^{d-1}+1} \quad \text{for all } \bx\in \R^d\,.
	\]
	
	Finally, since $\omega$ is a minimizer for~\eqref{eq:def-scat-energy} we have
	\[
		b(v)= \int_{\R^d} \left( 2 |\nabla \omega( \bx)|^2 + v(\bx) |1-\omega(\bx)|^2 \right) \d \bx \,.
	\]
	On the other hand, from the scattering equation we have
	\[
		\int_{\R^d} \left( 2|\nabla \omega( \bx)|^2 + v(\bx) (\omega(\bx)-1) \omega(\bx) \right) \d \bx = 0\,.
	\]
	Thus, we can rewrite
	\[
		b(v)= \int_{\R^d} \left( v(\bx) |1-\omega(\bx)|^2 - v(\bx) (\omega(\bx)-1) \omega(\bx) \right) \d \bx = \int_{\R^d} v(\bx) (1-\omega(\bx)) \d \bx\,.
	\]
	Using the scattering equation $\omega = (-2\Delta + v)^{-1}v$ and the Hardy--Littlewood--Sobolev inequality, we can estimate
	\begin{align*}
		0 \le \int_{\R^{d}} v - b(v) = \int_{\R^{d}} v \omega &= \int_{\R^{d}} v \left(- 2\Delta + v\right)^{-1} v \\
			&\le \frac{1}{2} \int_{\R^{d}} v \left(- \Delta \right)^{-1} v\le C_d \norm{v}_{\frac{2d}{d+2}}^2\,. \qedhere
	\end{align*}
\end{proof}

\begin{remark}\label{rem:Born}
	It is well-known (see for example~\cite{BriSol-20,Rougerie-20b}) that, by repeatedly using the scattering equation $\omega = (-2\Delta + v)^{-1}v$ and the resolvent formula, we can write the scattering energy as a Born series expansion
	\[
		b(v)= \int_{\R^d}v - \int_{\R^d} v (-2\Delta + v)^{-1}v = \int_{\R^d}v - \int_{\R^d} v (-2\Delta)^{-1}v + \ldots \,.
	\]
	If $v\ge 0$ and $v\not \equiv 0$, then $b(v)<\int_{\R^d}v$ since $\left(- 2\Delta + v\right)^{-1}>0$ on $L^2(\R^d)$.
\end{remark}

\subsection{Three-body symmetry}
	Let $V:\R^6 \to \R_+$ satisfy the three\nobreakdash-body symmetry~\eqref{eq:sym}. One can check that~\eqref{eq:sym} is equivalent to
\[
	V(x,y) = V(y,x) = V(x-y,-y) \quad \text{for all } (x,y) \in \R^3\times\R^3\,.
\]
Put differently, $V= V( g \cdot)$ when $g: \R^3\times \R^3 \to \R^3 \times \R^3$ is equal to
\[
	S :=
	\begin{pmatrix}
		0 & 1 \\
		1 & 0
	\end{pmatrix}
	\quad \text{or} \quad
	A :=
	\begin{pmatrix}
		1 & -1 \\
		0 & -1
	\end{pmatrix}.
\]
Note that both $S$ and $A$ are symmetries ($S^2 = A^2 = 1$) and that $SAS = ASA$. Thus, the group generated by $A$ and $S$ is finite and is given by
\begin{equation} \label{eq:G}
	\cG = \left\{ I, S, A, AS, SA, ASA\right\}.
\end{equation}
In summary, the symmetry~\eqref{eq:sym} is equivalent to the fact that $V=V(g \cdot)$ for all $g\in \cG$ ---namely, $V$ is invariant under the action of $\cG$.

	Next, let us consider the scattering problem associated to the three\nobreakdash-body interaction potential $V(x-y,x-z)$. Consider the operator
\[
	-\Delta_{x_1} - \Delta_{x_2} - \Delta_{x_3} + V(x_1-x_2,x_1-x_3) \quad \text{on }L^2\!\left( (\R^3)^3 \right).
\]
After removing the center of mass, we are left with the two\nobreakdash-body operator
\[
	-2\Delta_{\cM} + V(x,y)\quad \text{on }L^2\!\left( (\R^3)^2 \right),
\]
where $-\Delta_{\cM}=|\cM \nabla_{\R^6}|^2$ with the matrix $\cM: \R^3\times \R^3 \to \R^3\times \R^3$ given by
\begin{equation}\label{eq:M}
	\cM:= \left( \frac{1}{2}
	\begin{pmatrix}
		2 & 1 \\
		1 & 2
	\end{pmatrix}
	\right)^{1/2} = \frac{1}{2\sqrt{2}}
	\begin{pmatrix}
		\sqrt{3}+1 & \sqrt{3}-1 \\
		\sqrt{3}-1 & \sqrt{3}+1
	\end{pmatrix}.
\end{equation}
			
	We define the {\em modified scattering energy}
\begin{equation} \label{eq:def-scat-energy-M}
	b_{\cM}(V) : = \inf_{\varphi \in \dot H^1(\R^d)} \int_{\R^d} \left( 2| \cM \nabla \varphi( \bx)|^2 + V(\bx) |1-\varphi(\bx)|^2 \right) \d \bx\,.
\end{equation}
As we will see, by a change of variables, the results from the previous section on the standard scattering energy $b(V)$, defined in~\eqref{eq:def-scat-energy}, can be used to understand $b_\cM(V)$, defined in~\eqref{eq:def-scat-energy-M}. To be precise, from Theorem~\ref{thm:scattering} we have the following.
\begin{theorem}[Modified scattering solution] \label{thm:scattering-M}
	Let $0\le V \in L^\infty(\R^6)$ be compactly supported and satisfy the symmetry~\eqref{eq:sym}. Then, the variational problem~\eqref{eq:def-scat-energy-M} has a unique minimizer $\omega=(-2\Delta_\cM +V)^{-1}V$. The function $\omega$ satisfies the symmetry~\eqref{eq:sym}, it solves, for almost every $\bx\in \R^6$, the modified scattering equation
	\[
		-2\Delta_\cM \omega (\bx) + V(\bx) (\omega(\bx)-1) = 0
	\]
	---which is equivalent to solve, for almost every $x,y,z\in \R^3$, the equation
	\[
		(-\Delta_x - \Delta_y - \Delta_z) \omega(x-y,x-z) + (V (\omega-1)) (x-y,x-z) = 0\text{---},
	\]
	and it satisfies, for all $\bx\in \R^6$, the pointwise estimates
	\[
		0\le \omega(\bx)<1\,, \quad \omega(\bx) \le \frac{C_{V}}{|\bx|^{4}+1}\,, \quad \text{and} \quad |\nabla \omega(\bx)| \leq \frac{C_{V}}{|\bx|^{5}+1}\,.
	\]
	The modified scattering energy satisfies
	\begin{equation}\label{eq:born-M}
		b_\cM(V)= \int_{\R^6} V(\bx) (1-\omega(\bx)) \d \bx \quad \text{and} \quad 0 \le \int_{\R^{6}} V - b_\cM( V) \leq C \norm{V}_{\frac{3}{2}}^2\,.
	\end{equation}
	Moreover, $b_\cM(V)=b(V( \cM\cdot )) \det \cM$ for $b$ defined in~\eqref{eq:def-scat-energy}. Here, $\det \cM = 3\sqrt 3/8$.
\end{theorem}
\begin{proof}
	Note that for every $\varphi\in \dot {H}^1(\R^6)$ we have
	\begin{align} \label{eq:M-change-var}
		\int_{\R^6} 2| \cM \nabla \varphi|^2 + V |1-\varphi|^2 &=\int_{\R^6} \left( 2| \cM (\nabla \varphi) (\cM \cdot)|^2 + V(\cM \cdot) |1-\varphi(\cM \cdot)|^2 \right) \det \cM \nn \\
			&= \det \cM \int_{\R^6} 2| \nabla (\varphi(\cM \cdot))|^2 + V(\cM \cdot) |1-\varphi(\cM \cdot)|^2\,.
	\end{align}
	Moreover, it is obvious that $\varphi \in \dot {H}^1 (\R^6)$ if and only if $\varphi (\cM \cdot) \in \dot {H}^1 (\R^6)$. Therefore,
	\begin{align*}
		b_{\cM}(V) &= \inf_{\varphi \in \dot H^1(\R^d)} \int_{\R^d} 2|\nabla_{\cM} \varphi|^2 + V |1-\varphi|^2 \\
			&= \det \cM \inf_{\varphi \in \dot H^1(\R^d)} \int_{\R^6} 2|\nabla (\varphi(\cM \cdot))|^2 + V(\cM \cdot) |1-\varphi(\cM \cdot)|^2 \\
			&= \det \cM \inf_{\varphi \in \dot H^1(\R^d)} \int_{\R^6} 2|\nabla \varphi|^2 + V(\cM \cdot) |1-\varphi|^2 = b(V( \cM\cdot )) \det \cM \,.
	\end{align*}
	
	Thanks to~\eqref{eq:M-change-var}, it is straightforward that the minimizer of $b_{\cM}(V)$ in~\eqref{eq:def-scat-energy-M} is $\omega= \widetilde \omega (\cM^{-1} \cdot)$ with $\widetilde \omega$ the unique minimizer of $b(V(\cM \cdot))$ defined in~\eqref{eq:def-scat-energy}. Thus, most of the statements in Theorem~\ref{thm:scattering-M} follow from Theorem~\ref{thm:scattering}. From the equation
	\[
		-2\Delta_\cM \omega (\bx) + V(\bx) (\omega(\bx)-1) =0 \quad \text{for a.e.\ } \bx\in \R^6\,,
	\]
	we can also deduce that
	\[
		(-\Delta_x - \Delta_y - \Delta_z) \omega(x-y,x-z) + (V (\omega-1)) (x-y,x-z) =0 \quad \text{for a.e.\ } x,y,z\in \R^3\,,
	\]
	by removing the center of mass similarly to~\eqref{eq:remove-center}.
	
	Finally, let us prove that $\omega$ is invariant under the actions of $\cG$. Since $V$ is invariant under the actions of $\cG$ and $\omega=(-2\Delta_\cM +V)^{-1}V$, it remains to check that $-\Delta_\cM$ is also invariant under the actions of $\cG$. For every $\varphi\in C_c^\infty(\R^6)$ and $g\in \cG$, using
	\begin{equation} \label{eq:M-G}
		g \cM^2 g^\dagger= \cM^2
	\end{equation}
	and $|\det g|=1$, we have
	\begin{align*}
		\pscal{ \varphi (g \cdot), -\Delta_\cM \varphi (g \cdot) }_{L^2(\R^6)} &= \int_{\R^6} | \cM \nabla ( \varphi (g\bx) ) |^2 \d \bx = \int_{\R^6} |\cM g^\dagger ( \nabla \varphi) (g\bx)|^2 \d \bx \\
			&=\int_{\R^6} |\cM ( \nabla \varphi) (\bx)|^2 |\det g| \d \bx = \int_{\R^6} |\cM ( \nabla \varphi) (\bx)|^2 \d \bx \\
			&= \pscal{ \varphi, -\Delta_\cM \varphi }_{L^2(\R^6)}\,.
	\end{align*}
	Thus, $\omega=\omega(g \cdot )$ for all $g\in\cG$ and it therefore satisfies the symmetry~\eqref{eq:sym}.
\end{proof}

\section{Dyson lemmas} \label{sec:Dyson}

\subsection{Dyson lemma for non-radial potentials}

\begin{theorem}[Dyson lemma for non-radial potentials]\label{dyson_lemma_Rd_nonradial}
	Let $d\ge 3$, $0\le v \in L^\infty(\R^d)$ with $\supp\, v \subset \{ |\bx| \le R_0\}$, and $0\leq U \in C(\R^d)$ be radial with $\int_{\R^d} U = 1$ and $\supp U \subset \{ R_1\le |\bx| \le R_2\}$. Then, we have the operator inequality
	\[
		- 2 \nabla_{\bx} \1_{\{|\bx| \le R_2\}} \nabla_{\bx}+ v(\bx) \geq b(v) \left( 1- \frac{C_d R_0}{R_1}\right) U (\bx)\quad \text{on }L^2(\R^d)
	\]
	with a constant $C_d>0$ depending only on the dimension $d$.
\end{theorem}
We will later use Theorem~\ref{dyson_lemma_Rd_nonradial} for $d=6$. Note that we need the characteristic function $\1_{\{|\bx| \le R_2\}}$ since we will apply the Dyson lemma to a specific region of the configuration space where three particles are in the same neighborhood.
\begin{proof}
	If $R_1<2R_0$, then we can take $C_d=2$ and the desired inequality holds trivially since the right-hand side is negative. Thus, it remains to consider the case $R_1\ge 2R_0$.
	
	Let $\omega$ be the scattering solution to $b(v)$ as in Theorem~\ref{thm:scattering}. Then, the function $f:=1-\omega$ satisfies
	\[
		0< f \le 1\,, \quad -2\Delta f + v f =0\,, \quad \text{and} \quad \lim_{|\bx|\to \infty} f(\bx) =1\,.
	\]
	We now take an arbitrary function $\varphi \in C_c^\infty(\R^d)$ and denote $ \eta = \varphi/f$. For every $R \in [R_1,R_2]$, integrating by parts and using the scattering equation, we have
	\begin{multline*}
		\int_{B(0,R)} 2 |\nabla \varphi|^2 + v |\varphi|^2 = \int_{B(0,R)} 2 f^2 |\nabla \eta |^2 + 2|\eta|^2 |\nabla f |^2 + 2 f (\nabla f) \nabla (|\eta|^2) + v f^2 |\eta|^2 \\
		= \int_{B(0,R)} 2 f^2 |\nabla \eta |^2 + \int_{B(0,R)} |\eta|^2 f \left(-2\Delta f + vf \right) + \int_{\partial B(0,R)} 2 |\eta|^2 f (\nabla f) \cdot \vec n \\
		= \int_{B(0,R)} 2 f^2 |\nabla \eta |^2 + \int_{\partial B(0,R)} 2 |\varphi|^2 \frac{(\nabla f) \cdot \vec n}{f}\,,
	\end{multline*}
	where $B(0,R) = \{|x|\leq R\}$ and $\vec n_{\bx}= \bx/|\bx|$ is the outward unit normal vector on the sphere $\partial B(0,R)$. Therefore, for every $R \in [R_1,R_2]$, we can bound
	\begin{equation} \label{eq:Dyson-1}
		\int_{B(0,R_2)} 2|\nabla \varphi|^2 + v |\varphi|^2 \ge \int_{B(0,R)} 2 |\nabla \varphi|^2 + v |\varphi|^2 \ge 2 \int_{\partial B(0,R)} |\varphi|^2 \frac{(\nabla f) \cdot \vec n}{f}\,.
	\end{equation}
	
	Let us now compute $((\nabla f) \cdot \vec n)/f$ on the sphere $\partial B(0,R)$.
	Recall from~\eqref{eq:D-om-con} that
	\[
		\nabla f (\bx) \cdot \vec n_{\bx}= \frac{1}{2|\bS^{d-1}|} \int_{\R^d} v ({\bf y}) f({\bf y}) \frac{ (\bx-{\bf y}) }{|\bx-{\bf y}|^{d}} \cdot \vec n_{\bx} \d {\bf y}\,.
	\]
	For every $|\bx| \ge 2R_0$ and $|{\bf y}|\le R_0$, a Taylor expansion gives
	\[
		\left| \frac{ \bx-{\bf y}}{|\bx-{\bf y}|^{d}} - \frac{ \bx}{|\bx|^{d}} \right| \le C_d \frac{R_0}{|\bx|^{d}}\,, \quad \text{hence} \quad \frac{ \bx-{\bf y}}{|\bx-{\bf y}|^{d}} \cdot \vec n_{\bx} \ge \frac{1}{|\bx|^{d-1}} \left( 1 - C_d \frac{R_0}{|\bx|} \right),
	\]
	where the triangle inequality was used to obtain the second estimate.
	Since $\supp v \subset B(0,R_0)$ and $b(v)= \int_{\R^d} vf$, we have for all $|\bx| \ge 2R_0$ that
	\begin{equation} \label{eq:Dyson-1b}
		\nabla f(\bx) \cdot \vec n_{\bx} \ge \frac{1}{2 |\bS^{d-1}|} \int_{\R^d} \frac{v({\bf y}) f({\bf y})}{|\bx|^{d-1}} \left( 1 - C_d \frac{R_0}{|\bx|} \right) \d {\bf y} = \frac{b(v)}{ 2 |\bS^{d-1}| |\bx|^{d-1}} \left( 1- \frac{C_d R_0}{|\bx|}\right).
	\end{equation}
	Consequently, on one hand, for every $R \in [R_1,R_2]$ such that $1- C_d R_0/R \ge 0$, inserting~\eqref{eq:Dyson-1b} in~\eqref{eq:Dyson-1} and using $1/f \ge 1$, we get
	\begin{equation} \label{eq:Dyson-1c}
		\int_{B(0,R_2)} 2 |\nabla \varphi|^2 + v |\varphi|^2 \ge \frac{b(v)}{ |\bS^{d-1}| R^{d-1}} \left( 1- \frac{C_d R_0}{R}\right) \int_{\partial B(0,R)} |\varphi|^2\,.
	\end{equation}
	On the other hand, if $1- C_d R_0/R < 0$, then~\eqref{eq:Dyson-1c} holds trivially since the left-hand side is always nonnegative. Thus,~\eqref{eq:Dyson-1c} holds for all $R \in [R_1,R_2]$. Integrating both sides of~\eqref{eq:Dyson-1c} against $ |\bS^{d-1}| R^{d-1} U(R)$ with $R\in [R_1,R_2]$ and using $\int_{\R^d} U=1$ for the left-hand side, we conclude that
	\[
		\int_{B(0,R_2)} 2 |\nabla \varphi|^2 + v |\varphi|^2 \ge b(v) \left( 1- \frac{C_d R_0}{R_1}\right) \int_{\R^d} U |\varphi|^2\,.
	\]
	Since the latter bound holds for all $\varphi \in C_c^\infty(\R^d)$, we obtain the desired operator inequality.
\end{proof}

\subsection{Dyson lemma with the three-body symmetry}
	We have the following variant of Theorem~\ref{dyson_lemma_Rd_nonradial} for interaction potentials with the three\nobreakdash-body symmetry.
\begin{theorem}[Dyson lemma with modified scattering energy]\label{dyson_lemma_Rd_nonradial-M}
	Let $0\le V \in L^\infty(\R^6)$ be supported in $B(0,R_0)$ and satisfy the symmetry~\eqref{eq:sym}. Let $0\leq \widetilde U \in C(\R^6)$ be radial with $\int_{\R^6} \widetilde U = 1$ and $\supp \widetilde U \subset \{ R_1\le |\bx| \le R_2\}$. Define
	\begin{equation} \label{eq:Ug}
		{U}:= \frac{1}{6} \sum_{g\in \cG} \widetilde U( \cM^{-1} g \cdot ) \det (\cM^{-1})\,,
	\end{equation}
	where $\cG$ and $\cM$ are given in~\eqref{eq:G} and~\eqref{eq:M}. Then, $0\le U \in C(\R^6)$ satisfies the symmetry~\eqref{eq:sym}, $\int_{\R^6} U =1$, and $\supp U \subset \{ \sqrt{2/3}R_1\le |\bx| \le \sqrt{2}R_2\}$. Moreover, we have the operator inequality
	\[
		-2 \cM \nabla_{\bf x} \1_{\{|{\bf x}| \le \sqrt{2}R_2\}} \cM \nabla_{\bf x} + V(\bx) \geq b_{\cM}(V) \left( 1- \frac{C R_0}{R_1}\right) {U} (\bx)\quad \text{on }L^2\!\left( \R^6 \right).
	\]
	Here, $C>0$ is a universal constant (independent of $V,U,R_0,R_1,R_2$).
\end{theorem}
\begin{proof}
	From the definition~\eqref{eq:Ug}, it is clear that $ U ({\bf x}) = U(g {\bf x})$ for all $g\in \cG$. Thus, $U$ satisfies the symmetry~\eqref{eq:sym}. On the other hand, it is straightforward to diagonalize $\cM$ and find that  its spectrum is equal to $\{\sqrt{1/2},  \sqrt{3/2}\}$, which in particular implies that $\sqrt{1/2} \le \cM \le \sqrt{3/2}$. Combining these bounds with~\eqref{eq:M-G}, we find that 
		\[
		| \cM^{-1} g {\bf x}| = | \cM^{-1} {\bf x}| \in \left[\sqrt{2/3} |{\bf x}|, \sqrt{2} |{\bf x}|\right] \quad \text{for all } {\bf x} \in \R^6\,.
	\]
	Therefore, from the assumption $\supp \widetilde U \subset \{ R_1\le |\bx| \le R_2\}$, we deduce that
	\[
		\supp U \subset \{ \sqrt{2/3} R_1 \le |\bx| \le \sqrt{2} R_2\}\,.
	\]
	Moreover, $\int_{\R^6} U = \int_{\R^6} \widetilde U = 1$ by change of variables and using $|\det g|=1$ for $g\in \cG$.
	
	Next, we prove the operator inequality. We start by applying Theorem~\ref{dyson_lemma_Rd_nonradial} to $V(\cM \cdot)$. Note that $\supp V(\cM \cdot)\subset B(0,\sqrt{3/2}R_0)$, since $\supp V\subset B(0,R_0)$ and $\sqrt{1/2} \le \cM \le \sqrt{3/2}$. Hence, Theorem~\ref{dyson_lemma_Rd_nonradial} gives
	\[
		-2\nabla_{\bx} \1_{\{|\bx| \le R_2\}} \nabla_{\bx} + V(\cM \bx) \ge \left( 1- \frac{C R_0}{R_1}\right) b(V(\cM \cdot)) \widetilde U(\bx)\quad \text{on }L^2\!\left( \R^6 \right).
	\]
	Since $V=V(g\cdot)$ for $g\in \cG$, the change of variable $\bx =\cM^{-1} g {\bf y}$ gives, on $L^2(\R^6)$,
	\[
		-2 \cM \nabla_{\bf y} \1_{\{|\cM^{-1}{\bf y}| \le R_2\}} \cM \nabla_{\bf y} + V({\bf y}) \ge \left( 1- \frac{C R_0}{R_1}\right) b(V(\cM \cdot)) \widetilde U(\cM^{-1} g {\bf y})\,.
	\]
	On the left-hand side, we use $\1_{\{|\cM^{-1}{\bf y}| \le R_2\}} \le \1_{\{|{\bf y}| \le \sqrt{2}R_2\}}$ since $|\cM^{-1}{\bf y}| \ge \sqrt{1/2} |\bf y|$. On the right-hand side, we average over $g\in \cG$ and use $b(V(\cM \cdot))= b_{\cM}(V) \det \cM^{-1}$ (see Theorem~\ref{thm:scattering-M}). The proof is therefore complete as it yields
	\[
		-2 \cM \nabla_{\bf y} \1_{\{|{\bf y}| \le \sqrt{2} R_2\}} \cM \nabla_{\bf y} + V({\bf y}) \ge \left( 1- \frac{C R_0}{R_1}\right) b_{\cM}(V) U({\bf y}) \quad \text{on }L^2\!\left( \R^6 \right). \qedhere
	\]
\end{proof}

\subsection{Many-body Dyson lemma}\label{sec:many-body-Dyson}
	We have the following many\nobreakdash-body version of the Dyson lemma.
\begin{lemma}[Many\nobreakdash-body Dyson lemma] \label{lem:gen_dyson_lemma}
	Let $0\le W \in L^\infty(\R^6)$ be supported in $B(0,R_0)$ and satisfy the symmetry~\eqref{eq:sym}. Let $0\le \widetilde U\in L^\infty(\R^6)$ be radial with $\int_{\R^6} \widetilde U=1$ and $\supp \widetilde U \subset \{1/8 \le |\bx| \le 1/4\}$. Define $U$ as in~\eqref{eq:Ug} and $U_R=R^{-6}U(R^{-1}\cdot)$. Then, for all $s>0$ and $0<\varepsilon<1$, we have
	\begin{align}\label{eq:lem_gen_Dyson}
		\sum_{i=1}^M p_i^2 &\1_{\{|p_i| >s\}} +\frac{1}{6} \sum_{\substack{ 1\le i,j,k \leq M \\ i\neq j \neq k \neq i } } W (x_i-x_j, x_i-x_k) \\
			&\geq \begin{multlined}[t][0.87\textwidth]
				\frac{b_{\cM}(W) (1-\varepsilon)}{6} \left( 1 - \frac{C R_0}{R}\right) \\
				\times \sum_{\substack{ 1\le i,j,k \leq M \\ i\neq j \neq k \neq i } } U(x_i-x_j, x_i-x_k) \prod_{\ell \neq i,j,k} \theta_{2R}\left(\frac{x_i+x_j+x_k}{3}-x_\ell \right) \\
			- C \varepsilon^{-1} s^5 R^3 M^2 \,.
		\end{multlined}\nn
	\end{align}
	Here, $C>0$ is a universal constant (independent of $W,U,M,R,\eps,s$).
\end{lemma}

	Note that in~\eqref{eq:lem_gen_Dyson} we only use the high-momentum part of the kinetic energy on the left-hand side. This is important for our application, since we need the low-momentum part to recover the NLS functional. If we use fully the kinetic energy, then the bound becomes simpler:
\begin{align}\label{eq:lem_gen_Dyson-simpler}
	\sum_{i=1}^M p_i^2 &+ \frac{1}{6} \sum_{\substack{ 1\le i,j,k \leq M \\ i\neq j \neq k \neq i } } W(x_i-x_j, x_i-x_k) \nn \\
	&\geq \begin{multlined}[t]
		\frac{b_{\cM}(W)}{6} \left( 1 - \frac{C R_0}{R}\right) \\
		\times \sum_{\substack{ 1\le i,j,k \leq M \\ i\neq j \neq k \neq i } } U(x_i-x_j, x_i-x_k) \prod_{\ell \neq i,j,k} \theta_{2R}\left(\frac{x_i+x_j+x_k}{3}-x_\ell \right).
	\end{multlined}
\end{align}
The latter bound follows from from~\eqref{eq:lem_gen_Dyson} by taking $s\to 0$ and then $\eps\to 0$.
\begin{proof}
	Denote $\chi_R(x)= \1_{\{|x|\le R\}}=1 - \theta_R(x)$ for $x\in \R^3$. For $(x_1,\ldots,x_M)\in (\R^3)^M$ and $i,j,k\in \{1,2,\ldots,M\}$ with $i\ne j\ne k\ne i$, we denote
	\[
		F_{ijk} := \chi_R(x_i-x_j)\chi_R(x_i-x_k) \chi_R(x_j-x_k) \prod_{\ell \neq i,j,k} \theta_{2R}\left(\frac{x_i+x_j+x_k}{3} -x_\ell\right).
	\]
	Clearly $ F_{ijk} \in \{0,1\}$. Moreover, by the triangle inequality
	\[
		F_{ijk} \le \chi_R(x_i-x_j)\chi_R(x_i-x_k) \prod_{\ell \neq i,j,k} \theta_R \left(x_i-x_\ell\right).
	\]
	Hence, for every $1\le i\le M$, there is at most one pair $j,k$ such that $F_{ijk}=F_{ikj}\ne 0$. Thus, we have the ``no four\nobreakdash-body collision'' bound
	\[
		\sum_{\substack{1\le j,k \le M\\ i\ne j \ne k \ne i}} F_{ijk} \le 2\,.
	\]
	Multiplying the above inequalities from the right and from the left by $p_i \1_{\{|p_i| >s\}}$, where $p_i = -{\bf i}\nabla_{x_i}$, then summing over $i$, we obtain
	\begin{equation}\label{eq:gen_dyson_1}
		\sum_{i=1}^M p_i^2 \1_{\{|p_i| >s\}} \geq \frac{1}{2} \sum_{\substack { 1\le i,j,k\le M \\ i\neq j \neq k \neq i}} \1_{\{|p_i| >s\}} p_i F_{ijk} p_i \1_{\{|p_i| >s\}}\,.
	\end{equation}
	
	Let us now remove the momentum cut\nobreakdash-off on the right-hand side of~\eqref{eq:gen_dyson_1}. By decomposing $\1_{\{|p_i| >s\}} = 1 - \1_{\{|p_i| \le s\}} $ and using the Cauchy--Schwarz inequality, we obtain for all $0<\varepsilon<1$
	\[
		\1_{\{|p_i| >s\}} p_i F_{ijk} p_i \1_{\{|p_i| > s\}} \ge (1-\eps) p_i F_{ijk} p_i - \eps^{-1} \1_{\{|p_i| \le s\}} p_i F_{ijk} p_i \1_{\{|p_i| \le s\}}\,.
	\]
	Hence, we deduce from~\eqref{eq:gen_dyson_1} that
	\begin{equation}\label{eq:gen_dyson_1b}
		\sum_{i=1}^M p_i^2 \1_{\{|p_i| >s\}} \geq \frac{1-\eps}{2} \sum_{\substack{1\le i,j,k \le M\\ i\ne j \ne k \ne i}} p_i F_{ijk} p_i - \frac{\eps^{-1}}{2} \sum_{\substack{1\le i,j,k \le M\\ i\ne j \ne k \ne i}} \1_{\{|p_i| \le s\}} p_i F_{ijk} p_i \1_{\{|p_i| \le s\}}\,.
	\end{equation}
	For every $1\le i,j \le M$ with $i\ne j$, we have
	\[
		\sum_{\substack{1\le k \le M\\ k\ne i,j}} F_{ijk} \le \chi_R(x_i-x_j)\,.
	\]
	On the other hand, using
	\[
		\norm{f(x)g(p)}_{\gS^2(L^2(\R^d))} = (2\pi)^{-d/2} \norm{f}_{L^2(\R^d)} \norm{g}_{L^2(\R^d)}\,,
	\]
	where $\gS^2(L^2(\R^d))$ is the space of Hilbert--Schmidt operators, we find that
	\begin{align*}
		0 &\le \1_{\{|p_i| \le s\}} p_i \chi_R(x_i-x_j) p_i \1_{\{|p_i| \le s\}} \\
		&\le \norm{ \1_{\{|p_i| \le s\}} p_i \chi_R(x_i-x_j) p_i \1_{\{|p_i| \le s\}} }_{op} \\
		&\le \begin{multlined}[t]
			\norm{ \1_{\{|p_i| \le s\}} p_i \chi_R(x_i-x_j) }_{\gS^2( L^2(\R^3, \d x_i) )}^2 \\
			= (2\pi)^{-3} \norm{\1_{\{|p_i| \le s\}} |p_i| }_{L^2(\R^3, \d p_i)}^2 \norm{\chi_R(x_i-x_j)}_{L^2(\R^3, \d x_i)}^2 \le C s^5 R^3\,.
		\end{multlined}
	\end{align*}
	Thus, we can bound the last term on the right-hand side of~\eqref{eq:gen_dyson_1b} as
	\begin{align*}
		\sum_{\substack{1\le i,j,k \le M\\ i\ne j \ne k \ne i}} \1_{\{|p_i| \le s\}}p_i F_{ijk} p_i \1_{\{|p_i| \le s\}} &\le \sum_{\substack{1\le i,j \le M\\ i\ne j}} \1_{\{|p_i| \le s\}} p_i \chi_R(x_i-x_j) p_i \1_{\{|p_i| \le s\}} \\
			&\leq C s^5 R^3 M^2\,.
	\end{align*}
	Hence,~\eqref{eq:gen_dyson_1b} reduces to
	\[
		\sum_{i=1}^M p_i^2 \1_{\{|p_i| >s\}} \geq \frac{1-\eps}{6} \sum_{\substack{1\le i,j,k \le M\\ i\ne j \ne k \ne i}} \sum_{n\in \{i,j,k\}} p_n F_{ijk} p_n - C\eps^{-1}s^5 R^3 M^2\,,
	\]
	which is equivalent to
	\begin{multline*}
		\sum_{i=1}^M p_i^2 \1_{\{|p_i| >s\}} + \frac{1-\eps}{6} \sum_{\substack{1\le i,j,k \le M\\ i\ne j \ne k \ne i}} W(x_i-x_j,x_i-x_k) \\
		\geq \frac{1-\eps}{6} \sum_{\substack{1\le i,j,k \le M\\ i\ne j \ne k \ne i}} \left[ \sum_{n\in \{i,j,k\}} p_n F_{ijk} p_n + W(x_i-x_j,x_i-x_k) \right] - C\eps^{-1}s^5 R^3 M^2\,.
	\end{multline*}

	Now, let us show that for any given $(i,j,k)$ with $i\ne j\ne k \ne i$, we have
	\begin{multline} \label{eq:many-body-Dyson-f2}
		\sum_{n\in \{i,j,k\}} p_n F_{ijk} p_n + W(x_i-x_j,x_i-x_k) \\
		\ge \left(1- \frac{CR_0}{R}\right) b_{\cM}(W) U(x_i-x_j, x_i-x_k) \prod_{\ell \ne i,j,k} \theta_{2R}\left(\frac{x_i+x_j+x_k}{3} -x_\ell\right).
	\end{multline}
	We do the change of variables similarly to~\eqref{eq:removing-center-of-mass}:
	\[
		r_{i} = \frac{1}{3}(x_i+x_j+x_k)\,, \quad r_{j} = x_i-x_j\,, \quad \text{and} \quad r_k = x_i-x_k\,.
	\]
	Then,
	\[
		p_{x_i} = \frac{1}{3} p_{r_i} + p_{r_j} + p_{r_k}\,, \quad p_{x_j} = \frac{1}{3} p_{r_i} - p_{r_j}\,,\quad \text{and} \quad p_{x_k} = \frac{1}{3} p_{r_i}-p_{r_k}\,,
	\]
	hence
	\begin{align*}
		\sum_{n\in \{i,j,k\}} &p_n F_{ijk} p_n = p_{x_i} F_{ijk} p_{x_i} + p_{x_j} F_{ijk} p_{x_j} + p_{x_k} F_{ijk} p_{x_k} \\
		&= \begin{multlined}[t]
		\left( \frac{1}{3} p_{r_i} + p_{r_j} + p_{r_k} \right) F_{ijk} \left( \frac{1}{3} p_{r_i} + p_{r_j} + p_{r_k} \right) \\
			+ \left( \frac{1}{3} p_{r_i} - p_{r_j} \right) F_{ijk} \left( \frac{1}{3} p_{r_i} - p_{r_j} \right) + \left( \frac{1}{3} p_{r_i}-p_{r_k} \right) F_{ijk} \left( \frac{1}{3} p_{r_i}-p_{r_k} \right)
	\end{multlined} \\
		& = \frac{1}{3}p_{r_i} F_{ijk} p_{r_i} + 2 p_{r_j} F_{ijk} p_{r_j} + 2 p_{r_k} F_{ijk} p_{r_k} + p_{r_j} F_{ijk} p_{r_k} + p_{r_k} F_{ijk} p_{r_j}\,.
	\end{align*}
	We can remove $p_{r_i} F_{ijk} p_{r_i}\ge 0$ for a lower bound. Moreover, by introducing the notation ${\bf r}_{jk} = (r_j,r_k) \in \R^6$ we have
	\[
		F_{ijk} = \chi_R(r_j)\chi_R(r_k) \chi_R(r_j-r_k) \prod_{\ell \neq i,j,k} \theta_{2R}( r_i - x_\ell) \ge \1_{ \{|{\bf r}_{jk}| \le R/2 \}} \prod_{\ell \neq i,j,k} \theta_{2R}(r_i -x_\ell)\,.
	\]
	Thus,
	\begin{multline*}	
		\sum_{n\in \{i,j,k\}} p_n F_{ijk} p_n \\
		\ge 2 p_{r_j} F_{ijk} p_{r_j} + 2 p_{r_k} F_{ijk} p_{r_k} + p_{r_j} F_{ijk} p_{r_k} + p_{r_k} F_{ijk} p_{r_j} = 2 \cM p_{{\bf r}_{jk}} F_{ijk} \cM p_{{\bf r}_{jk}} \\
		\ge 2 \cM p_{{\bf r}_{jk}} \1_{ \{|{\bf r}_{jk}| \le R/2 \}} \cM p_{{\bf r}_{jk}} \prod_{\ell \neq i,j,k}\theta_{2R}(r_i -x_\ell)\,.
	\end{multline*}
	Since $W\ge 0$, we have the obvious bound
	\[
		W(x_i-x_j,x_i-x_k) = W ( {\bf r}_{jk}) \ge W ({\bf r}_{jk}) \prod_{\ell \neq i,j,k} \theta_{2R}(r_i -x_\ell)\,.
	\]
	Using Theorem~\ref{dyson_lemma_Rd_nonradial-M}, we obtain
	\begin{multline*}
		\sum_{n\in \{i,j,k\}} p_n F_{ijk} p_n + W(x_i-x_j,x_i-x_k) \\
		\begin{aligned}[t]
			&\ge \left( 2 \cM p_{{\bf r}_{jk}} \1_{ \{|{\bf r}_{jk}| \le R/2 \}} \cM p_{{\bf r}_{jk}} + W ({\bf r}_{jk})\right) \prod_{\ell \neq i,j,k} \theta_{2R}(r_i -x_\ell)\\
			&\ge (1-C R_0/R) b_{\cM} (W) U_R({\bf r}_{jk}) \prod_{\ell \neq i,j,k} \theta_{2R}(r_i -x_\ell)\,.
		\end{aligned}
	\end{multline*}
	Thus,~\eqref{eq:many-body-Dyson-f2} holds, completing the proof of Lemma~\ref{lem:gen_dyson_lemma}.
\end{proof}

\section{Reduction to softer interaction potentials}\label{sec:reduction_soft}
	In this section, we prove Lemma~\ref{lem:softer} below. This is the central piece for the proof of the lower bound on the energy $E_N$, defined in~\eqref{Def_E_N}, which is given in the next section (Theorem~\ref{theo:lower_bound}). It allows to replace the singular potential $V_N$ by a potential whose scaling is as close to the mean-field scaling as wanted. For Lemma~\ref{lem:softer} to hold, we need some extra condition on the magnetic potential which will be lifted in the end of the proof of Theorem~\ref{theo:lower_bound} ---namely,
\begin{equation}\label{eq:cond_A_V}
	\lim_{|x| \to \infty} \frac{|A(x)|^2}{\Vext(x)} = 0\,.
\end{equation}

\begin{lemma}[Reduction to softer potentials]\label{lem:softer}
	Let $\beta\in(0,3/8]$, $0<\eps<1<s$, and $\delta \in (0, \eps^2/2)$. Assume $\Vext $ to be as in~\eqref{eq:Vext} and $A \in L^3_{\mathrm{loc}}(\mathbb{R}^{3})$ to satisfy~\eqref{eq:cond_A_V}. Let $0\le \widetilde U\in L^\infty(\R^6)$ be radial with $\int_{\R^6} \widetilde U=1$ and $\supp \widetilde U \subset \{1/8 \le |\bx| \le 1/4\}$. Define $U$ as in~\eqref{eq:Ug} and $U_R=R^{-6}U(R^{-1}\cdot)$.
	Then, for all integers $N\geq3$, there exist an integer $M\in[(1-\eps)N, N]$ and $R\in [N^{-\beta}, N^{-\beta/2}]$ such that
	\begin{multline*}
		E_N \ge \inf \sigma_{L^2_s (\R^{3M})} \bigg( \sum_{i=1}^{M} (h_{\eps,s})_i + \frac{b_{\cM}(V)}{6(M-1)(M-2)} \sum_{\substack{ 1\le i,j,k\le M \\ i\ne j\ne k \ne i}} U_{R}(x_i-x_j, x_i-x_k) \bigg) \\
		- C_{\eps,s,\beta,\delta} R^{2/7} N - \eps C_{\beta} N - \delta C_{\eps} N\,,
	\end{multline*}
	where $h_{\eps,s} := h - (1-\eps) \1_{\{ |p| > s\}} p^2$ with $h$ defined in~\eqref{Def_h}.
\end{lemma}
Note that, thanks to~\eqref{eq:cond_A_V}, the operator $h_{\eps,s}$ is bounded below for all $\eps \in (0,1)$ and $s \geq 1$. In fact, for all $\eps \in (0,1)$, there is $C_\eps$ such that
\begin{equation} \label{eq:p_h}
	h_{\eps,s} \geq \frac{\eps}{2} p^2 - C_\eps \quad \text{for all } s\geq 1 \,.
\end{equation}
We first state and prove some preliminary results before giving the proof of Lemma~\ref{lem:softer} in Section~\ref{sec:proof_lem_softer}.

\subsection{Binding inequality}
	Consider the Hamiltonian
\begin{equation} \label{eq:HMN}
	H_{M,N} = \sum_{i=1}^M h_{i} + \sum_{1\leq i < j < k \leq M} NV(N^{1/2}(x_i-x_j), N^{1/2}(x_i-x_k))
\end{equation}
and denote by $E(M,N)$ its ground state energy and by $\Gamma_{M,N}$ the zero-temperature limit of the bosonic Gibbs state ---that is, the uniform average over all ground states of $H_{M,N}$. For any observable $A$, we denote $\braket{A}_{\Gamma_{M,N}} = \Tr A \Gamma_{M,N}$.
\begin{lemma}\label{lem:binding}
	There exists a constant $C>0$ such that for any integer $N\geq3$ and any $0<\eps<1$, there exists an integer $M \equiv M(N,\eps) \in [(1-\eps)N, N]$ satisfying
	\begin{equation}\label{eq:binding_inequality_M}
		E(M,N) - E(M-4,N) \leq C\eps^{-1}\,.
	\end{equation}
\end{lemma}
\begin{proof}
	Denote
	\[
		Z:= \min_{m\in [(1-\eps)N, N]\cap\N} \left( E(m,N) - E(m-4,N) \right).
	\]
	Since $\Vext \geq0$ and $V \geq 0$, $E(m,N)$ is nonnegative and increasing in $m$. Hence,
	\[
		4E(N,N) \ge \sum_{m\in[(1-\eps)N, N]\cap\N} \left( E(m,N) - E(m-4,N) \right) \ge Z (\eps N -1)\,.
	\]
	Combined with the simple upper bound $E(N,N) \le C N$ ---take for instance $u^{\otimes N}$ with $u \in C^{\infty}_c$---, it gives $Z\le C\eps^{-1}$.
\end{proof}
	
\begin{remark}
	Using the concavity of the NLS functional in the parameter in front of the nonlinearity and following the argument in the proof of~\cite[Proposition 1]{LieSei-06}, we can even find, for all $N$, an $M \equiv M(N)$ such that $N - o(N) \le M \le N$ and
	\[
		E(M,N) - E(M-4,N) \leq C\,.
	\]
	However, this stronger conclusion is not needed for our purpose.	
\end{remark}
	
\subsection{Four-body estimate}

\begin{lemma}\label{lem:4_body_collision}
	Let $M \equiv M(N,\eps)$ be as in Lemma~\ref{lem:binding}. Then, the zero temperature limit of the Gibbs state $\Gamma_{M,N}$ of $H_{M,N}$ satisfies
	\begin{equation}\label{eq:4_body_collision}
		\bigg\< \prod_{i=2}^{4} \1_{\{|x_1-x_{i}|\leq R \}} \bigg\>_{\Gamma_{M,N}} \leq C_\eps R^9\,.
	\end{equation}
\end{lemma}
\begin{proof}
	With an immediate adaptation of the proof of~\cite[Lemma 2]{LieSei-06} to three\nobreakdash-body interaction potentials, in particular using that $\Vext \geq 0$ and $V \geq 0$, we obtain
	\[
		\pscal{\xi(x_1,x_2,x_3,x_4)}_{\Gamma_{M,N}} \leq e^{(E(M,N) - E(M-4,N))} \norm{ \sqrt{\xi} \prod_{i=1}^4 e^{\Delta_{x_i}} \sqrt{\xi}}_{op}
	\]
	for any measurable function $\xi(x_1,x_2,x_3,x_4) \geq 0$, where the operator norm in the right-hand side is the one in $L^2(\R^{12})$. In particular, applying this bound to
	\[
		\xi (x_1,x_2,x_3,x_4) = \prod_{i=2}^4 \1_{\{|x_1-x_{i}|\leq R \}}\,, \quad x_i\in \R^3\,,
	\]
	and using $E(M,N) - E(M-4,N) \le C\eps^{-1}$, we find that
	\begin{multline*}
		\bigg\< \prod_{i=2}^4 \1_{\{|x_1-x_{i}|\leq R \}} \bigg\>_{\Gamma_{M,N}} \\
		\begin{aligned}[t]
			&\leq C_\eps \norm{\prod_{i=2}^4 \1_{\{|x_1-x_{i}| \le R\}} \prod_{i=1}^4 e^{\Delta_{x_i}} \prod_{i=2}^4 \1_{\{|x_1-x_{i}| \le R\}}}_{op} \\
			&\leq
			\begin{multlined}[t][0.87\textwidth]
					C_\eps \prod_{i=2}^4 \norm{ \1_{\{|x_1-x_{i}| \le R\}} e^{- \frac1{2} p_i^2}}_{\gS^2(L^2(\R^3,\d x_i))}^2 \\
					= C_\eps \prod_{i=2}^4 (2\pi)^{-3} \norm{ \1_{\{|x_1-x_{i}| \le R\}} }_{L^2(\R^3, \d x_i)}^2 \norm{ e^{ - \frac1{2} p_i^2} }_{L^2(\R^3,\d p_i)}^2 \leq C_\eps R^9\,,
			\end{multlined}
		\end{aligned}
	\end{multline*}
	where we used that $\norm{K}_{\gS^2(L^2(\R^3))}^2 = \int_{\mathbb{R}^{3}\times \mathbb{R}^{3}} |K(x,y)|^2 \d x \d y$ for any operator $K \in \gS^2(L^2(\R^3))$ with kernel $K(x,y)$. \qedhere
\end{proof}

\subsection{Proof of Lemma~\ref{lem:softer}} \label{sec:proof_lem_softer}
	We now are ready to conclude.
\begin{proof}[Proof of Lemma~\ref{lem:softer}]
	{\bf Step~1.}
	Let $s>1>\eps>0$ be independent of $N$. Let $M$, with $(1-\eps)N\le M \le N$, be as in Lemma~\ref{lem:binding} and $\Gamma_{M,N}$ be the zero temperature limit of the Gibbs state for $H_{M,N}$ in~\eqref{eq:HMN}. We first prove a bound on $E(M,N)$ and deduce, in the last step of this proof, the desired estimate for $E(N,N)$. Let
	\begin{equation}\label{eq:condition_R_0}
		M^{-1/2} \ll R_0 \ll M^{-1/3}\,.
	\end{equation}
	Applying Lemma~\ref{lem:gen_dyson_lemma} with $W=N V(N^{1/2}\cdot)$, which is supported in $B(0,C N^{-1/2})$ and has $b_{\cM}(W)= b_\cM (V)/N^2$, we have the operator inequality on $L^2_s(\R^{3M})$
	\begin{align*}
		\sum_{i=1}^M p_i^2 \1_{\{|p_i| >s\}} &+ \frac{1}{6} \sum_{\substack{ 1\le i,j,k \leq M \\ i\neq j \neq k \neq i } } N V (N^{1/2}(x_i-x_j), N^{1/2}(x_i-x_k)) \\
		&\geq \begin{multlined}[t]
			\frac{b_{\cM}(V) (1-\varepsilon)}{6 N^2} \left( 1 - \frac{C}{N^{1/2}R_0}\right) \\
			\times \sum_{\substack{ 1\le i,j,k \leq M \\ i\neq j \neq k \neq i } } U_{R_0}(x_i-x_j, x_i-x_k) \prod_{\ell \neq i,j,k} \theta_{2{R_0}}\left(\frac{x_i+x_j+x_k}{3}-x_\ell \right) \\
			- C \varepsilon^{-1} s^5 {R_0^3} M^2\,.
		\end{multlined}
	\end{align*}
	Multiplying both sides by $(1-\eps)/M$ and using $(1-\eps)N\le M \le N$, we obtain
	\begin{multline} \label{eq:soft_potential_first_step}
		\frac{H_{M,N}}{M} \geq \frac{b_{\cM}(V)(1-\eps)^4} {6M^3} \left(1- \frac{C}{M^{1/2} {R_0}} \right) \\
		\times \sum_{\substack{ 1\le i,j,k\le M \\ i\ne j\ne k \ne i}} U_{R_0}(x_i-x_j, x_i-x_k) \prod_{\ell \neq i,j,k} \theta_{2{R_0}}\left( \frac{x_i+x_j+x_k}{3}-x_\ell\right) \\
		+ \frac{1}{M}\sum_{i=1}^M (h_{\eps,s})_i - C_{\eps,s} M {R_0^3}\,,
	\end{multline}
	where we recall that $h_{\eps,s} = h - (1-\eps) \1_{\{ |p| > s\}} p^2$.

	Next, we remove the four\nobreakdash-body cut\nobreakdash-off in the interaction term on the right-hand side of~\eqref{eq:soft_potential_first_step}. Recall that $\chi_R(x) =\1_{\{|x| \le R\}} = 1-\theta_R(x)$. Using $\supp U_R \subset B(0,R)$ and Bernoulli's inequality, we have for every $1\le i,j,k\le M$ with $i\ne j \ne k\ne i$,
	\begin{align*}
		0 &\le U_{R_0}(x_i-x_j,x_i-x_k)\bigg[ 1 - \prod_{\ell \neq i,j,k} \theta_{2{R_0}}\left( \frac{x_i+x_j+x_k}{3}-x_\ell\right) \bigg] \\
		&\le \frac{C}{R_0^6} \chi_{R_0}(x_i-x_j) \chi_{R_0}(x_i-x_k) \bigg[ 1 - \prod_{\ell \neq i,j,k} \theta_{2{R_0}}\left( \frac{x_i+x_j+x_k}{3}-x_\ell\right) \bigg]\\
		&\le \begin{multlined}[t]
			\frac{C}{R_0^6} \chi_{R_0}(x_i-x_j) \chi_{R_0}(x_i-x_k) \bigg[ 1 - \prod_{\ell \neq i,j,k} \theta_{4{R_0}}\left( x_i-x_\ell\right) \bigg]\\
			= \frac{C}{R_0^6} \chi_{R_0}(x_i-x_j) \chi_{R_0}(x_i-x_k) \bigg[ 1 - \prod_{\ell \neq i,j,k} \left( 1- \chi_{4{R_0}}\left(x_i-x_\ell\right) \right)\bigg]
		\end{multlined}\\
		&\le \frac{C}{R_0^6} \sum_{\ell \ne i,j,k} \chi_{4{R_0}}(x_i-x_j) \chi_{4{R_0}}(x_i-x_k) \chi_{4{R_0}}(x_i-x_\ell)\,.
	\end{align*}
	Combining it with Lemma~\ref{lem:4_body_collision}, we find that
	\begin{align*}
		&\sum_{\substack{ 1\le i,j,k\le M \\ i\ne j\ne k\ne i} } \bigg\< U_{R_0}(x_i-x_j,x_i-x_k) \bigg[ 1 - \prod_{\ell \neq i,j,k} \theta_{2{R_0}}\left( \frac{x_i+x_j+x_k}{3}-x_\ell\right) \bigg] \bigg\>_{\Gamma_{M,N}} \\
		&\le \frac{C}{R_0^6} \sum_{\substack{ 1\le i,j,k\le M \\ i\ne j\ne k\ne i} } \sum_{\ell \ne i,j,k} \Braket{ \chi_{4{R_0}}(x_i-x_j)\chi_{4{R_0}}(x_i-x_k)\chi_{4{R_0}}(x_i-x_\ell)}_{\Gamma_{M,N}} \le C_\eps M^4 {R_0^3}
	\end{align*}
	thence, from~\eqref{eq:soft_potential_first_step}, that
	\begin{multline}\label{eq:removing_4_body_cutoff-00}
		\frac{E(M,N)}{M} = \frac{1}{M} \pscal{ H_{M,N} }_{\Gamma_{M,N} } \\
		\ge \frac{1}{M}\bigg\< \sum_{i=1}^M (h_{\eps,s})_i + \frac{b_\eps}{6M^2} \left( 1 - \frac{C}{M^{1/2}{R_0}}\right) \sum_{\substack{ 1\le i,j,k\le M \\ i\ne j\ne k \ne i}} U_{R_0}(x_i-x_j, x_i-x_k) \bigg\>_{\Gamma_{M,N}} \\
		- C_{\eps,s} M {R_0^3}\,,
	\end{multline}
	where we used the notation $b_\eps = b_{\cM}(V)(1-\eps)^4$.
	Combining it with the simple upper bound $E(M,N) \le CM$, we get
	\begin{equation}\label{eq:UR-sb}
		\frac{b_\eps}{6M^3} \bigg\< \sum_{\substack{ 1\le i,j,k\le M \\ i\ne j\ne k \ne i}} U_{R_0}(x_i-x_j, x_i-x_k) \bigg\>_{\Gamma_{M,N}} \le C_{\eps,s}\,.
	\end{equation}
	Hence,~\eqref{eq:removing_4_body_cutoff-00} can be simplified to
	\begin{multline} \label{eq:soft-potential-main0}
		\frac{E(M,N)}{M} \ge \frac{1}{M}\bigg\< \sum_{i=1}^M (h_{\eps,s})_i + \frac{b_\eps}{6M^2} \sum_{\substack{ 1\le i,j,k\le M \\ i\ne j\ne k \ne i}} U_{R_0}(x_i-x_j, x_i-x_k) \bigg\>_{\Gamma_{M,N}}\\
		- C_{\eps,s} \left( \frac{1}{M^{1/2}{R_0}}+ M{R_0^3} \right).
	\end{multline}
	The optimal choice of ${R_0}$ is determined by $1/(M^{1/2}{R_0}) = M{R_0^3}$ ---namely,
	\[
		R_0 = M^{-3/8}\,.
	\]
	The condition~\eqref{eq:condition_R_0} is clearly satisfied. With this choice, we have
	\begin{multline} \label{eq:soft-potential-main1}
		\frac{E(M,N)}{M} \ge \frac{1}{M}\bigg\< \sum_{i=1}^M (h_{\eps,s})_i + \frac{b_\eps}{6M^2} \sum_{\substack{ 1\le i,j,k\le M \\ i\ne j\ne k \ne i}} U_{R_0}(x_i-x_j, x_i-x_k) \bigg\>_{\Gamma_{M,N}} \\
		- C_{\eps,s} M^{-1/8}\,.
	\end{multline}
	
	\bigskip
	\noindent{\bf Step~2.}
	The error term $M R_0^3$ in~\eqref{eq:soft-potential-main0} forbids to directly take $R_0 \sim N^{-\beta}$ and to conclude. Essentially, it means that~\eqref{eq:soft-potential-main0} is only useful in the dilute regime $R_0 \lesssim M^{-1/3}$. We now use the bosonic symmetry to reformulate the energy as the one of a system with fewer particles $M_1 \ll M$, broadening the range of the interaction for which the system is dilute. Using again the Dyson Lemma, we replace the potential $U_{R_0}$ by a softer one $U_{R_1}$ for $ R_0 \ll R_1 \ll M_1^{-1/3}$. Unlike in Step~1, here, we can only use a small fraction of the kinetic energy. For that reason, we replace the particle number $M$ by a smaller parameter $M_1$ in order to improve the error caused by the removal of the four\nobreakdash-body cut\nobreakdash-off. Keeping in mind that ${R_0} =M^{-3/8}$, the parameters $M_1$ and $R_1$ are chosen such that
	\begin{equation} \label{eq:con-R1-M1}
		M_1^{-1/3} \gg R_1\gg {R_0} \gg M_1^{-1/2}\,.
	\end{equation}
	
	To be precise, by the bosonic symmetry of $\Gamma_{M,N}$, the main term in~\eqref{eq:soft-potential-main1} can be written as
	\begin{align}\label{eq:soft-potential-main-2a}
		&\frac{1}{M}\bigg\< \sum_{i=1}^M (h_{\eps,s})_i + \frac{b_\eps}{6M^2} \sum_{\substack{ 1\le i,j,k\le M \\ i\ne j\ne k \ne i}} U_{R_0} (x_i-x_j, x_i-x_k) \bigg\>_{\Gamma_{M,N}} \nn \\
		&= \frac{1}{M_1} \bigg\< \sum_{i=1}^{M_1} (h_{\eps,s})_i + \frac{b_\eps (M-1)(M-2)}{6M^2 (M_1-1)(M_1-2)} \sum_{\substack{ 1\le i,j,k\le M_1 \\ i\ne j\ne k \ne i}} U_{R_0} (x_i-x_j, x_i-x_k) \bigg\>_{\Gamma_{M,N}} \nn\\
		&\ge \frac{1}{M_1}\bigg\< \sum_{i=1}^{M_1} (h_{\eps,s})_i + \frac{b_\eps}{6M_1^2} \sum_{\substack{ 1\le i,j,k\le M_1 \\ i\ne j\ne k \ne i}} U_{R_0} (x_i-x_j, x_i-x_k) \bigg\>_{\Gamma_{M,N}}\,.
	\end{align}
	Here, we used that
	\[
		\frac{(M-1)(M-2)}{M^2} \geq \frac{(M_1-1)(M_1-2)}{M_1^2} \ge 0 \quad \text{for all } M\geq M_1 \geq 2\,.
	\]
	Next, we take a small parameter $\delta \in (0,\eps^2/2)$ and apply the Dyson lemma to the potential $W= \delta^{-1} b_\eps M_1^{-2}U_{R_0} $, which is supported in $B(0,{R_0} )$ and has the scattering energy
	\[
		b_{\cM}(W) \ge \norm{W}_{L^1(\R^6)} - C \norm{W}_{L^{3/2}(\R^6)}^2 \ge \frac{b_\eps}{\delta M_1^2} \left( 1- \frac{C_{\eps,\delta}} {M_1^{2}{R_0} ^{4}} \right).
	\]
	Here, we used~\eqref{eq:born-M} in the latter estimate (the condition ${R_0} \gg M_1^{-1/2}$ ensures that the scattering energy of $W$ is well approximated by its first Born approximation).
	Thus, from~\eqref{eq:lem_gen_Dyson-simpler}, after multiplying both sides by $\delta$, we have the operator inequality on $L^2_s(\R^{3M_1})$
	\begin{multline*}
		\delta \sum_{i=1}^{M_1} p_i^2 + \frac{b_\eps}{6 M_1^2} \sum_{\substack{ 1\le i,j,k\le M_1 \\ i\ne j\ne k \ne i}} U_{R_0} (x_i-x_j, x_i-x_k) \\
		\geq \begin{multlined}[t]
			\frac{b_\eps }{6 M_1^2} \left(1- \frac{C_{\eps,\delta}}{ M_1^{2}{R_0} ^{4}} \right) \left( 1 - \frac{C {R_0} }{R_1}\right) \\
			\times \sum_{\substack{ 1\le i,j,k \leq M_1 \\ i\neq j \neq k \neq i } } U_{R_1}(x_i-x_j, x_i-x_k) \prod_{\ell \neq i,j,k} \theta_{2R_1}\left(\frac{x_i+x_j+x_k}{3}-x_\ell \right).
		\end{multlined}
	\end{multline*}
	Thanks to Lemma~\ref{lem:4_body_collision}, we can remove the four\nobreakdash-body cut\nobreakdash-off in the interaction term similarly to Step~1 ---namely,
	\begin{align*}
		&\sum_{\substack{ 1\le i,j,k\le M_1 \\ i\ne j\ne k\ne i} } \bigg\< U_{R_1}(x_i-x_j,x_i-x_k) \bigg[ 1 - \prod_{\ell \neq i,j,k} \theta_{2{R_1}}\left( \frac{x_i+x_j+x_k}{3}-x_\ell\right) \bigg] \bigg\>_{\Gamma_{M,N}} \\
		&\le \frac{C}{R_1^6} \sum_{\substack{ 1\le i,j,k\le M_1 \\ i\ne j\ne k\ne i} } \sum_{\ell \ne i,j,k} \braket{ \chi_{4R_1}(x_i-x_j)\chi_{4R_1}(x_i-x_k)\chi_{4R_1}(x_i-x_\ell) }_{\Gamma_{M,N}} \le C_\eps M_1^4 R_1^3\,.
	\end{align*}
	Hence,
	\begin{multline*}
		\frac{1}{M_1}\bigg\< \sum_{i=1}^{M_1} \delta p_i^2 + \frac{b_\eps}{6M_1^2} \sum_{\substack{ 1\le i,j,k\le M_1 \\ i\ne j\ne k \ne i}} U_{R_0} (x_i-x_j, x_i-x_k) \bigg\>_{\Gamma_{M,N}}\\
		\ge \frac{b_\eps}{6M_1^3} \left(1- \frac{C_{\eps,\delta}}{ M_1^{2}{R_0} ^{4}} \right) \left( 1 - \frac{C {R_0} }{R_1}\right) \bigg\< \sum_{\substack{ 1\le i,j,k\le M_1 \\ i\ne j\ne k \ne i}} U_{R_0} (x_i-x_j, x_i-x_k) \bigg\>_{\Gamma_{M,N}} \\
		- C_\eps M_1R_1^3\,.
	\end{multline*}
	From~\eqref{eq:p_h} and $\delta \in (0,\eps^2/2)$, we obtain $\delta p^2 \le \eps h_{\eps,s} + \delta C_\eps$. Inserting the latter bound in~\eqref{eq:soft-potential-main-2a}, we deduce from~\eqref{eq:soft-potential-main1} that
	\begin{multline}\label{eq:soft-potential-main-2b}
		\frac{E(M,N)}{M (1-\eps)} \ge \frac{1}{M_1}\bigg\< \sum_{i=1}^{M_1} (h_{\eps,s})_i \bigg\>_{\Gamma_{M,N}} \\
		+ \frac{b_\eps}{6M_1^2} \left(1- \frac{C_{\eps,\delta}}{M_1^{2}{R_0} ^{4}} \right) \left( 1 - \frac{C {R_0} }{R_1}\right) \bigg\< \sum_{\substack{ 1\le i,j,k\le M_1 \\ i\ne j\ne k \ne i}} U_{R_0} (x_i-x_j, x_i-x_k) \bigg\>_{\Gamma_{M,N}} \\
		- C_{\eps,s} M^{-1/8} - C_\eps M_1R_1^3 - C_\eps \delta\,.
	\end{multline}
	Combined with the simple bound $E(M,N) \le CM$, it gives an analogue of~\eqref{eq:UR-sb}:
	\[
		\frac{b_\eps}{6M_1^3} \bigg\< \sum_{\substack{ 1\le i,j,k\le M_1 \\ i\ne j\ne k \ne i}} U_{R_1}(x_i-x_j, x_i-x_k) \bigg\>_{\Gamma_{M,N}} \le C_{\eps,s}\,.
	\]
	Thus,~\eqref{eq:soft-potential-main-2b} can be simplified into
	\begin{multline*}
		\frac{E(M,N)}{M(1-\eps)} \ge \frac{1}{M_1}\bigg\< \sum_{i=1}^{M_1} (h_{\eps,s})_i + \frac{b_\eps}{6M_1^2} \sum_{\substack{ 1\le i,j,k\le M_1 \\ i\ne j\ne k \ne i}} U_{R_0} (x_i-x_j, x_i-x_k) \bigg\>_{\Gamma_{M,N}} \\
		- C_{\eps,s} M^{-1/8} - C_{\eps,\delta} \left( M_1R_1^3 + \frac{1}{M_1^2 {R_0} ^4} + \frac{{R_0} }{R_1}\right) -C_\eps \delta \,.
	\end{multline*}
	We can choose $M_1$ and $R_1$ such that
	\[
		M_1R_1^3= \frac{1}{M_1^2 R_0^4} = \frac{{R_0} }{R_1} = \left ( \left ( M_1R_1^3 \right)^2 \frac{1}{M_1^2 R^4} \left (\frac{{R_0} }{R_1} \right)^6 \right)^{1/9}= {R_0} ^{2/9},
	\]
	namely
	\[
		R_1={R_0} ^{7/9} \quad \text{and} \quad M_1= \frac{{R_0} ^{2/9}}{R_1^{3}} = R_1^{-19/7}\,.
	\]
	The condition~\eqref{eq:con-R1-M1} is clearly satisfied. With this choice of parameters, we also have $M^{-1/8} \ll R_1^{2/7}$, hence we arrive at
	\begin{multline*}
		\frac{E(M,N)}{M(1-\eps)} \ge \frac{1}{M_1}\bigg\< \sum_{i=1}^{M_1} (h_{\eps,s})_i + \frac{b_\eps}{6M_1^2} \sum_{\substack{ 1\le i,j,k\le M_1 \\ i\ne j\ne k \ne i}} U_{R_1} (x_i-x_j, x_i-x_k) \bigg\>_{\Gamma_{M,N}} \\
		- C_{\eps,s,\delta} R_1^{2/7} - C_{\eps} \delta\,.
	\end{multline*}

	\bigskip
	\noindent{\bf Step~3.}
	We can iterate the argument in Step~2 in order to reach softer potentials. Denote
	\begin{equation} \label{eq:RjMj-choice}
		R_j = R_{j-1}^{7/9} \quad \text{and} \quad M_j = R_j^{-19/7} \quad \text{for all } j =1, 2, \dots \,.
	\end{equation}
	Then, by induction, we can prove that, for every $J \in \N\setminus\{0\}$,
	\begin{multline}\label{eq:soft-potential-main-3}
		\frac{E(M,N)}{M (1-\eps)^{J}} \ge \frac{1}{M_J}\bigg\< \sum_{i=1}^{M_J} (h_{\eps,s})_i + \frac{b_\eps}{6M_J^2} \sum_{\substack{ 1\le i,j,k\le M_J \\ i\ne j\ne k \ne i}} U_{R_J}(x_i-x_j, x_i-x_k) \bigg\>_{\Gamma_{M,N}} \\
		- C_{\eps,s,\delta,J} R_{J} ^{2/7} - C_{\eps,J}\delta\,.
	\end{multline}
	Indeed, the case $J=1$ has been handled in Step~2 and the general case is very similar so we only mention some main estimates for the reader's convenience. Assuming that~\eqref{eq:soft-potential-main-3} holds for $J-1$, we can write by the bosonic symmetry that
	\begin{multline*}
		\frac{E(M,N)}{M (1-\eps)^{J-1}} + C_{\eps,s,\delta,J} R_{J-1} ^{2/7} + C_{\eps} \delta \nn \\
		\begin{aligned}
			&\ge \frac{1}{M_{J-1}}\bigg\< \sum_{i=1}^{M_{J-1}} (h_{\eps,s})_i + \frac{b_\eps}{6M_{J-1}^2} \sum_{\substack{ 1\le i,j,k\le M_{J-1} \\ i\ne j\ne k \ne i}} U_{R_{J-1}}(x_i-x_j, x_i-x_k) \bigg\>_{\Gamma_{M,N}} \nn \\
			&\ge \frac{1}{M_{J}}\bigg\< \sum_{i=1}^{M_{J}} (h_{\eps,s})_i + \frac{b_\eps}{6M_{J}^2} \sum_{\substack{ 1\le i,j,k\le M_{J} \\ i\ne j\ne k \ne i}} U_{R_{J-1}}(x_i-x_j, x_i-x_k) \bigg\>_{\Gamma_{M,N}}\,.
		\end{aligned}
	\end{multline*}
	Applying the Dyson lemma to $W= \delta b_\eps M_{J}^{-2}U_{R_{J-1}}$, we deduce from~\eqref{eq:lem_gen_Dyson-simpler} that
	\begin{align*}
		\delta \sum_{i=1}^{M_J} p_i^2 &+ \frac{b_\eps}{6 M_J^2} \sum_{\substack{ 1\le i,j,k\le M_J \\ i\ne j\ne k \ne i}} U_{R_{J-1}} (x_i-x_j, x_i-x_k) \\
		&\geq \begin{multlined}[t]
			\frac{b_\eps }{6 M_J^2} \left(1- \frac{C_{\eps,\delta}}{M_J^{2}R_{J-1}^{4}} \right) \left( 1 - \frac{C R_{J-1}}{R_{J}}\right) \\
			\times \sum_{\substack{ 1\le i,j,k \leq M_{J} \\ i\neq j \neq k \neq i } } U_{R_J}(x_i-x_j, x_i-x_k) \prod_{\ell \neq i,j,k} \theta_{2R_J}\left(\frac{x_i+x_j+x_k}{3}-x_\ell \right).
		\end{multlined}
	\end{align*}
	The four\nobreakdash-body cut\nobreakdash-off can be removed by Lemma~\ref{lem:4_body_collision}, leading to
	\begin{align*}
		&\sum_{\substack{ 1\le i,j,k\le M_J \\ i\ne j\ne k\ne i} } \bigg \< U_{R_J}(x_i-x_j,x_i-x_k)\bigg[ 1 - \prod_{\ell \neq i,j,k} \theta_{2{R_J}}\left( \frac{x_i+x_j+x_k}{3}-x_\ell\right) \bigg] \bigg\>_{\Gamma_{M,N}} \\
		&\le \frac{C}{R_J^6} \sum_{\substack{ 1\le i,j,k\le M_J \\ i\ne j\ne k\ne i} } \sum_{\ell \ne i,j,k} \braket{ \chi_{4R_J}(x_i-x_j)\chi_{4R_J}(x_i-x_k)\chi_{4R_J}(x_i-x_\ell) }_{\Gamma_{M,N}} \le C_\eps M_J^4 R_J^3\,.
	\end{align*}
	Moreover, combining it with the simple upper bound $E(M,N) \le CM$, we get
	\begin{equation} \label{eq:upper-mJ}
		\frac{b_\eps}{6M_J^3} \bigg\< \sum_{\substack{ 1\le i,j,k\le M_J \\ i\ne j\ne k \ne i}} U_{R_J}(x_i-x_j, x_i-x_k) \bigg\>_{\Gamma_{M,N}} \le C_{\eps,s}\,.
	\end{equation}
	Hence,
	\begin{multline*}
		\frac{1}{M_J}\bigg\< \delta \sum_{i=1}^{M_J} p_i^2 + \frac{b_\eps}{6 M_J^2} \sum_{\substack{ 1\le i,j,k\le M_J \\ i\ne j\ne k \ne i}} U_{R_{J-1}} (x_i-x_j, x_i-x_k) \bigg\>_{\Gamma_{M,N}} \\
		\ge \frac{b_\eps}{6 M_J^3} \bigg\< \sum_{\substack{ 1\le i,j,k\le M_J \\ i\ne j\ne k \ne i}} U_{R_{J}} (x_i-x_j, x_i-x_k) \bigg\>_{\Gamma_{M,N}} \\
		- C_{\eps,s,\delta} \left( M_JR_J^3 + \frac{C R_{J-1}}{R_{J}} + \frac{1}{M_J^{2}R_{J-1}^{4}} \right).
	\end{multline*}
	Thus, using again that $\delta p^2 \le \eps h_{\eps,s} + \delta C_\eps$, we obtain
	\begin{multline*}
		\frac{E(M,N)}{M (1-\eps)^{J}}
		\ge \frac{1}{M_{J}}\bigg\< \sum_{i=1}^{M_{J-1}} (h_{\eps,s})_i + \frac{b_\eps}{6M_{J}^2} \sum_{\substack{ 1\le i,j,k\le M_{J} \\ i\ne j\ne k \ne i}} U_{R_{J}}(x_i-x_j, x_i-x_k) \bigg\>_{\Gamma_{M,N}} \\
		- C_{\eps,s,\delta,J} R_{J-1} ^{2/7} - C_{\eps,s,\delta} \left( M_JR_J^3 + \frac{1}{M_J^{2}R_{J-1}^{4}} + \frac{ R_{J-1}}{R_{J}} \right) - C_{\eps} \delta\,.
	\end{multline*}
	With the choice in~\eqref{eq:RjMj-choice}, we have
	\[
		M_JR_J^3 = \frac{1}{M_J^{2}R_{J-1}^{4}}= \frac{R_{J-1}}{R_{J}} =R_{J-1}^{2/9} = R_J^{2/7}
	\]
	and the desired estimate~\eqref{eq:soft-potential-main-3} follows.

	\bigskip
	\noindent {\bf Step~4.}
	To conclude, we choose $J \in \N\setminus\{0\}$ depending only on $\beta$ such that
	\[
		\frac{\beta}{2} < \frac{3}{8} \left(\frac{7}{9}\right)^{J} \le \beta\,.
	\]
	Then,
	\[
		N^{-\beta/2}\gg R_J = (R_0)^{\left(\frac{7}{9}\right)^{J} } = M^{- \frac{3}{8} \left(\frac{7}{9}\right)^{J}} \ge N^{-\beta}\,.
	\]
	From~\eqref{eq:soft-potential-main-3} and~\eqref{eq:upper-mJ}, we have
	\begin{align*}
		&\frac{E(M,N)}{M (1-\eps)^{J}} \\
		&\ge \begin{multlined}[t][.95\textwidth]
			\frac{1}{M_J}\bigg\< \sum_{i=1}^{M_J} (h_{\eps,s})_i + \frac{b_\eps}{6M_J^2} \sum_{\substack{ 1\le i,j,k\le M_J \\ i\ne j\ne k \ne i}} U_{R_J}(x_i-x_j, x_i-x_k) \bigg\>_{\Gamma_{M,N}} \\
			- C_{\eps,s,\delta,J}R_{J} ^{2/7} - C_{\eps} \delta
		\end{multlined}\\
		&= \begin{multlined}[t][.95\textwidth]
			\frac{1}{M}\bigg\< \sum_{i=1}^{M} (h_{\eps,s})_i + \frac{b_\eps (M_J-1)(M_J-2)}{6M_J^2 (M-1)(M-2)} \sum_{\substack{ 1\le i,j,k\le M_J \\ i\ne j\ne k \ne i}} U_{R_J}(x_i-x_j, x_i-x_k) \bigg\>_{\Gamma_{M,N}}\\
			- C_{\eps,s,\delta,J} R_{J} ^{2/7} - C_{\eps} \delta
		\end{multlined} \\
		&\ge \begin{multlined}[t][.95\textwidth]
			\frac{1}{M}\bigg\< \sum_{i=1}^{M} (h_{\eps,s})_i + \frac{b_{\cM}(V)}{6(M-1)(M-2)} \sum_{\substack{ 1\le i,j,k\le M \\ i\ne j\ne k \ne i}} U_{R_J}(x_i-x_j, x_i-x_k) \bigg\>_{\Gamma_{M,N}} \\
			- C_{\eps,s,\delta,J} R_{J} ^{2/7} - \varepsilon C_J - C_{\eps} \delta\,,
		\end{multlined}
	\end{align*}
	where we used that $1-\frac{(M_J-1)(M_J-2)}{M_J^2} \leq C M_J^{-1}$ and $|b_\eps - b_{\cM}(V) | \leq C\varepsilon$, together with the estimate
	\begin{multline*}
		\bigg\< \frac{1}{6(M-1)(M-2)} \sum_{\substack{ 1\le i,j,k\le M_J \\ i\ne j\ne k \ne i}} U_{R_J}(x_i-x_j, x_i-x_k) \bigg\>_{\Gamma_{M,N}} \\
		\le C \braket{U_{R_J}(x_1-x_2, x_1-x_3)}_{\Gamma_{M,N}} \leq C_J \left( \frac{E(M,N)}{M} + 1 \right) \leq C_{J}\,.
	\end{multline*}
	Since $N\ge M \ge N(1-\eps)$ and $CN \ge E_N \ge E(M,N)$, we have the immediate bound
	\[
		\frac{E(M,N)}{M (1-\eps)^{J}} \le \frac{E_N}{N (1-\eps)^{J+1}} \le \frac{E_N}{N} + \eps C_J\,.
	\]
	The desired conclusion of Lemma~\ref{lem:softer} follows.
\end{proof}

\section{Conclusion of the energy lower bound} \label{sec:lower}
	Recalling that
\[
	\GPnrg = \inf\limits_{\norm{u}_2=1} \Egp(u) = \inf\limits_{\norm{u}_2=1} \left\{ \pscal{ u, hu } + \frac{b_{\cM}(V)}{6} \int_{\R^3} |u|^6 \right\},
\]
as defined in~\eqref{eq:NLS-GSE}--\eqref{eq:NLS}, this section is devoted to the proof of the following theorem.
\begin{theorem}[Energy lower bound]\label{theo:lower_bound}
	Suppose that $\Vext $ and $A$ satisfy~\eqref{eq:Vext} and~\eqref{eq:A}. Then,
\[
		\lim\limits_{N\to\infty} \frac{E_N}{N} \geq \GPnrg\,.
\]
\end{theorem}
\begin{proof}
	We prove Theorem~\ref{theo:lower_bound} with the extra assumption that $A$ satisfies~\eqref{eq:cond_A_V}. This assumption can be removed, at the end, following an argument of~\cite[Sect.~4B]{NamRouSei-16} that we omit here.
	
	Recall that $h_{\eps,s} = h - (1-\eps) \1_{\{ |p| > s\}} p^2$. For all $\varepsilon>0$, there exist $C_{\eps}, c_0 >0$ such that, on $L^2(\R^3)$, we have
	\[
		h_{\eps,s} \ge \frac{\eps}{2} p^2 + \Vext(x) - C\eps^{-1} |A(x)|^2 \geq \frac{\eps}{4} p^2 + \frac{1}{2}\Vext(x) - C_{\eps} \ge \frac{\eps}{4} p^2 + c_0 |x|^{\alpha} - C_{\eps}\,.
	\]
	Let us therefore define $\widetilde{h}_{\eps,s} = h_{\eps,s} - \kappa_{\eps,s}$, where $\kappa_{\eps,s} := \inf \sigma (h_{\eps,s}) - 1 $. Thanks to the Lieb--Thirring inequality in~\cite[Theorem 3]{DolFelLosPat-06}, we have
	\[
		\Tr ( (-\Delta+|x|^\alpha +1)^{-q}) \le \int_{\R^3}\int_{\R^3} \frac{1}{(|2\pi k|^2 + |x|^\alpha +1)^q} \d k \d x <\infty
	\]
	for any $q>q_*(\alpha):=\frac32+\frac3\alpha$. Thus, $\Tr (\widetilde{h}_{\eps,s}^{-q}) \le C_{\eps}$.
	Therefore, if we introduce the projections
	\[
		P = \1(\widetilde{h}_{\eps,s}\le L), \quad Q= 1- P = \1(\widetilde{h}_{\eps,s}> L)
	\]
	for some parameter $L>0$, then we have
	\begin{equation} \label{eq:TrP}
		\Tr P = \Tr \1(\widetilde{h}_{\eps,s}\le L) \le \Tr \left ( \frac{L^q}{\widetilde{h}_{\eps,s}^q}\right) \le C_\eps L^q\,.
	\end{equation}
	
	Let $\beta>0$ be small (depending on $q$). In view of Lemma~\ref{lem:softer}, we consider the Hamiltonian
	\begin{equation} \label{eq:HM-final-def}
		\widetilde{H}_{M}:=\sum_{i=1}^{M} (\widetilde{h}_{\eps,s})_i + \frac{b_{\cM}(V)}{6(M-1)(M-2)} \sum_{\substack{ 1\le i,j,k\le M \\ i\ne j\ne k \ne i}} U_{R}(x_i-x_j, x_i-x_k)
	\end{equation}
	for some $N\ge M \ge (1-\eps)N$ and $N^{-\beta/2} \ge R\ge N^{-\beta}$. Recall that $U_R=R^{{-}6} U(R^{{-1}}\cdot)$ for a fixed function $0\le U\in L^\infty(\R^6)$ satisfying the three\nobreakdash-body symmetry~\eqref{eq:sym} and $\int_{\R^6} U =1$.
	
	Let $\Psi_{M}$ be a ground state for $\widetilde{H}_{M}$ and let us denote $\gamma_M = \KetBra{\Psi_M}{\Psi_M}$. Its $k$\nobreakdash-body density matrix is the operator
	\[
		\gamma_M^{(k)} := \Tr_{k+1\to M} [\gamma_M]
	\]
	on $L^2_s(\R^{3k})$ with kernel
	\begin{multline*}
		\gamma_M^{(k)} (x_1,\ldots,x_k; y_1,\ldots,y_k) = \int_{\R^{3(M-k)}} \Psi_M (x_1,\ldots,x_k, x_{k+1},\ldots,x_M) \\
		\times\overline{\Psi_M (y_1,\ldots,y_k, x_{k+1},\ldots,x_M) } \d x_{k+1}\ldots \d x_M\,.
	\end{multline*}
	Thus, $\gamma_M^{(k)} \ge 0$ and $\Tr \gamma_M^{(k)} = 1$.
	Using this notation, we can write
	\[
		\frac{\widetilde{E}_{M}}{M} := \frac{1}{M} \inf \sigma_{L^2_s(\R^{3M})} \widetilde{H}_{M} = \frac{\pscal{\Psi_M, \widetilde{H}_{M} \Psi_M }}{M} = \Tr \left[\widetilde{h}_{\eps,s} \gamma_M^{(1)}\right] + \frac{b_{\cM}(V)}{6} \Tr \left[ U_R \gamma_M^{(3)}\right].
	\]
	Notice that, since $U_R\geq0$ and $\gamma_M^{(3)}\geq0$, this implies in particular that
	\begin{equation}\label{Trace_one_body_bounded_above}
		\Tr \left[\widetilde{h}_{\eps,s} \gamma_M^{(1)}\right] \leq \frac{\widetilde{E}_{M}}{M} \leq C_\eps\,,
	\end{equation}
	for some constant $C_\eps>0$ independent of $N$.
	
	Now let us impose the finite dimensional cut\nobreakdash-off $P$ and use the quantitative quantum de Finetti theorem. We want to replace $\gamma_{M}^{(3)}$ by some $\tilde\gamma_M^{(3)}$ satisfying $\tilde\gamma_M^{(3)}=P^{\otimes3}\tilde\gamma^{(3)}_MP^{\otimes3}$ that will be chosen later using the quantum de Finetti theorem. We bound from below the one\nobreakdash-body term as follows
	\begin{multline}\label{lwbd_onebody_term}
		3 \Tr \left[\widetilde{h}_{\eps,s} \gamma_M^{(1)}\right] = \Tr \left[\left(\widetilde{h}_1+\widetilde{h}_2+\widetilde{h}_3\right) \gamma_M^{(3)}\right] \\
		\begin{aligned}[b]
			&\geq \Tr \left[\left(\sum\limits_{i=1}^3 \widetilde{h}_i\right) P^{\otimes3}\gamma_M^{(3)}P^{\otimes3}\right] \\
			& \hphantom{\frac{1}{1+\eps}} ={} \Tr \left[\left(\sum\limits_{i=1}^3 \widetilde{h}_i\right) \tilde\gamma_M^{(3)}\right] + \Tr \left[\left(\sum\limits_{i=1}^3 \widetilde{h}_i\right) \left( P^{\otimes3}\gamma_M^{(3)}P^{\otimes3}-\tilde\gamma_M^{(3)}\right)\right] \\
			&\geq \frac{1}{1+\eps} \Tr \left[\left(\sum\limits_{i=1}^3 \widetilde{h}_i\right) \tilde\gamma_M^{(3)}\right] + \Tr \left[\left(\sum\limits_{i=1}^3 \widetilde{h}_i\right) \left( P^{\otimes3}\gamma_M^{(3)}P^{\otimes3}-\tilde\gamma_M^{(3)}\right)\right],
		\end{aligned}
	\end{multline}
	for any $\tilde\gamma_M^{(3)}=P^{\otimes3}\tilde\gamma^{(3)}_MP^{\otimes3}$ and where we used the shortened notation $\widetilde{h}_i:=(\widetilde{h}_{\eps,s})_i$.
	
	For the three\nobreakdash-body term, we apply the Cauchy--Schwarz inequality, with $0<\eps<1$, as follows
	\begin{align*}
		U_R &= \left(\1^{\otimes3}-P^{\otimes3}+P^{\otimes3}\right) U_R \left(\1^{\otimes3}-P^{\otimes3}+P^{\otimes3}\right) \\
			&= \begin{multlined}[t]
				P^{\otimes3} U_R P^{\otimes3} + \left(\1^{\otimes3}-P^{\otimes3}\right) U_R P^{\otimes3} + P^{\otimes3} U_R \left(\1^{\otimes3}-P^{\otimes3}\right) \\
				+ \left(\1^{\otimes3}-P^{\otimes3}\right) U_R \left(\1^{\otimes3}-P^{\otimes3}\right)
			\end{multlined} \\
			&\geq \left(1-\frac{\eps}{1+\eps}\right) P^{\otimes3} U_R P^{\otimes3} + \left(1-\frac{1+\eps}{\eps}\right) \left(\1^{\otimes3}-P^{\otimes3}\right) U_R \left(\1^{\otimes3}-P^{\otimes3}\right) \\
			&\geq \frac{1}{1+\eps} P^{\otimes3} U_R P^{\otimes3} - \eps^{-1} \frac{C}{R^6} \left(\1^{\otimes3}-P^{\otimes3}\right),
	\end{align*}
	where we used for the second inequality that
	\[
		\left(\1^{\otimes3}-P^{\otimes3}\right) U_R \left(\1^{\otimes3}-P^{\otimes3}\right) \leq \norm{U_R}_\infty \left(\1^{\otimes3}-P^{\otimes3}\right)^2 \leq \frac{C}{R^6} \left(\1^{\otimes3}-P^{\otimes3}\right).
	\]
	Moreover, since
	\begin{align*}
		\1^{\otimes3} &= Q\otimes\1^{\otimes2} + P\otimes\1^{\otimes2} =: Q_1 + P\otimes\1^{\otimes2} \\
			&= Q_1 + P\otimes Q\otimes\1 + P\otimes P\otimes Q + P^{\otimes3} \leq Q_1 + Q_2 + Q_3 + P^{\otimes3}\,,
	\end{align*}
	we have
	\[
		U_R \geq \frac{1}{1+\eps} P^{\otimes3} U_R P^{\otimes3} - \eps^{-1} \frac{C}{R^6} \left(Q_1 + Q_2 + Q_3\right).
	\]
	Recall that $\gamma_M^{(3)} := \Tr_{4\to M} [\gamma_M]$ and that $\tilde\gamma_M^{(3)}=P^{\otimes3}\tilde\gamma^{(3)}_MP^{\otimes3}$ by assumption. Now using that
	\begin{align*}
		\Tr \left[ P^{\otimes3} U_R P^{\otimes3} \gamma_M^{(3)} \right] &= \Tr \left[ U_R P^{\otimes3} \gamma_M^{(3)} P^{\otimes3} \right] \\
			&= \Tr \left[ U_R \tilde\gamma_M^{(3)} \right] + \Tr \left[ U_R \left( P^{\otimes3} \gamma_M^{(3)} P^{\otimes3} - \tilde\gamma_M^{(3)}\right) \right]
	\end{align*}
	and
	\[
		\Tr \left[ \left(Q_1 + Q_2 + Q_3\right) \gamma_M^{(3)} \right] = 3\Tr \left[ Q \gamma_M^{(1)} \right] \leq 3 \Tr \left[ \frac{\widetilde{h}_{\eps,s}}{L} \gamma_M^{(1)} \right] \leq \frac{3}{L} \frac{\widetilde{E}_{M}}{M}
	\]
	by~\eqref{Trace_one_body_bounded_above}, we obtain
	\begin{equation}\label{lwbd_threebody_term}
		\Tr \left[ U_R \gamma_M^{(3)}\right] \geq \frac{1}{1+\eps} \Tr \left[ U_R \tilde\gamma_M^{(3)} \right] + \frac{1}{1+\eps}\Tr \left[ U_R \left( P^{\otimes3} \gamma_M^{(3)} P^{\otimes3} - \tilde\gamma_M^{(3)}\right) \right] - \frac{3C}{\eps L R^6} \frac{\widetilde{E}_{M}}{M}\,.
	\end{equation}
	
	Combining~\eqref{lwbd_onebody_term} and~\eqref{lwbd_threebody_term}, we have
	\begin{multline*}
		\frac{\widetilde{E}_{M}}{M} \geq \frac{1}{1+\eps} \Tr \left[ \left(\frac{\widetilde{h}_1+\widetilde{h}_2+\widetilde{h}_3}{3} + \frac{b_{\cM}(V)}{6} U_R \right) \tilde\gamma_M^{(3)}\right] \\
		+ \frac{1}{3} \Tr \left[\left(\sum\limits_{i=1}^3 \widetilde{h}_i\right) \left( P^{\otimes3}\gamma_M^{(3)}P^{\otimes3}-\tilde\gamma_M^{(3)}\right)\right] \\
		+ \frac{1}{1+\eps} \frac{b_{\cM}(V)}{6} \Tr \left[ U_R \left( P^{\otimes3} \gamma_M^{(3)} P^{\otimes3} - \tilde\gamma_M^{(3)}\right) \right] - \frac{C}{\eps L R^6} \frac{\widetilde{E}_{M}}{M}\,,
	\end{multline*}
	hence
	\begin{multline*}
		\left(1+\frac{C}{\eps L R^6}\right)\frac{\widetilde{E}_{M}}{M} \geq \frac{1}{1+\eps} \Tr \left[ \left(\frac{\widetilde{h}_1+\widetilde{h}_2+\widetilde{h}_3}{3} + \frac{b_{\cM}(V)}{6} U_R \right) \tilde\gamma_M^{(3)}\right] \\
		- \left(\frac{\normt{P^{\otimes3}(\widetilde{h}_1+\widetilde{h}_2+\widetilde{h}_3)P^{\otimes3}}}{3} + \frac{C}{R^6} \right) \Tr \left|P^{\otimes3} \gamma_M^{(3)} P^{\otimes3} - \tilde\gamma_M^{(3)}\right|,
	\end{multline*}
	for any $\tilde\gamma_M^{(3)}=P^{\otimes3}\tilde\gamma^{(3)}_MP^{\otimes3}$. We will now use the quantum de Finetti theorem to find such a $\tilde\gamma_M^{(3)}$ to approximate $P^{\otimes3} \gamma_M^{(3)} P^{\otimes3}$ in trace norm. We recall its formulation~\cite[Theorem 3.1]{LewNamRou-16} for the convenience of the reader.
	\begin{theorem}[\textbf{Quantitative quantum de Finetti theorem in finite dimension}]\label{thm:CKMR}
		Let $\cK$ be a finite dimensional Hilbert space and $k\in\N\setminus\{0\}$. For every state $G_k$ on $\cK^k:=\bigotimes_s^k\cK$ and every $p = 1, 2, \dots, k$, we have
		\[
			\Tr_{\cK^p} \left| \Tr_{p+1\to k} [G_k] - \int_{S\cK} \ketbra{ u^{\otimes p} }{ u^{\otimes p} } \d\mu_{G_k} (u) \right| \le \frac{4p \dim \cK}{k} \Tr[ G_k ]\,,
		\]
		where
		\[
			\d\mu_{G_k}(u) :=\dim \cK^k \pscal{u^{\otimes k},G_k u^{\otimes k}} \d u
		\]
		with $\d u$ being the normalized uniform (Haar) measure on the unit sphere $S\cK$.
	\end{theorem}
	Following exactly~\cite[Lemma 3.4]{LewNamRou-16}, which deals with two\nobreakdash-body density matrices, we derive now a localized version for the three\nobreakdash-body density matrices.

	Let $\gH:=L^2(\R^3)$ and define the notations $P_-:=P$ and $P_+:=Q$. Let $\gamma_{N}$ be an arbitrary $N$\nobreakdash-body (mixed) state.
	Then, there exist localized states $G_N^-=G_N$ and $G_N^+=G_N^\perp$ (we use both notations) in the Fock space $\cF(\gH)=\C\oplus\gH\oplus\gH^2\oplus\cdots$ of the form
	\[
		G_N^{\pm} = G_{N,0}^ {\pm} \oplus G_{N,1}^ {\pm} \oplus\cdots\oplus G_{N,N}^ {\pm} \oplus0\oplus\cdots
	\]
	whose reduced density matrices satisfy, for any $0 \leq n \leq N$,
	\begin{equation}\label{eq:localized-DM}
		P_{\pm}^{\otimes n} \gamma^{(n)}_{N} P_{\pm}^{\otimes n} = \left(G_N^{\pm}\right)^{(n)} = \binom{N}{n}^{-1} \sum_{k=n}^N \binom{k}{n}\Tr_{n+1\to k}[G^{\pm}_{N,k}]\,.
	\end{equation}
	As reminded in the aforementioned paper, the relations~\eqref{eq:localized-DM} determine uniquely the localized states $G_N$ and $G_N^\perp$ and ensure that they are (mixed) states on the Fock spaces $\cF (P\gH)$ and $\cF (Q\gH)$, respectively:
	\begin{equation}\label{eq:nomalization-localized-state}
		\sum_{k=0}^N \Tr[ G_{N,k}] = \sum_{k=0}^N \Tr [ G_{N,k}^\perp]=1\,.
	\end{equation}

	We now apply the quantitative de Finetti Theorem~\ref{thm:CKMR} to $\cK=P\gH$, $G_k=G_{N,k}$, and $p=3$. We obtain the following Lemma, already proven in~\cite[Lemma 3.2]{LunRou-15} (see also~\cite[Theorem 3.2]{Girardot-20} for an improved version). We prove it here for the convenience of the reader.
	\begin{lemma}[\textbf{Quantitative quantum de Finetti for the localized state}] \label{deF_localized_state_3body}
		Let $\gamma_{N}$ be an arbitrary $N$\nobreakdash-body (mixed) state. Then,
		\begin{equation} \label{eq:def-mu-N-localized-bound}
			\Tr_{\gH^3} \left| P^{\otimes3} \gamma_{N}^{(3)} P^{\otimes3} - \int_{SP\gH} \KetBra{ u^{\otimes3} }{ u^{\otimes3} } \d\mu_N(u)\right| \leq \frac{12 \Tr P}{N}\,,
		\end{equation}
		where
		\begin{equation} \label{eq:def-mu-N-localized}
			\d\mu_N(u) = \sum_{k=3}^N \frac{k! (N-3)!}{N! (k-3)!} \d\mu_{N,k}(u)\,, \quad \d\mu_{N,k}(u) = \dim (P\gH)_s^k \pscal{u^{\otimes k}, G_{N,k} u^{\otimes k}} \d u\,.
		\end{equation}
	\end{lemma}
	\begin{proof}
		The proof follows the lines of the one in~\cite[Lemma 3.4]{LewNamRou-16}. Applying Theorem~\ref{thm:CKMR} to $\cK=P\gH$, $G_k=G_{N,k}$, and $p=3$, we have
		\begin{multline*}
			\Tr_{(P\gH)^3} \left| \Tr_{4\to k}[G_{N,k}] - \int_{SP\gH} \KetBra{ u^{\otimes3} }{ u^{\otimes3} } \d\mu_{N,k} (u) \right| \\
			\leq \frac{12 \dim(P\gH)}{k} \Tr[G_{N,k}] = \frac{12 \Tr P}{k} \Tr[G_{N,k}]\,,
		\end{multline*}
		where we notice that, on the left-hand side, we can replace $\Tr_{(P\gH)^3}$ by $\Tr_{\gH^3}$.
		Combining this and~\eqref{eq:localized-DM}, the triangle inequality gives
		\begin{multline*}
			\Tr_{\gH^3} \left| P^{\otimes3} \gamma^{(3)}_{N} P^{\otimes3} - \int_{SP\gH} \KetBra{ u^{\otimes3} }{ u^{\otimes3} } \d\mu_N (u) \right| \\
			\leq \sum_{k=3}^N \binom{N}{3}^{-1} \binom{k}{3} \frac{12 \Tr P}{k} \Tr[G_{N,k}] \leq 12 \frac{\Tr P}{N} \sum_{k=3}^N \Tr[G_{N,k}] \leq 12 \frac{\Tr P}{N}\,,
		\end{multline*}
		where the last inequality is due to~\eqref{eq:nomalization-localized-state}.
	\end{proof}
	
	We return to our lower-bound and apply Lemma~\ref{deF_localized_state_3body} to $\gamma_M$. We therefore choose
	\[
		\tilde\gamma_M^{(3)} := \int_{SP\gH} \KetBra{ u^{\otimes3} }{ u^{\otimes3} } \d\mu_M(u)\,,
	\]
	in our ongoing lower bound. It obviously satisfies $P^{\otimes3} \tilde\gamma_M^{(3)} P^{\otimes3} = \tilde\gamma_M^{(3)}$ as required. This yields
	\begin{align}
		\left(1+\frac{C_\eps}{L R^6}\right)\frac{\tilde{E}_{M}}{M} &\geq \frac{(1-\eps)^4}{1+\eps} \int_{SP\gH} \cE_{\eps,s,R}(u) \d\mu_M(u) - C_\eps \left(L + R^{-6} \right) \frac{L^q}{M} \nn \\
			&\geq {e}_{\eps,s,R} \frac{(1-\eps)^4}{1+\eps} \int_{SP\gH}\d\mu_M(u) - C_\eps \left(L + R^{-6} \right) \frac{L^q}{M}\,, \label{eq:LB_EM}
	\end{align}
	where we defined the functional $\cE_{\eps,s,R}$ and the associated groundstate energy $\tilde{e}_{\eps,s,R}$ by
	\[
		e_{\eps,s,R} := \inf\limits_{S\gH} \cE_{\eps,s,R}(u) := \inf\limits_{S\gH} \left( \pscal{u, \widetilde{h}_{\eps,s} u} + \frac{b_{\cM}(V)}{6} \pscal{u^{\otimes3}, U_R u^{\otimes3}} \right)
	\]
	and we used that
	\[
		\Tr \left[ \left(\frac{\widetilde{h}_1+\widetilde{h}_2+\widetilde{h}_3}{3} + \frac{b_{\cM}(V)}{6} U_R \right) \tilde\gamma_M^{(3)}\right] \geq (1-\eps)^4 \int_{SP\gH} \cE_{\eps,R}(u) \d\mu_M(u)\,,
	\]
	where we recall that $b_{\eps} = (1-\eps)^4 b_{\cM}(V)$ and where we have denoted $S\gH:=\left\{u \in \gH\,\left|\, \norm{u}_2=1\right.\right\}$.
	
	To go further, we need to estimate $\mu_N(SP\gH) \simeq 1$ and ${e}_{\eps,s,R}$ by below in terms of $\GPnrg$. For the first part, we use the following lemma (proved at the end of this section).
	\begin{lemma}\label{Lemma_M_limit} Let $\gamma_{N}$ and $\d\mu_N$ be as in~\eqref{eq:def-mu-N-localized}. Then,
		\[
			1\ge \int_{SP\gH}\d\mu_N(u) \ge 1- 3 \Tr (Q \gamma_N^{(1)}) - \frac{ 12 \Tr (P)}{N}\,.
		\]
	\end{lemma}
	\noindent In particular, if $\gamma_M= \ketbra{\Psi_M}{\Psi_M}$ is a ground state for $\widetilde{H}_{M}$ in~\eqref{eq:HM-final-def}, then
	\[
		\Tr (Q \gamma_M^{(1)}) \le L^{-1} \Tr \left ( \widetilde{h}_{\eps,s} \gamma_M^{(1)} \right) \le C_\eps L^{-1}
	\]
	thanks to the kinetic energy bound~\eqref{Trace_one_body_bounded_above}.
	Moreover, recalling that $\Tr P \le C_\eps L^q$ from~\eqref{eq:TrP}, if $L$ is chosen such that $1\ll L \ll N^{1/q}$, then Lemma~\ref{Lemma_M_limit} gives
	\[
		1\ge \int_{SP\gH}\d\mu_N(u) \ge 1- C_\eps L^{-1} - C_\eps \frac{L^q}{N} \ge 1 + o(1)\,.
	\]
	
	We now deal with ${e}_{\eps,s,R}$. The first step is to approximate the interaction term. To this purpose, we define the functional $\cE_{\eps,s}$ and its groundstate energy $e_{\eps,s}$ by
	\[
		e_{\eps,s} := \inf\limits_{\norm{u}_2=1} \cE_{\eps,s}(u) := \inf\limits_{\norm{u}_2=1} \pscal{u, \widetilde{h}_{\eps,s} u} + \frac{b_{\cM}(V)}{6} \norm{u}_{6}^6 \quad \text{for all } s>1>\eps\geq0\,,
	\]
	and we use this other lemma (also proved at the end of this section).
	\begin{lemma}\label{lemma_R_limit}
		For any $\eps, s>0$, there exists a constant $C_{\eps,s}\geq0$ such that
		\[
			{e}_{\eps,s,R} \geq -C_{\eps,s} R +e_{\eps,s} \,.
		\]
	\end{lemma}
	Hence, using that $e_{\eps,s} \leq C - \kappa_{\eps,s}$, with $C$ independent of $\eps$ and $s$,~\eqref{eq:LB_EM} becomes
	\[
		\frac{\tilde{E}_{M}}{M} \geq e_{\eps,s} - (C -\kappa_{\eps,s}) \eps - C_\eps(L^{-1} R^{-6} + L^{q+1} M^{-1} + L^q R^{-6} M^{-1})\,.
	\]
	Recall that $\tilde{E}_{M} = \inf \sigma_{L^2_s(\R^{3M})} \widetilde{H}_{M}$, defined in~\eqref{eq:HM-final-def}. Using Lemma~\ref{lem:softer}, $M \geq (1-\eps)N$, and Lemma~\ref{lemma_R_limit}, we obtain
	\begin{align*}
		\frac{E_N}{N} &\geq (1-\eps)\frac{\tilde{E}_{M}}{M} + (1-\eps)\kappa_{\eps,s} - C_{\eps,s,\delta} R^{2/7} - (C -\kappa_{\eps,s})\eps - \delta C_{\eps}\\
			&\geq \begin{multlined}[t][0.9\textwidth]
				(1-\eps)(e_{\eps,s} + \kappa_{\eps,s}) - (C -\kappa_{\eps,s})\eps - C_{\eps} \delta \\
				- C_{\eps,s,\delta,\beta}(L^{-1} R^{-6} + L^{q+1} M^{-1} + L^q R^{-6} M^{-1} + R^{2/7}) \quad \text{for all } \delta>0\,.
			\end{multlined}
	\end{align*}
	Recall that $N\geq M \geq (1-\eps)N$ and $N^{-\beta/2} \geq R\geq N^{-\beta}$. Therefore, choosing $L = N^{1/(q+2)}$, we obtain
	\begin{multline*}
		\frac{E_N}{N} \geq (1-\eps)(e_{\eps,s} + \kappa_{\eps,s}) - (C -\kappa_{\eps,s})\eps \\
		- C_{\eps} \delta - C_{\eps,s,\delta,\beta}( N^{6\beta-\tfrac{1}{q+2}} + N^{\tfrac{-1}{q+2}} + N^{6\beta - \tfrac{2}{q+2}}+ N^{-\beta/7})\,.
	\end{multline*}
	Taking $0 < \beta < 1/(6(q+2))$, we obtain
	\[
		\liminf_{N \to \infty} \frac{E_N}{N} \geq (1-\eps)(e_{\eps,s} + \kappa_{\eps,s}) - (C -\kappa_{\eps,s})\eps - C_{\eps} \delta \quad \text{for all } \delta>0\,.
	\]
	In particular, we can take $\delta \to 0$ in order to remove the last term above. Finally, by a standard compactness argument (see, e.g.,~\cite{LieSei-06} or~\cite[Sect.~4B]{NamRouSei-16}), we have $\lim_{s\to \infty} \kappa_{\eps,s} = \inf \sigma(h) -1$ and
	\[
		\lim_{\eps \to 0} \lim_{s\to \infty} e_{\eps,s} + \kappa_{\eps,s} = \GPnrg\,.
	\]
	This finishes the proof of Theorem~\ref{theo:lower_bound} up to the proofs of Lemmas~\ref{Lemma_M_limit} and~\ref{lemma_R_limit} that we give below, concluding this section.
\end{proof}

\begin{proof}[Proof of Lemma~\ref{Lemma_M_limit}]
	Note that, from~\eqref{eq:def-mu-N-localized}, we obtain $\int \d \mu_{N} \le 1$ since every measure $\mu_{N,k}$ is normalized. It remains to prove the lower bound. From~\eqref{eq:def-mu-N-localized-bound} and the triangle inequality, we have
	\[
		\int \d \mu_{N} \ge \Tr ( P^{\otimes3} \gamma_{N}^{(3)} P^{\otimes3} ) - \frac{ 12 \Tr (P)}{N}\,.
	\]
	On the other hand, by the cyclic property of the trace, we can decompose
	\begin{align*}
		1&=\Tr ( \gamma_{N}^{(3)} ) = \Tr ( (P_1+Q_1) \gamma_{N}^{(3)} (P_1+Q_1))= \Tr ( P_1 \gamma_{N}^{(3)} P_1) + \Tr ( Q_1 \gamma_{N}^{(3)} Q_1) \\
		&= \Tr ( P_1 (P_2+Q_2) \gamma_{N}^{(3)} (P_2+Q_2)P_1) + \Tr ( Q_1 \gamma_{N}^{(3)} Q_1) \\
		&= \Tr ( P_1 P_2 \gamma_{N}^{(3)} P_2 P_1) + \Tr ( P_1 Q_2 \gamma_{N}^{(3)} Q_2 P_1) + \Tr ( Q_1 \gamma_{N}^{(3)} Q_1)\\
		&= \begin{multlined}[t]
			\Tr ( P_1 P_2 P_3 \gamma_{N}^{(3)} P_3 P_2 P_1) + \Tr ( P_1 P_2 Q_3 \gamma_{N}^{(3)} Q_3 P_2 P_1) + \Tr ( P_1 Q_2 \gamma_{N}^{(3)} Q_2 P_1) \\
			+ \Tr ( Q_1 \gamma_{N}^{(3)} Q_1)
		\end{multlined} \\
		&\le \Tr ( P^{\otimes3} \gamma_{N}^{(3)} P^{\otimes3} ) + 3 \Tr (Q \gamma_N^{(1)})\,. \qedhere
	\end{align*}
\end{proof}

	We now the give proof of Lemma~\ref{lemma_R_limit}, which is an adaptation of the one of~\cite[Lemma 4.1]{LewNamRou-16}.
\begin{proof}[Proof of Lemma~\ref{lemma_R_limit}]
	We have
	\[
		\cE_{\eps,s,R}(u) - \cE_{\eps,s}(u) = \frac{b_{\cM}(V)}{6} \left( \pscal{u^{\otimes3}, U_R u^{\otimes3}} - \norm{u}_{6}^6 \right).
	\]
	Rewriting $\pscal{u^{\otimes3}, U_R u^{\otimes3}}$ by a change of variables, we have
	\begin{align*}
		\pscal{u^{\otimes3}, U_R u^{\otimes3}} &= \int_{\R^9} |u(x)|^2 |u(y)|^2 |u(z)|^2 R^{-6} U(R^{-1}(x-y),R^{-1}(x-z)) \d x \d y \d z \\
			&= \int_{\R^9} |u(x)|^2 |u(x-Ry)|^2 |u(x-Rz)|^2 U(y,z) \d x \d y \d z
	\end{align*}
	and, since $\int_{\R^6} U=1$, we obtain
	\begin{multline*}
		\pscal{u^{\otimes3}, U_R u^{\otimes3}} - \norm{u}_{6}^6 \\
		= \int_{\R^9} |u(x)|^2 U(y,z) \left( |u(x-Ry)|^2 |u(x-Rz)|^2 - |u(x)|^4 \right) \d x \d y \d z\,.
	\end{multline*}
	We now write the term in parenthesis as an integral:
	\begin{align*}
		|u(x-Ry)|^2 |u(x-Rz)|^2 - |u(x)|^4 &= \int_0^1 \frac{\d}{\d t}\left(|u(x-tRy)|^2 |u(x-tRz)|^2\right) \d t \\
			&=\begin{aligned}[t]
				&\int_0^1 \nabla|u|^2(x-tRy)\cdot Ry \, |u(x-tRz)|^2 \d t \\
				&+ \int_0^1 |u(x-tRy)|^2 \nabla( |u|^2)(x-tRz)\cdot Rz \d t\,.
			\end{aligned}
	\end{align*}
	Using that
	\[
		\int_{\R^3} |u(x-tRy)|^2 \left|\nabla|u|^2(x-tRz)\right| \d x \leq 2 \norm{u}_{6}^3 \norm{\nabla |u|}_{2}
	\]
	for all $z,y \in \mathbb{R}^{3}$, $t\in (0,1)$, and $R>0$, we obtain
	\[
		\left|\pscal{u^{\otimes3}, U_R u^{\otimes3}} - \norm{u}_{6}^6\right| \leq C R \norm{u}_{2}^2 \norm{\nabla u}_{2}^4 \norm{z\norm{U(\cdot,z)}_{L^\infty(\R^3)}}_{1}\,.
	\]
	We used above the Sobolev inequality and that $\norm{\nabla |u|}_{2} \leq \norm{\nabla u}_{2}$, see~\cite[Theorem 7.8]{LieLos-01}.
	Therefore, we have
	\[
		| \cE_{\eps,s}(u) - \cE_{\eps,s,R}(u)| \leq C R \left( 1+ \norm{\nabla u}_{2}^4 \right) \norm{z\norm{U(\cdot,z)}_{L^\infty(\R^3)}}_1 \quad \text{for all } u\in S\gH\,,
	\]
	and, since $\norm{\nabla u}_{2}^2\leq 2\eps^{-1}\cE_{\eps,s,R}(u)$, we obtain
	\[
		\inf\limits_{S\gH} \cE_{\eps,s} \leq \cE_{\eps,s}(u) \leq \cE_{\eps,s,R}(u) + C R \left(1+ 2\eps^{-1}\cE_{\eps,s,R}(u)\right)^2 \quad \text{for all } u\in S\gH\,.
	\]
	Applying this to a minimizing sequence $\{u_n\}_n\subset S\gH$ for ${e}_{\eps,s,R}$ and passing to the limit gives
	\[
		\inf\limits_{S\gH} \cE_{\eps,s} \leq {e}_{\eps,s,R} + C R \left(1+ 2\eps^{-1}{e}_{\eps,s,R}\right)^2 \leq {e}_{\eps,s,R} + C_{\eps,s} R \quad \text{for all } u\in SP\gH\,,
	\]
	where we used the simple estimate ${e}_{\eps,s,R} \leq C$ independently of $R,s,\eps>0$.
\end{proof}

\section{Energy upper bound}\label{sec:upper}
	Let us recall that $V_N:=NV(N^{1/2}\cdot)$ and that $\omega = 1-f$, where $f$ is the solution to the scattering equation of~$V$, satisfies, for all $\bx\in\R^6$, the estimates
\begin{equation}\label{eq:bounds_on_w}
	0\leq \omega(\bx) < 1\,, \quad \omega(\bx) \leq \frac{C}{|\bx|^4+1}\,, \quad \text{and} \quad |\nabla \omega(\bx)| \leq \frac{C}{|\bx|^5+1}\,.
\end{equation}
We define $\omega_N := \omega(N^{1/2} \cdot)$ and $f_N:=f(N^{1/2} \cdot)$.

	Let us also recall, for $M \in \N\setminus\{0\}$, the notation
\begin{equation}\label{Def_H_MN}
	H_{M,N} = \sum_{j=1}^M h_j + \sum_{1\leq i<j<k \leq M} N V(N^{1/2}(x_i-x_j, x_i-x_k))
\end{equation}
and the definition
\begin{align*}
	\mathbb{H}_N :={}& 0\oplus \bigoplus_{M\geq 1} H_{M,N} \\
		={}& \int_{\R^{3}} a^*_x h_x a_x \d x + \frac{1}{6} \int_{(\R^{3})^3} V_N(x-y,x-z) a^*_x a^*_y a^*_z a_x a_y a_z \d x \d y \d z\,.
\end{align*}

For $f \in L^2(\R^{3})$, we define the Weyl operator
\[
	W(f) = \exp (a^*(f) - a(f))
\]
which is a unitary operator and, for $g\in L^2(\R^{3})$, satisfies
\[
	W(f)^* a(g) W(f) = a(g) + \pscal{ g, f } \quad \text{and} \quad W(f)^* a^*(g) W(f) = a^*(g) + \pscal{ f, g }.
\]
We also define, for $\varphi \in H^2(\R^{3})$, $\cB \equiv \cB[\varphi] \in L^2(\mathbb{R}^{9})$ as
\[
	\cB(x,y,z) = \omega_N(x,y,z) \varphi(x)\varphi(y) \varphi(z)
\]
and $B_1 \equiv B_1[\varphi]$ as
\begin{equation}\label{Def_B1}
	B_1 = -\frac{1}{6} N^{\frac32} \int_{(\R^{3})^3} \cB(x,y,z) a_x a_y a_z \d x \d y \d z\,.
\end{equation}

Finally, we denote
\begin{equation}\label{Def_B}
	\Theta \equiv \Theta(\cN) := \1_{[0,N^{1/2}]}(\cN)\,, \quad B := \Theta B_1^* - B_1 \Theta \,, \quad \text{and} \quad U_N := e^{-B}\,.
\end{equation}

This section is devoted to the proof of the following theorem and its corollary.
\begin{theorem}[Energy upper bound]\label{theo:upper_bound_fock_space}
	Let $\varphi \in D(h) \cap L^\infty(\mathbb{R}^{3})$ with $\norm{\varphi}_{2} = 1$. Then, there exists $C_{\varphi}>0$, depending only on $\norm{h \varphi}_{2}$ and $\norm{\varphi}_{\infty}$, such that
	\[
		\pscal{\Omega, U_N^* W(\sqrt N \varphi)^* \mathbb{H}_N W(\sqrt N \varphi) U_N \Omega} \leq N \Egp(\varphi) + C_{\varphi} N^{1/2}\,,
	\]
	where $\Omega$ is the vacuum.
\end{theorem}
\begin{corollary}\label{Cor_upper_bound}
	There exists a constant $C>0$, independent of $N$, such that
	\[
		E_N \leq N \GPnrg + C N^{2/3}\,.
	\]
\end{corollary}

	As a convention for this section, the constants $C$ only depend on $\norm{h \varphi}_{2}$ and $\norm{\varphi}_{\infty}$. Also, note that since $h \geq \Vext \geq 1$, the diamagnetic inequality gives $\norm{h \varphi}_{2} \geq \norm{h^{1/2} \varphi}_{2} \geq \norm{ \nabla| \varphi| }_{2}$.

\subsection{The transformation \texorpdfstring{$U_N$}{UN}}
	Let $p \in \N\setminus\{0\}$, denote $\sharp = (\sharp_1,\dots,\sharp_p) \in \{\cdot,^*\}^p$, and define $\kappa(\sharp) = 2 \#\{i | \sharp_i = \, ^* \} - p $. Then,
\[
	a^{\sharp_1}_{x_1} \dots a^{\sharp_p}_{x_p} \{\cN = n\} \subset \{\cN = n+\kappa(\sharp)\}\,, \quad n= 0, 1, 2, \dots\,.
\]
Therefore, we have $U_N \Omega \in \{\cN \in 3 \mathbb{N}_0\}$ and if $\kappa(\sharp) \notin 3 \mathbb{Z}_0$, then
\begin{equation}\label{eq:coprime_with_3}
	\braket{ \Omega, U_N^*a^{\sharp_1}_{x_1} \dots a^{\sharp_p}_{x_p} U_N \Omega } = 0\,.
\end{equation}
For example, we have $\braket{ \Omega, U_N^*a_{x_1} U_N \Omega } = \braket{ \Omega, U_N^*a_{x_1} a_{x_2} U_N \Omega } = 0$. We now state two lemmas, which will be used in the proof of Theorem~\ref{theo:upper_bound_fock_space}.
\begin{lemma}\label{lem:commutator_B_B_star}
	Let $\varphi \in C^\infty(\R^{3})$ and $B_1$ be as in~\eqref{Def_B1}. Then, there exists $C_{\varphi}>0$, depending only on $\norm{h \varphi}_{2}$ and $\norm{\varphi}_{\infty}$, such that
	\[
		\forall\, N \in \N\setminus\{0\}, \quad \pm \left([B_1,B_1^*] - \mathcal{Q}(\varphi)\right) \leq C_\varphi \frac{\cN^2}{N}\,,
	\]
	where
	\[
		\mathcal{Q}(\varphi) = \frac{1}{24} \int_{\R^{9}} N^{3} \omega_N(x,y,z)^2 |\varphi(x)|^2|\varphi(y)|^2 |\varphi(z)|^2 \d x \d y \d z \geq 0\,.
	\]
	Moreover, there exists $C>0$ such that
	\[
		\forall\, N \in \N\setminus\{0\}\,,\quad \forall\, \varphi \in C^\infty(\R^{3})\,, \quad \mathcal{Q}(\varphi) \leq C \norm{\varphi}_{\infty}^4 \norm{\varphi}_{2}^2 \norm{\omega}_{L^{2}(\R^6)}^2\,.
	\]
\end{lemma}

\begin{lemma}\label{lem:nb_particule}
	Define $B$ as in~\eqref{Def_B}. Then, for all $k\in \N\cup\{0\}$, there exists $C_k>0$ such that
	\begin{equation}\label{eq:lem_nb_particule}
		\forall\, (\lambda,N) \in [-1,1]\times (\N\setminus\{0\})\,, \quad \pscal{\Omega, e^{-\lambda B} \cN^k e^{\lambda B} \Omega} \leq C_k\,.
	\end{equation}
	In particular, for all $k\in \N\cup\{0\}$, there exists $C_k>0$ such that
	\begin{equation}\label{eq:lem_nb_particule_U_N}
		\forall\, N \in \N\setminus\{0\}\,, \quad \pscal{\Omega, U_N^* \cN^k U_N \Omega} \leq C_k\,.
	\end{equation}
\end{lemma}

\begin{proof}[Proof of Lemma~\ref{lem:commutator_B_B_star}]
	On the bosonic Fock space $\cF(\gH)$, we have
	\begin{equation}\label{eq:wick_3_3}
		[a_x a_y a_z, a^*_{x'}a^*_{y'}a^*_{z'}] = 6 \delta_{x=x'} \delta_{y=y'} \delta_{z=z'} + 18 a^*_x a_{x'} \delta_{y=y'} \delta_{z=z'} + 9 a^*_x a^*_y a_{x'}a_{y'} \delta_{z=z'}\,.
	\end{equation}
	Hence,
	\begin{align*}
		[B_1, B^*_1] ={} &\frac{1}{24} \int_{\R^{9}} N^{3} \omega_N(x,y,z)^2 |\varphi(x)|^2|\varphi(y)|^2 |\varphi(z)|^2 \d x \d y \d z\\
			&+ \frac{1}{8}
				\begin{multlined}[t]
					\int_{\R^{12}} N^{3/2}\omega_N(x,y,z) \\
					\times N^{3/2}\omega_N(x',y,z) \varphi(x) \overline{\varphi(x')}|\varphi(y)|^2 |\varphi(z)|^2 a^*_x a_{x'} \d x \d x' \d y \d z
				\end{multlined} \\
			&+ \frac{1}{16} \int_{\R^{15}}
				\begin{multlined}[t]
					N^{3/2} \omega_N(x,y,z) N^{3/2}\omega_N(x',y',z) \\
					\times \varphi(x) \overline{\varphi(x')} \varphi(y) \overline{\varphi(y')} |\varphi(z)|^2 a^*_x a^*_y a_{x'}a_{y'} \d x \d x' \d y \d y' \d z
				\end{multlined} \\
			={} & (\mathrm{I}) + (\mathrm{I}) + (\mathrm{III})\,,
	\end{align*}
	which defines the terms $(\mathrm{I})$, $(\mathrm{I})$, and $(\mathrm{III})$.
	In order to estimate these terms, let us first define
	\[
		T(f)(z) = \varphi(z) \int_{\R^6} N^{\frac{3}{2}} \omega_N(x,y,z) \varphi(x) \varphi(y) f(x,y) \d x \d y
	\]
	for any $(z,f) \in \R^3\times L^2(\R^{6})$.
	The terms $(\mathrm{I})$, $(\mathrm{I})$, and $(\mathrm{III})$ are respectively proportional to $\norm{T}_{\gS^2(L^2(\mathbb{R}^{6}))}^2$, to the second quantization of $TT^*$, and to the second quantization of $T^*T$.
	We have
	\[
		\norm{T}_{\gS^2(L^2(\mathbb{R}^{6}))}^2 \leq C \norm{\varphi}_{\infty}^4 \norm{\varphi}_{2}^2 \norm{ N^{3/2} \omega_N }_{L^{2}(\R^6)}^2 \leq C \norm{\varphi}_{\infty}^4 \norm{\varphi}_{2}^2 \norm{\omega}_{L^{2}(\R^6)}^2\,,
	\]
	proving the statement about $Q(\varphi) =: (\mathrm{I})$. Note that, contrarily to the two\nobreakdash-body interaction case, here we have $\omega \in L^2(\R^{6})$. To bound the other terms, let us compute~$\norm{T}_{op}$. Using~\eqref{eq:bounds_on_w}, for all $z\in \mathbb{R}^{3}$, we have
	\begin{multline*}
		|T(f)(z)| \\
		\begin{aligned}
			&\leq \frac{C}{N^{\frac{1}{2}}} |\varphi(z)| \int_{\R^{3}} |\varphi(y)| \left(\int_{\R^{3}} \frac{\varphi(x) |f(x,y)|}{(|x-y|^2 + |y-z|^2)^2} \d x \right) \d y \\
			&\leq \frac{C}{N^{\frac{1}{2}}} |\varphi(z)| \int_{\R^{3}} |\varphi(y)| \left(\int_{\R^{3}} \frac{\d u}{(u^2 + |y-z|^2)^4}\right)^{\frac{1}{2}} \left(\int_{\R^{3}} |\varphi(x)|^{2} |f(x,y)|^2 \d x \right)^{\frac{1}{2}} \d y \\
			&\leq \frac{C}{N^{\frac{1}{2}}} \norm{\varphi}_{\infty}^2 |\varphi(z)| \int_{\R^{3}} \frac{1}{|y-z|^{5/2}} \norm{f(\cdot,y)}_{L^2(\R^{3})} \d y\,.
		\end{aligned}
	\end{multline*}
	Hence, by the Hardy--Littlewood--Sobolev inequality, we obtain
	\[
		\pscal{ g, T(f) }_{L^2(\R^{3})} \leq \frac{C}{N^{\frac{1}{2}}} \norm{\varphi}_{\infty}^2 \norm{\varphi}_{6} \norm{g}_{2} \norm{f}_{L^2(\R^{6})} \quad \text{for all } g\in L^2(\mathbb{R}^{3})\,.
	\]
	Therefore, $\norm{T}_{op} \leq C_{\varphi} N^{-1/2}$ and we obtain
	\[
		(\mathrm{II}) \leq C \norm{T}_{op}^2 {\cN} \leq C_{\varphi} N^{-1} \cN \quad \text{and} \quad (\mathrm{III}) \leq C \norm{T}_{op}^2 \cN(\cN-1) \leq C_{\varphi} N^{-1} \cN^2\,. \qedhere
	\]
\end{proof}

\begin{proof}[Proof of Lemma~\ref{lem:nb_particule}]
	The case $k=0$ is immediate.
	
	Let $k\in \N\setminus\{0\}$, $\xi(\cN) := \cN^k$, and
	\[
		\partial \xi = \xi (\cdot+ 3) - \xi\,.
	\]
	Then, $\xi$ is such that $\xi(0) = 0$ and $\partial^j \xi \geq 0$ for $j\geq 0$. Using the Duhamel formula, we have
	\[
		e^{-\lambda B} \xi(\cN) e^{\lambda B} = \xi(\cN) - \lambda [B,\xi(\cN)] + \lambda^2 \int_0^1 \int_0^s e^{- \lambda uB} [B,[B,\xi(\cN)]] e^{\lambda uB} \d u \d s\,.
	\]
	We are only interested in the expectation on the vacuum $\Omega$, hence, since $\xi(0)=0$, only the third term will give a non zero contribution. Using that $a_x \cN = (\cN+1)a_x$, we obtain
	\[
		- [B,\xi(\cN)] = [ B_1 \Theta ,\xi(\cN)] + h.c. = \partial \xi (\cN) B_1\Theta + h.c.
	\]
	and
	\begin{align}\label{eq:double_commutator}
		[B,[B,\xi(\cN)]] &= \left[B_1 \Theta - \Theta B_1^*, \partial \xi (\cN) B_1\Theta \right] + h.c.\nn \\
			&\begin{multlined}[b][.67\textwidth]
				=\partial^2 \xi (\cN) \left(B_1 \Theta \right)^2 - \Theta B_1 ^* \partial^2 \xi (\cN) B_1\Theta \\
				+ \partial \xi (\cN) \left[B_1 \Theta , \Theta B_1 ^*\right] + h.c.
			\end{multlined}
	\end{align}
	
	The last term, which is the most regular, is controlled using that $\partial \xi (\cN)$ is nonnegative and commutes with $[B_1 \Theta , \Theta B_1 ^*]$, and that
	\begin{align*}
		\left[B_1 \Theta , \Theta B_1 ^*\right] &= \partial \Theta B_1^* B_1 + \Theta(\cN+3) [B_1,B_1^*] \nn\\
			&\leq \Theta(\cN+3) [B_1,B_1^*] \leq \Theta(\cN+3) \left(\mathcal{Q}(\varphi) + C \frac{\cN^2}{N}\right) \leq C \Theta(\cN+3)\,,
	\end{align*}
	where we used that $\Theta$ and $B_1^* B_1$ commute, that $\partial \Theta \leq 0$ for the first inequality, and Lemma~\ref{lem:commutator_B_B_star} for the second inequality.
	
	The first two terms in~\eqref{eq:double_commutator} are controlled using the Cauchy--Schwarz inequality:
	\begin{multline*}
		\begin{multlined}[t][.9\textwidth]
			\partial^2 \xi (\cN) (B_1 \Theta )^2 - \Theta B_1 ^* \partial^2 \xi (\cN) B_1\Theta + h.c \\
			=(B_1 \Theta )\partial^2 \xi (\cN-3) (B_1 \Theta ) - 2 \Theta B_1 ^* \partial^2 \xi (\cN) B_1\Theta
		\end{multlined} \\
		\begin{aligned}
			&\leq
			\begin{multlined}[t][.9\textwidth]
				(B_1 \Theta ) \partial^2 \xi (\cN-3) (\Theta B_1^*) + (\Theta B_1 ^*)\partial^2 \xi (\cN-3) (\Theta B_1) - 2 \Theta B_1 ^* \partial^2 \xi (\cN) B_1\Theta \\
				\begin{aligned}[t]
					&= \partial^2 \xi (\cN)[B_1 \Theta ,\Theta B_1^*] + \left( \partial^2\xi (\cN) + \partial^2\xi (\cN-6) - 2\partial^2\xi (\cN-3) \right) \Theta B_1 ^* B_1\Theta \\
					&= \partial^2 \xi (\cN)[B_1 \Theta ,\Theta B_1^*] + \partial^4\xi (\cN-6) \Theta B_1 ^* B_1\Theta
				\end{aligned}
			\end{multlined} \\
			&\leq C \partial^2 \xi (\cN) + C \partial^4\xi (\cN-6) (\cN+1)^3\,.
		\end{aligned}
	\end{multline*}
	
	Gathering the last three inequalities, we obtain
	\[
		[B,[B,\xi(\cN)]] \leq C \left(\partial \xi (\cN)+ \partial^2 \xi (\cN)\right) + C \partial^4\xi (\cN-6)(\cN+1)^3\,.
	\]
	Choosing $\xi(\cN) = \cN$, the above inequality shows that~\eqref{eq:lem_nb_particule} holds for $k=1$. For any integer $k\geq 2$, we have $\xi(\cN) = \cN^{k}$ and, using that $\partial^{i} \xi$ is a linear combination of~$\cN^{j}$ for~$i\in \N\setminus\{0\}$ and~$0\leq j \leq k-i$, one concludes the proof by induction.
\end{proof}

\subsection{Proof of Theorem~\ref{theo:upper_bound_fock_space}}
	We first conjugate $\mathbb{H}_N$ with the Weyl operator, we obtain
\[
	W(\sqrt N \varphi)^* \d \Gamma(h) W(\sqrt N \varphi) = \d \Gamma(h) + \sqrt N a^*\left(h \varphi\right) + \sqrt N a\left(h \varphi\right) + N \pscal{\varphi,h\varphi}
\]
and
\[
	\frac{1}{6} \int_{\R^{9}} V_N(x-y,x-z) W(\sqrt N \varphi)^* a^*_x a^*_y a^*_z a_x a_y a_z W(\sqrt N \varphi) \d x \d y \d z = \sum_{i=0}^6\cL_i\,,
\]
where
\begin{align*}
	\cL_6 &= \frac{1}{6} \int_{\R^{9}} V_N(x-y,x-z) a^*_x a^*_y a^*_z a_x a_y a_z \d x \d y \d z\,, \\
	\cL_5 &= \frac{N^{1/2}}{2} \int_{\R^{9}} V_N(x-y,x-z) a^*_x a^*_y a^*_z a_x a_y \varphi(z) \d x \d y \d z + h.c.\,, \\
	\cL_4 &=
		\begin{multlined}[t][0.9\textwidth]
			\frac{N}{2} \int_{\R^{9}} V_N(x-y,x-z) a^*_x a^*_y a^*_z a_x \varphi(y) \varphi(z) \d x \d y \d z + h.c. \\
			+\frac{N}{2} \int_{\R^{9}} V_N(x-y,x-z) a^*_x a^*_y a_x a_z \varphi(y) \overline{\varphi(z)} \d x \d y \d z + h.c. \\
			+\frac{N}{2} \int_{\R^{9}} V_N(x-y,x-z) a^*_x a^*_y a_x a_y |\varphi(z)|^2 \d x \d y \d z
		\end{multlined} \\
	&=: \cL_4^{(1)} + \cL_4^{(2)} + \cL_4^{(3)}\,, \\
	\cL_3 &=
		\begin{multlined}[t][0.9\textwidth]
			\frac{N^{3/2}}{6} \int_{\R^{9}} V_N(x-y,x-z) a^*_x a^*_y a^*_z \varphi(x)\varphi(y)\varphi(z) \d x \d y \d z + h.c. \\
			+\frac{N^{3/2}}{2} \int_{\R^{9}} V_N(x-y,x-z) a^*_x a^*_y a_x \varphi(y) |\varphi(z)|^2 \d x \d y \d z + h.c. \\
			+\frac{N^{3/2}}{2} \int_{\R^{9}} V_N(x-y,x-z) a^*_x a^*_y a_z \varphi(x)\varphi(y)\overline{\varphi(z)} \d x \d y \d z + h.c.
		\end{multlined} \\
	&=: \cL_3^{(1)} + \cL_3^{(2)} + \cL_3^{(3)}\,, \\
	\cL_2 &=
		\begin{multlined}[t][0.9\textwidth]
			\frac{N^2}{2} \int_{\R^{9}} V_N(x-y,x-z) a^*_x a^*_y \varphi(x)\varphi(y)|\varphi(z)|^2 \d x \d y \d z + h.c. \\
			+\frac{N^2}{2} \int_{\R^{9}} V_N(x-y,x-z) a^*_x a_x |\varphi(y)|^2|\varphi(z)|^2 \d x \d y \d z \\
			+ N^2 \int_{\R^{9}} V_N(x-y,x-z) a^*_x a_y \varphi(x)\overline{\varphi(y)} |\varphi(z)|^2 \d x \d y \d z
		\end{multlined} \\
	&=: \cL_2^{(1)} + \cL_2^{(2)} + \cL_2^{(3)}, \\
	\cL_1 &= \frac{N^{5/2}}{2} \int_{\R^{9}} V_N(x-y,x-z) a^*_x \varphi(x)|\varphi(y)|^2|\varphi(z)|^2 \d x \d y \d z + h.c.\,, \\
	\intertext{and}
	\cL_0 &= \frac{N}{6} \int_{\R^{9}} N^{2}V_N(x-y,x-z) |\varphi(x)|^2|\varphi(y)|^2|\varphi(z)|^2 \d x \d y \d z\,.
\end{align*}

Gathering up the constant terms, we obtain
\[
	N \left\{ \pscal{ \varphi, h\varphi }_{L^2(\mathbb{R}^{3})}+ \frac{1}{6} \int_{\R^9} N^3V\left(N^{\frac{1}{2}}(x-z,y-z)\right) |\varphi(x)|^2|\varphi(y)|^2|\varphi(z)|^2 \d x \d y \d z \right\},
\]
which is the mean-field energy associated to the particle $\varphi$. This quantity does not take into account the contribution from the scattering process. To obtain it, one must include contributions hidden in the terms $\d \Gamma (h)$, $\cL_3^{(1)}$, and $\cL_6$. We will first prove that the other terms are negligible and then extract the main contribution from the aforementioned terms.

\subsubsection{Controlling the error terms}
	Using~\eqref{eq:coprime_with_3}, we obtain
\[
	\pscal{\Omega, e^{-B} A e^{B} \Omega} = 0 \quad \text{for all } A\in \left\{ a^*\left(h \varphi\right), \cL_1, \cL_2^{(1)}, \cL_3^{(2)}, \cL_3^{(3)}, \cL_4^{(1)}, \cL_5 \right\}
\]
because they create or annihilate a number of particles which is coprime with $3$. It remains to bound the following terms.

\begin{lemma}\label{lem:L_2_4}
	There exists $C_\varphi>0$, depending only on $\norm{\varphi}_{\infty}$, such that
	\[
		\pscal{\Omega, e^{-B} (\cL_2^{(2)} + \cL_2^{(3)}) e^{B} \Omega} \leq C \quad \text{and} \quad \pscal{\Omega, e^{-B} (\cL_4^{(2)}+ \cL_4^{(3)}) e^{B} \Omega} \leq C_\varphi N^{1/2}\,.
	\]
\end{lemma}
\begin{proof}
	We prove
	\begin{align*}
		\cL_2^{(2)} + \cL_2^{(3)} &\leq C \norm{V}_{L^1(\R^6)} \norm{\varphi}_{\infty}^2 \norm{\varphi}_{2}^2 \cN
		\intertext{and}
		\cL_4^{(2)} + \cL_4^{(3)} &\leq C \sup_{x\in \R^{3} } \norm{V(\cdot,x)}_{L^1(\R^{3})} \norm{\varphi}_{\infty}^2 N^{1/2} \cN^2\,,
	\end{align*}
	since the result then follows by~\eqref{eq:lem_nb_particule_U_N}. The first bound comes from
	\[
		\cL_2^{(2)} = \frac{1}{2} \int_{\mathbb{R}^{3}} \braket{ V(x-\cdot,x-\cdot) }_{\varphi^{\otimes 2}} a^*_x a_x \d x \leq C \norm{V}_{L^1(\R^{6})} \norm{\varphi}_{\infty}^4 \cN
	\]
	and, using the Cauchy--Schwarz inequality,
	\begin{align*}
		\cL_2^{(3)} &\leq \int_{\R^{9}} N^{3 }V(N^{1/2}(x-y,x-z)) \left( a^*_x a_x |\varphi(y)|^2 + a^*_y a_y |\varphi(x)|^2 \right) |\varphi(z)|^2 \d x \d y \d z \\
			&\leq C \norm{V}_{L^1(\R^{6})} \norm{\varphi}_{\infty}^4 \cN\,.
	\end{align*}
	
	For the second bound, on the one hand, using the pointwise inequality
	\[
		\int_{\R^{3}} N^2 V(N^{1/2}(x-y,x-z)) |\varphi(z)|^2 \d z \leq N^{1/2} \norm{\varphi}_{\infty}^2\sup_{x\in \R^{3} } \norm{V(x,\cdot)}_{L^1(\R^{3})}\,,
	\]
	we obtain
	\[
		\cL_4^{(3)} \leq C \sup_{x\in \R^{3} } \norm{V(x,\cdot)}_{L^1(\R^{3})} \norm{\varphi}_{\infty}^2 N^{1/2} \cN^2\,,
	\]
	and, on the other hand, by the Cauchy--Schwarz inequality, we have
	\begin{align*}
		\cL_4^{(2)} &\leq \frac{N}{2} \int_{\R^{9}} V_N(x-y,x-z) \left(a^*_x a^*_y a_x a_y |\varphi(z)|^2 + a^*_x a^*_z a_x a_z |\varphi(y)|^2 \right) \d x \d y \d z \\
			&\leq C \sup_{x\in \R^{3} } \norm{V(x,\cdot)}_{L^1(\R^{3})} \norm{\varphi}_{\infty}^2 N^{1/2} \cN^2\,. \qedhere
	\end{align*}
\end{proof}

\subsubsection{Contribution from $\d \Gamma(h)$}
\begin{lemma}\label{lem:E_kin}
	There exists $C_\varphi>0$, depending only on $\norm{h \varphi}_{2}$ and $\norm{\varphi}_{\infty}$, such that
	\begin{multline*}
		\pscal{\Omega, e^{-B} \d \Gamma(h) e^{B} \Omega} \\
		\leq \frac{1}{6} N^3 \int_{\R^9} (V_Nf_N\omega_N)(x-z,y-z) |\varphi(x)|^2| \varphi(y)|^2|\varphi(z)|^2 \d x \d y \d z + C_\varphi N^{1/2}\,.
	\end{multline*}
\end{lemma}
\begin{proof}
	We have
	\begin{equation}\label{eq:lem_apriori_kin_1}
		e^{-B} \d \Gamma(h) e^{B} = \d \Gamma(h) - \int_0^1e^{-sB} [B,\d \Gamma(h)]e^{sB} \d s\,.
	\end{equation}
	Let us compute
	\begin{align*}
		-&[B,\d \Gamma(h)] \\
		&= [B_1\Theta , \d \Gamma(h)] + h.c. \nn \\
		&= \int_{\R^9} \int_{\R^3} \overline{\cB(x,y,z)} [a_x a_y a_z \Theta , a^*_{x'}h_{x'} a_{x'}] \d x \d y \d z \d x' + h.c. \nn \\
		&= 3 \int_{\R^9} (h_1 \cB)(x,y,z) a_x a_y a_z \Theta \d x \d y \d z + h.c. \nn \\
		&=
			\begin{multlined}[t][.87\textwidth]
				-\frac{1}{2} N^{3/2} \int_{\R^9} \Big( (-\Delta_1 \omega_N(x,y,z)) \varphi(x) \varphi(y) \varphi(z) - \omega_N(x,y,z) (h\varphi) (x) \varphi(y) \varphi(z) \\
				- 2 i\nabla_1 \omega_N (x,y,z) \cdot \big( [-i \nabla + A](\varphi)(x) \big) \varphi(y) \varphi(z) \Big) a_x a_y a_z \Theta \d x \d y \d z + h.c.
			\end{multlined} \nn \\
		&=: \cK_1 + \cK_2 + \cK_3\,.
	\end{align*}
	The main contribution comes from $\cK_1$. Indeed, using the Hölder inequality with
	\begin{multline*}
		\norm{N^{3/2} \nabla_1 \omega_N (x,y,z) \cdot ([-i\nabla + A](\varphi)(x)) \varphi(y) \varphi(z)}_{L^2(\R^9)} \\
		\leq \norm{ (-i\nabla + A) \varphi }_{2} \norm{\varphi}_{\infty}^2 \norm{ N^{3/2} \nabla_{1} \omega_N }_{L^2(\mathbb{R}^{6})} \leq C N^{1/2}
	\end{multline*}
	and
	\[
		\norm{N^{3/2} \omega_N(x,y,z) (h\varphi) (x) \varphi(y) \varphi(z)}_{L^2(\R^9)} \leq \norm{\omega}_{L^2(\mathbb{R}^{6})} \norm{h \varphi}_{2} \norm{\varphi}_{\infty}\leq C\,,
	\]
	we obtain
	\begin{equation}\label{eq:lem_apriori_kin_3}
		\pm (\cK_2 + \cK_3) \leq C(1+N^{1/2})(\cN+1)^{3/2}\,.
	\end{equation}
	Noticing that
	\begin{align*}
		-\Delta_1\omega_N(x,y,z) - \Delta_2\omega_N(x,y,z) -\Delta_3 \omega_N(x,y,z)
			&=-2\Delta_{\cM}\omega_N (x-z,y-z) \\
			&= - (V_N f_N)(x-z,y-z)
	\end{align*}
	and using symmetry, we can rewrite $\cK_1$ as
	\[
		\cK_1 = -\frac{1}{6} N^{3/2} \int_{\R^9} (V_N f_N)(x-z,y-z) \varphi(x) \varphi(y) \varphi(z) a_x a_y a_z \Theta \d x \d y \d z + h.c.
	\]
	From~\eqref{eq:lem_apriori_kin_1}--\eqref{eq:lem_apriori_kin_3}, we obtain
	\begin{align*}
		e^{-B} \d \Gamma(h) e^{B} &= \d \Gamma(h) + \int_0^1 e^{-sB} (\cK_1 + \cK_2 + \cK_3) e^{sB} \d s \\
			&\leq
			\begin{multlined}[t][0.8\textwidth]
				\d \Gamma(h) + \cK_1 - \int_0^1\int_{0}^s e^{-uB} [B,\cK_1]e^{uB} \d u \d s \\
				+ C N^{1/2} \int_0^1 e^{-sB} (\cN+1)^{3/2}e^{sB} \d s\,.
			\end{multlined}
	\end{align*}
	Since we test against the vacuum, the first two terms will not contribute and the last one will give a contribution of order $N^{1/2}$. Let us then compute
	\begin{align*}
		&\!\!-[B,\cK_1] \\
		&=\begin{multlined}[t][.9\textwidth]
			-\frac{1}{6} N^{3/2} \int_{\R^9} (V_N f_N)(x-z,y-z) \varphi(x) \varphi(y) \varphi(z) \left[a_x a_y a_z \Theta , \Theta B_1^* \right] \d x \d y \d z + h.c.
		\end{multlined} \\
		&=
		\begin{multlined}[t][.9\textwidth]
			\begin{multlined}[t][.9\textwidth]
				-\frac{1}{6} N^{3/2} \int_{\R^9} \int_{\R^9} (V_N f_N)(x-z,y-z) \varphi(x) \varphi(y) \varphi(z) b(x',y',z') \\
				\times \Theta(\cN +3)[a_x a_y a_z, a^*_{x'} a^*_{y'} a^*_{z'}] \d x \d y \d z \d x' \d y' \d z' + h.c
			\end{multlined} \\
			\begin{multlined}[t][.9\textwidth]
				-\frac{1}{6} N^{3/2} \int_{\R^9} \int_{\R^9} (V_N f_N)(x-z,y-z) \varphi(x) \varphi(y) \varphi(z) b(x',y',z') \\
				\times \partial \Theta a^*_{x'} a^*_{y'} a^*_{z'} a_x a_y a_z \d x \d y \d z \d x' \d y' \d z' + h.c.
			\end{multlined}
		\end{multlined} \\
		&=: \Theta(\cN +3) \cK_{1,1} + \partial \Theta \cK_{1,2}\,,
	\end{align*}
	where we recall that $\partial \Theta = \Theta (\cN + 3) - \Theta$. The main contribution comes from $\cK_{1,1}$. Indeed, using $ | \partial \Theta | \leq C N^{-1} (\cN+1)^2$, we have
	\[
		\pm \partial \Theta \cK_{1,2} \leq C | \partial \Theta | \norm{ N^{3/2} (V_Nf_N) \varphi^{\otimes 3} }_{L^2(\R^9)} \norm{ \cB }_{L^2(\R^9)} (\cN+1)^3 \leq C (\cN+1)^5\,.
	\]
	Expanding the commutator in $\cK_{1,1}$, we obtain
	\begin{align*}
		\cK_{1,1} &=
		\begin{multlined}[t][.9\textwidth]
			\frac{1}{3} N^3 \int_{\R^9} (V_Nf_N\omega_N)(x-z,y-z) |\varphi(x)|^2| \varphi(y)|^2|\varphi(z)|^2 \d x \d y \d z \\
			\begin{multlined}[t][.85\textwidth]
				+ N^3 \int_{\R^{12}} (V_Nf_N)(x,y,z) \omega_N(x',y,z) \varphi(x) \overline{\varphi(x')} | \varphi(y)|^2|\varphi(z)|^2 \\
				\times a^*_x a_{x'} \d x \d x' \d y \d z
			\end{multlined} \\
			\begin{multlined}[t][.85\textwidth]
				+ \frac{1}{2} N^3 \int_{\R^{15}} (V_Nf_N)(x,y,z) \omega_N(x',y',z) \varphi(x) \varphi(y) \overline{\varphi(x')} \overline{\varphi(y')} |\varphi(z)|^2 \\
				\times a^*_x a^*_y a_{x'} a_{y'} \d x \d x' \d y \d y' \d z
			\end{multlined}
		\end{multlined} \\
		&=: (\mathrm{I}) + (\mathrm{II}) + (\mathrm{III})\,.
	\end{align*}
	Of course $(\mathrm{I})$ is the main contribution. Let us bound the other terms. From~\eqref{eq:bounds_on_w_general}, we have
	\[
		\int_{\R^{3}} N^2 \omega_N(x',y,z) |\varphi(x')| \d x' \leq C \int_{\R^{3}} \frac{\norm{\varphi}_{\infty}}{\left(|x'-y|^2 + |y-z|^2\right)^2} \d x' \leq C\frac{ \norm{\varphi}_{\infty} }{|y-z|}\,,
	\]
	which leads to
	\begin{align*}
		(\mathrm{II})
		&\leq \begin{multlined}[t][.87\textwidth]
			C N^3 \int_{\R^{12}} (V_Nf_N)(x,y,z) \omega_N(x',y,z) |\varphi(x) \overline{\varphi(x')}| | \varphi(y)|^2|\varphi(z)|^2 \\
			\times (a^*_x a_{x} + a^*_{x'} a_{x'}) \d x \d x' \d y \d z
		\end{multlined}\\
		&\leq \begin{multlined}[t][.87\textwidth]
			C \norm{\varphi}_{\infty} N \int_{\R^{9}} (V_Nf_N)(x,y,z) |\varphi(x)| \frac{| \varphi(y)|^2|\varphi(z)|^2}{|y-z|} a^*_x a_{x} \d x \d y \d z \\
			+ C N^{\frac{5}{2}} \int_{\R^{9}} \omega_N(x',y,z) |\overline{\varphi(x')}| | \varphi(y)|^2|\varphi(z)|^2 a^*_{x'} a_{x'} \d x' \d y \d z
		\end{multlined} \\
		&\leq \begin{multlined}[t][.87\textwidth]
			C \norm{\varphi}_{\infty}^4 \norm{|\varphi|^2 \ast \frac{1}{|\cdot|} \ast \1_{ \left\{|\cdot|\leq C N^{-\frac{1}{2}}\right\} }}_{L^\infty(\mathbb{R}^{3})} N^2 \cN \\
			+ C \norm{|\varphi|^2 \ast \frac{1}{|\cdot|}}_{L^\infty(\mathbb{R}^{3})} \norm{\varphi}_{\infty}^4 N^{\frac{1}{2}} \cN
		\end{multlined}\\
		& \leq C \left( \norm{h \varphi}_{2} + \norm{\varphi}_{\infty}\right)^6 N^{\frac{1}{2}} \cN\,.
	\end{align*}
	Similar computations give
	\begin{align*}
		(\mathrm{III})
		&\leq \begin{multlined}[t]
			N^3 \int_{\R^{15}} (V_Nf_N)(x,y,z) \omega_N(x',y',z) \varphi(x) \varphi(y) \overline{\varphi(x')} \overline{\varphi(y')} |\varphi(z)|^2 \\
			\times \left( a^*_x a^*_y a_{x} a_{y} + a^*_{x'} a^*_{y'} a_{x'} a_{y'} \right) \d x \d x' \d y \d y' \d z
		\end{multlined} \\
		&\leq C \norm{\varphi}_{H^1(\mathbb{R}^{3})}^2 \norm{\varphi}_{\infty}^4 N^{\frac{1}{2}} \cN^2\,.
	\end{align*}
	We therefore obtain
	\begin{multline*}
		\cK_{1,1} \leq \frac{1}{3} N^3 \int_{\R^9} (V_N f_N \omega_N)(x-z,y-z) |\varphi(x)|^2| \varphi(y)|^2|\varphi(z)|^2 \d x \d y \d z \\
		+ C N^{\frac{1}{2}} (1+\cN^2)\,.
	\end{multline*}
	Noting that $\int_0^1 \int_0^s \d u \d s = 1/2$ and $0\leq 1-\Theta \leq N^{-1/2} \cN$, an application of Lemma~\ref{lem:nb_particule} concludes the proof of Lemma~\ref{lem:E_kin}.
\end{proof}

\subsubsection{Contribution from $\cL_6$}
\begin{lemma}\label{lem:L_6}
	There exists $C_\varphi>0$, depending only on $\norm{h \varphi}_{2}$ and $\norm{\varphi}_{\infty}$, such that
	\begin{multline*}
		\pscal{\Omega, e^{-B} \cL_6 e^{B} \Omega} \\
		\leq \frac{1}{6} N^3 \int_{\R^9} (V_N\omega_N^2)(x-z,y-z) |\varphi(x)|^2| \varphi(y)|^2|\varphi(z)|^2 \d x \d y \d z + C_\varphi N^{\frac{1}{2}}\,.
	\end{multline*}
\end{lemma}
\begin{proof}
	We have
	\[
		e^{-B} \cL_6 e^{B} = \cL_6 - \int_0^1e^{sB} [B,\cL_6] e^{-sB} \d s \,.
	\]
	When estimated against the vacuum, only the last term gives a non zero contribution. Let us therefore focus on it. Using that $[\mathcal{X} \mathcal{Y}, \mathcal{Z}] = \mathcal{X} [\mathcal{Y}, \mathcal{Z}] $ when $[\mathcal{X} ,\mathcal{Z}] = 0$, we have
	\begin{align*}
		&-[B,\cL_6] \\
		&\begin{multlined}[t][.87\textwidth]
			\,\, = -\frac{1}{36} \int_{\R^{9}} \int_{\R^{9}} V_N(x-y,x-z) [a_x a_y a_z, a^*_{x'} a^*_{y'} a^*_{z'}]a_x a_y a_z \\
			\times \Theta N^{3/2}\omega_N(x',y',z') \overline{\varphi(x')}\overline{\varphi(y')}\overline{ \varphi(z')} \d x \d y \d z \d x' \d y' \d z' + h.c.
		\end{multlined} \\
		&\begin{multlined}[t][.87\textwidth]
			\,\, = -\frac{1}{6} \int_{\R^{9}} N^{3/2} V_N(x-y,x-z) \omega_N(x,y,z) \overline{\varphi(x)\varphi(y) \varphi(z)} a_x a_y a_z \Theta \d x \d y \d z + h.c. \\
			\begin{multlined}[t][.87\textwidth]
				-2 \int_{\R^{12}} N^{3/2} V_N(x-y,x-z) \omega_N(x',y,z) \overline{\varphi(x')\varphi(y) \varphi(z)} \\
				\times a^*_x a_x a_y a_z a_{x'} \Theta \d x \d y \d z \d x' + h.c.
			\end{multlined} \\
			\begin{multlined}[t][.87\textwidth]
				- \frac{1}{4} \int_{\R^{15}} N^{3/2} V_N(x-y,x-z) \omega_N(x',y',z) \overline{\varphi(x')\varphi(y') \varphi(z)} \\
				\times a^*_x a^*_y a_x a_y a_{x'} a_{y'} a_z \Theta \d x \d y \d z \d x' \d y' + h.c.
			\end{multlined}
		\end{multlined} \\
		&\,\,=: \widetilde{\cL}_6^{(1)} + \widetilde{\cL}_6^{(2)} + \widetilde{\cL}_6^{(3)}\,.
	\end{align*}
	We now prove that only the contribution of $\widetilde{\cL}_6^{(1)}$ is of order $N$.
	
	On one hand, using the Cauchy--Schwarz inequality, we obtain for all $\eta>0$ that
	\begin{align*}
		\widetilde{\cL}_6^{(2)}
		&\leq \begin{multlined}[t][.87\textwidth]
			C \eta \int_{\R^{12}} \Theta a^*_x a^*_y a^*_z a^*_{x'} \left(\cN+1\right)^{-3/2} a_x a_y a_z a_{x'} \d x \d y \d z \d x' \\
			+ C\eta^{-1} \int_{\R^{3}} \Theta(\cN+3) \left(\int_{\R^9} N^{3} V_N(x-y,x-z)^2 \omega_N(x',y,z)^2 \right. \\
			\left. \times |\varphi(x')|^2|\varphi(y)|^2 |\varphi(z)|^2 \d y \d z \d x' \vphantom{\int_{\R^9}}\right) a^*_x \left(\cN+1\right)^{3/2}a_x\d x
		\end{multlined} \\
		&\leq \begin{multlined}[t][.87\textwidth]
			C \eta (\cN+1)^{5/2} + \eta^{-1} N^{1/2} \norm{\varphi}_{\infty}^4 \norm{N^{3/2} \omega_N}_{L^2(\mathbb{R}^{6})}^2 \\
			\times \norm{|\varphi|^2 \ast N^{3/2}\1_{ \left\{ |\cdot|\leq C N^{-1/2} \right\} }}_{L^\infty(\mathbb{R}^{3})} (\cN+1)^{5/2}
		\end{multlined} \\
		&\leq C N^{1/4} (\cN +1)^{5/2}\,,
	\end{align*}
	where we used that $V_N(x,y,z) \leq C N \1_{ \left\{ |x-z|\leq CN^{-1/2} \right\} }$ and optimized over $\eta$.
	
	On the other hand, and similarly, we have for all $\eta>0$ that
	\begin{align*}
		\widetilde{\cL}_6^{(3)}
		&\leq \begin{multlined}[t][.87\textwidth]
			C \eta \int_{\R^{15}} a^*_x a^*_y a^*_z a^*_{x'} a^*_{y'} \left(\cN+1\right)^{-3/2} a_x a_y a_z a_{x'} a_{y'} \d x \d y \d z \d x' \d y' \\
			+ C \eta^{-1} \int_{\R^{15}} N^5 V(N^{1/2}(x-y,x-z))^2 \omega_N(x',y',z)^2 |\varphi(x')|^2 |\varphi(y')|^2 |\varphi(z)|^2 \\
			\times a^*_x a^*_y\left(\cN+1\right)^{3/2} a_x a_y \d x \d y \d z \d x' \d y'
		\end{multlined} \\
		&\leq \begin{multlined}[t][.87\textwidth]
			C \eta \left(\cN+1\right)^{7/2} + \eta^{-1} N^{1/2} \norm{\varphi}_{\infty}^4 \norm{N^{3/2}\omega_N}_{L^2(\mathbb{R}^{6})}^2 \\
			\times \norm{|\varphi|^2 \ast N^{3/2}\1_{ \left\{ |\cdot|\leq C N^{-1/2} \right\} }}_{L^\infty(\mathbb{R}^{3})} \left(\cN+1\right)^{7/2}
		\end{multlined} \\
		&\leq C N^{1/4} \left(\cN+1\right)^{7/2}\,,
	\end{align*}
	where we also optimized over $\eta$ in order to obtain the last inequality.
	
	Finally, to access the contribution of $\widetilde{\cL}_6^{(1)}$, we apply one more time the Duhamel formula:
	\[
		\int_0^1e^{-sB} \widetilde{\cL}_6^{(1)} e^{sB} \d s = \widetilde{\cL}_6^{(1)} - \int_0^1 \int_0^s e^{-uB}[B, \widetilde{\cL}_6^{(1)}] e^{uB} \d u \d s\,.
	\]
	As before, the first term vanishes when tested against the vacuum state $\Omega$ and we focus therefore on the second term. Introducing the notations
	\[
		\widetilde{\cL}_6^{(1,\circ)} := -\frac{1}{6} \int_{\R^{9}} N^{3/2} V_N(x-y,x-z) \omega_N(x,y,z) \overline{\varphi(x)\varphi(y) \varphi(z)} a_x a_y a_z \d x \d y \d z
	\]
	and $\widetilde{\cL}_6^{(1,\dagger)} := (\widetilde{\cL}_6^{(1,\circ)})^*$, we have
	\[
		\widetilde{\cL}_6^{(1)} = \Theta \widetilde{\cL}_6^{(1,\dagger)} + \widetilde{\cL}_6^{(1,\circ)}\Theta \,.
	\]
	A direct computation gives $[B_1\Theta , \widetilde{\cL}_6^{(1,\circ)}\Theta ]=0$, from which we obtain
	\begin{align*}
		-[B, \widetilde{\cL}_6^{(1)}] &= [ B_1\Theta - B_1^*\Theta , \widetilde{\cL}_6^{(1,\circ)}\Theta ] + h.c. = [\widetilde{\cL}_6^{(1,\circ)}\Theta ,\Theta B_1^*] + h.c. \\
			&= \Theta(\cN+3) [\widetilde{\cL}_6^{(1,\circ)},B_1^*] + \partial \Theta B_1^* \widetilde{\cL}_6^{(1,\circ)} +h.c. \\
			&\leq [\widetilde{\cL}_6^{(1,\circ)},B_1^*] + h.c. + C(\cN+1)^5\,,
	\end{align*}
	where we applied the Cauchy--Schwarz inequality and used
	\[
		B_1^* B_1 \leq C \left(\cN+1\right)^3\,, \quad \widetilde{\cL}_6^{(1,\dagger)} \widetilde{\cL}_6^{(1,\circ)} \leq C N^2 (\cN+1)^3\,, \quad \text{and} \quad \partial\Theta \leq C N^{-1}(\cN+1)^2\,.
	\]
	In view of~\eqref{eq:wick_3_3}, we have
	\begin{multline*}
		[\widetilde{\cL}_6^{(1,\circ)},B_1^*] = \frac{1}{6} \int_{\R^9} N^3 (V_N\omega_N^2)(x,y,z) |\varphi(x)|^2 |\varphi(y)|^2 |\varphi(z)|^2 \d x \d y \d z \\
		+ \frac{1}{2} N^3 \int_{\R^{12}} (V_N\omega_N)(x,y,z) \omega(x,y,z') |\varphi(x)|^2 |\varphi(y)|^2 \varphi(z) \overline{\varphi(z')} a^*_z a_{z'} \d x \d y \d z \d z' \\
		\begin{multlined}[t]
			+ \frac{1}{4} N^3 \int_{\R^{15}} (V_N\omega_N)(x,y,z) \omega(x,y',z') |\varphi(x)|^2 \varphi(y) \varphi(z) \overline{\varphi(y')} \overline{\varphi(z')} \\
			\times a^*_y a^*_z a_{y'} a_{z'} \d x \d y \d y' \d z \d z'
		\end{multlined}
	\end{multline*}
	hence
	\begin{multline}\label{eq:comm_L_6_1_B_1}
		[\widetilde{\cL}_6^{(1,\circ)},B_1^*] + h.c. \leq \frac{1}{6} \int_{\R^9} N^3 (V_N\omega_N^2)(x,y,z) |\varphi(x)|^2 |\varphi(y)|^2 |\varphi(z)|^2 \d x \d y \d z \\
		+ C(N^{1/2}\cN + \cN^2)\,,
	\end{multline}
	where we used $V_N \leq \norm{V}_{\infty} N$ and computations similar to the ones in the proof of Lemma~\ref{lem:commutator_B_B_star}. This concludes the proof of Lemma~\ref{lem:L_6}.
\end{proof}

\subsubsection{Contribution from $\cL_3^{(1)}$}

\begin{lemma}\label{lem:L_3}
	There exists $C_\varphi>0$, depending only on $\norm{h \varphi}_{2}$ and $\norm{\varphi}_{\infty}$, such that
	\begin{multline*}
		\pscal{\Omega, e^{-B} \cL_3^{(1)} e^{B} \Omega} \\
		\leq -\frac{1}{3} \int_{\R^9} N^{3} (V_N\omega_N)(x,y,z)|\varphi(x)|^2 |\varphi(y)|^2 |\varphi(z)|^2 \d x \d y \d z + C_\varphi N^{1/2}\,.
	\end{multline*}
\end{lemma}
\begin{proof}
	Let us first introduce the notations
	\[
		\cL_3^{\dagger} = \frac{N^{3/2}}{6} \int_{\R^{9}} V_N(x-y,x-z) a^*_x a^*_y a^*_z \varphi(x)\varphi(y)\varphi(z) \d x \d y \d z
	\]
	and $\cL_3^{\circ} = (\cL_3^{\dagger})^* $, so that $\cL_3^{(1)} = \cL_3^{\dagger} + \cL_3^{\circ}$. Still denoting $\Theta = \1_{[0,N^{1/2}]}(\cN)$, we have
	\begin{multline}\label{eq:est_L_3}
		(1-\Theta) \cL_3^{\dagger} + \cL_3^{\circ} (1 - \Theta) \\
		\begin{aligned}[b]
			&\leq (1-\Theta) \left(\int N^{3} V_N(x-y,x-z)^2 |\varphi(x)\varphi(y)\varphi(z)|^2 \d x \d y \d z \right)^{\frac{1}{2}} \left(\cN +1\right)^{\frac{3}{2}}\\
			&\leq C N \norm{V}_{L^2(\mathbb{R}^{6})} \norm{\varphi}_{\infty}^2 \norm{\varphi}_{2} \left(1-\Theta \right)\left(\cN +1\right)^{\frac{3}{2}} \leq C \left(\cN +1\right)^{\frac{7}{2}}\,.
		\end{aligned}
	\end{multline}
	Therefore, $\cL_3^{(1)} \leq \Theta \cL_3^{\dagger} + \cL_3^{\circ} \Theta + C (\cN+1)^{7/2}$ and it is enough to look at
	\[
		e^{-B} (\Theta \cL_3^{\dagger} + \cL_3^{\circ} \Theta ) e^{B} = (\Theta \cL_3^{\dagger} + \cL_3^{\circ} \Theta ) - \int_0^1e^{sB} [B,\Theta \cL_3^{\dagger} + \cL_3^{\circ} \Theta ] e^{-sB} \d s\,.
	\]
	The first term vanishes when tested against $\Omega$. Let us compute the commutator. Using~\eqref{eq:wick_3_3}, we obtain
	\begin{align*}
		-[B,\Theta \cL_3^{\dagger} + \cL_3^{\circ} \Theta ]
			&= -[B,\cL_3^{\circ} \Theta ] +h.c. = [B_1 \Theta - \Theta B_1^*,\cL_3^{\circ} \Theta ] + h.c. \\
			&= [\cL_3^{\circ} \Theta ,\Theta B_1^*]+h.c. = \Theta(\cN+3) [\cL_3^{\circ} , B_1^*] + \partial\Theta B_1^*\cL_3^{\circ} + h.c. \\
			&\leq [\cL_3^{\circ} , B_1^*] +h.c. + C (\cN+1)^5 \,,
	\end{align*}
	where we used that $|\partial\Theta| \leq N^{-1} (\cN+1)^2$ and a computation similar to~\eqref{eq:est_L_3}.
	Then, the same computations as in~\eqref{eq:comm_L_6_1_B_1} give
	\begin{align*}
		[\cL_3^{\circ}, {}&B_1^*] + h.c. \\
			&=\begin{multlined}[t][.87\textwidth]
				-\frac{1}{36} \int_{\R^9 \times \R^9} N^{3} V_N(x,y,z) \omega_N(x',y',z') \overline{\varphi(x)\varphi(y)\varphi(z)}\varphi(x')\varphi(y')\varphi(z') \\
				\times[a_x a_y a_z,a^*_{x'}a^*_{y'}a^*_{z'}] \d x \d y \d z \d x' \d y' \d z' + h.c.
			\end{multlined} \\
			&\leq -\frac{1}{3} \int_{\R^9} N^{3} (V_N\omega_N)(x,y,z)|\varphi(x)|^2 |\varphi(y)|^2 |\varphi(z)|^2 \d x \d y \d z + C(N^{1/2}\cN + \cN^2)\,,
	\end{align*}
	which concludes the proof.
\end{proof}

\subsubsection{Conclusion of the proof of Theorem~\ref{theo:upper_bound_fock_space}: collecting the leading contributions} \label{sec:concl_proof_th_UB}
	In Lemmas~\ref{lem:E_kin}, \ref{lem:L_3}, and~\ref{lem:L_6}, we extracted the contributions of $\d \Gamma(h), \cL_3$ and $\cL_6$ to the leading order. Adding the contribution from $\cL_0$ and controlling the remainder terms using Lemma~\ref{lem:L_2_4}, we obtain
\begin{multline*}
	\braket{ \Omega, e^{-B} W(\sqrt{N} \varphi)^* \mathbb{H}_N W(\sqrt{N} \varphi) e^{B} \Omega } \\
	\begin{aligned}[t]
		&\leq \begin{multlined}[t][0.8\textwidth]
			N \int_{\R^3} \left( |(-i\nabla + A(x)) \varphi (x)|^2 + \Vext(x)|\varphi (x)|^2 \right) \d x \\
			+ \frac{N}{6 }\int_{\R^9} (N^2 V_Nf_N)(x,y,z) |\varphi(x)|^2 |\varphi(y)|^2 |\varphi(z)|^2 \d x \d y \d z \\
			+ \int_{\R^9} \left(\frac{f_N + \omega_N}{6} -\frac{1}{3} + \frac{1}{6} \right) (N^2 V_N\omega_N)(x,y,z) |\varphi(x)|^2 |\varphi(y)|^2 |\varphi(z)|^2 \d x \d y \d z \\ + CN^{1/2}
		\end{multlined} \\
		&\leq N\Egp(\varphi) + CN^{1/2}\,.
	\end{aligned}
\end{multline*}
We have used that $f + \omega = 1$, so that the third term above vanishes, and also that
\[
	\frac{1}{6} \int_{\R^9} (N^2 V_Nf_N)(x,y,z) |\varphi(x)|^2 |\varphi(y)|^2 |\varphi(z)|^2 \d x \d y \d z \leq \frac{b_{\cM}(V)}{6}\int_{\R^3} |\varphi(x)|^6\d x\,,
\]
which is obtained from the H{\"o}lder inequality and from that $b_{\cM}(V) = \int_{\mathbb{R}^{6}} Vf = \norm{Vf}_{1}$. This concludes the proof of Theorem~\ref{theo:upper_bound_fock_space}.

\subsection{Proof of Corollary~\ref{Cor_upper_bound}}
	For $\mu > 0$, denote
	\[
		E^\mu(M,N) = \inf \sigma (H_{M,N}) - \mu (M + M^3/N^2)\,,
	\]
	where $H_{M,N}$ is defined in~\eqref{Def_H_MN}, and $E(M,N) := E^0(M,N)$.
	We claim that there exists $\mu_0>0$ such that, for all $\mu > \mu_0$ and $N\in\N$, $\left\{E^\mu (M,N)\right\}_{M\in\N}$ is a non-increasing sequence. Indeed, let $\Gamma_{M-1,N} = \ket{\Psi_{M-1,N}}\bra{\Psi_{M-1,N}}$ be the density matrix of a ground state $\Psi_{M-1,N}$ of $H_{M-1,N}$, then we have the inequalities
	\begin{align*}
		E(M,N) &\leq M \Tr_{L^2(\mathbb{R}^{3N})} h_1 \Gamma_{M-1,N} + \binom{M}{3} \Tr_{L^2(\mathbb{R}^{3N})} V_N(x_1,x_2,x_3) \Gamma_{M-1,N} \\
			&\leq \begin{multlined}[t][.85\textwidth]
				E(M-1,N) + \Tr_{L^2(\mathbb{R}^{3N})} \left\{ \left(\frac{M}{M-1} -1\right) h_1 \right. \\
				\left. + \left(\frac{M}{M-3}-1\right) \binom{M-1}{3} V_N(x_1,x_2,x_3) \right\} \Gamma_{M-1,N}
			\end{multlined} \\
			&\leq E(M-1,N) + \frac{C}{M}E(M-1,N) \\
			&\leq E(M-1,N) + C \norm{h^{1/2} \varphi}_{2}^6 \left(1 + \frac{M^2}{N^2}\right),
	\end{align*}
	where we bounded $E(M-1,N)$ by the energy of the ansatz $\varphi^{\otimes M-1}$ in the last inequality. Hence, it suffices to choose $\mu_0 = C\norm{h^{1/2} \varphi}_{2}^6$.
	
	Let $\widetilde N = N - N^{2/3}$ and let us denote $\Psi_{N} := W(\sqrt{\widetilde N} \varphi) U_N \Omega \in \cF(L^2(\R^3))$. As a consequence of~\eqref{eq:coprime_with_3},~\eqref{eq:lem_nb_particule_U_N}, and
	\[
		W(\sqrt{\widetilde N} \varphi)^* \cN W(\sqrt{\widetilde N} \varphi) = \cN + \sqrt{\widetilde N} (a^*(\varphi) + a(\varphi)) + \widetilde N\,,
	\]
	$\Psi_{N}$ satisfies
	\begin{align}
		\braket{ \Psi_N, \cN \Psi_N } &= N - N^{2/3} + \cO(1)\,, \label{eq:moment_estimate_ansatz1}\\
		\braket{ \Psi_N, \left(\cN - \braket{ \Psi_N, \cN \Psi_N }\right)^2 \Psi_N } &= \cO(N)\,, \label{eq:moment_estimate_ansatz2}
		\intertext{and}
		\braket{ \Psi_N, \cN^3 \Psi_N } &= N^3 + \cO(N^{8/3})\,.\label{eq:moment_estimate_ansatz3}
	\end{align}
	
	Let us first argue that
	\begin{equation}\label{eq:modified_ansatz}
		\braket{ \mathbb{H}_N }_{\Psi_N} \leq N \Egp(\varphi) + C N^{2/3}\,.
	\end{equation}
	Indeed, following the proof of Theorem~\ref{theo:upper_bound_fock_space}, we have
	\begin{multline*}
		\braket{ \Omega, W(\sqrt {\widetilde{N}} \varphi)^* \mathbb{H}_N W(\sqrt {\widetilde{N}} \varphi) }_{U_N \Omega} \\
		= \braket{ \Omega, U_N^* W(\sqrt {N} \varphi)^* \mathbb{H}_N W(\sqrt {N} \varphi) }_{U_N \Omega} \\
		+ (\widetilde N-N) \pscal{\varphi, h\varphi} + \sum_{i=0}^6 \left( (\widetilde N/N)^{(6-i)/2}-1 \right) \braket{ \cL_i }_{U_N \Omega}\,.
	\end{multline*}
	Therefore, using Theorem~\ref{theo:upper_bound_fock_space} and that $|\!\braket{ \cL_i }_{U_N \Omega}\!| \leq C N$ (see the proof of Theorem~\ref{theo:upper_bound_fock_space}), we obtain~\eqref{eq:modified_ansatz}.
	
	Let us now denote $\Psi_{M,N} = \1_{\cN=M} \Psi_N$. For $\mu > \mu_0$, using that $\left\{E^\mu (M,N)\right\}_{M\in\N}$ is a non-increasing sequence and noticing that $\sum_{M\leq N} \norm{ \Psi_{M,N} }^2 =1 - \braket{ \1_{\cN > N} }_{\Psi_N}$, we have
	\begin{multline*}
		\left( 1 - \braket{ \1_{\cN > N} }_{\Psi_N} \right) E^\mu(N,N) \\
		\leq \begin{multlined}[t][0.81\textwidth]
			\sum_{M\leq N} E^\mu(M,N) \norm{ \Psi_{M,N} }^2 \\
			\begin{aligned}
				&= \sum_{M=0}^{\infty} E^\mu(M,N) \norm{ \Psi_{M,N} }^2 - \sum_{M> N} E^\mu(M,N) \norm{ \Psi_{M,N} }^2 \\
				&\leq \Braket{ \mathbb H _N - \mu \left(\cN +\frac{\cN^3}{N^2} \right) + \mu \1_{\cN > N} \left(\cN + \frac{\cN^3}{N^2}\right) }_{\Psi_{N}}\,,
			\end{aligned}
		\end{multlined}
	\end{multline*}
	
	Now, using~\eqref{eq:moment_estimate_ansatz1}--\eqref{eq:moment_estimate_ansatz3} and Chebyshev's inequality, we obtain that
	\[
		\Braket{ \1_{\cN > N} }_{\Psi_N} = \Braket{ \1_{\cN > \Braket{ \cN }_{\Psi_N} + N^{\frac23} + \cO(1)} }_{\Psi_N} \leq C N^{-\frac43} \Braket{ \left(\cN - \Braket{ \cN }_{\Psi_N}\right)^2 }_{\Psi_N} \leq C N^{-\frac13}\,,
	\]
	\begin{align*}
		\Braket{ \1_{\cN > N} \cN }_{\Psi_N} &= \Braket{ \1_{\cN > N} \left(\cN - \Braket{ \cN }_{\Psi_N} \right) }_{\Psi_N} + \Braket{ \1_{\cN > N} }_{\Psi_N} \Braket{ \cN }_{\Psi_N} \\
			&\leq C N^{-\frac23} \Braket{ \left(\cN - \Braket{ \cN }_{\Psi_N}\right)^2 }_{\Psi_N} + C N^{-\frac13} \Braket{ \cN }_{\Psi_N} \leq C N^{\frac23}\,,
	\end{align*}
	and
	\begin{align*}
		\Braket{ \1_{\cN > N} \frac{\cN^3}{N^2} }_{\Psi_N}
		&\begin{multlined}[t][.8\textwidth]
			= N^{-2} \Braket{ \1_{\cN > N} \left(\cN - \Braket{ \cN }_{\Psi_N} \right)^3 }_{\Psi_N} \\
			+ 3 N^{-2} \Braket{ \1_{\cN > N} \left( \cN - \Braket{ \cN }_{\Psi_N} \right)^2 }_{\Psi_N} \Braket{ \cN }_{\Psi_N} \\
			+ 3 N^{-2} \Braket{ \1_{\cN > N} \left( \cN - \Braket{ \cN }_{\Psi_N} \right) }_{\Psi_N} \Braket{ \cN }_{\Psi_N}^2 \\
			+ N^{-2} \Braket{ \1_{\cN > N} }_{\Psi_N} \Braket{ \cN }_{\Psi_N}^3
		\end{multlined} \\
		&\leq N^{-2} \Braket{ \left( \cN - \Braket{ \cN }_{\Psi_N} \right)^3 }_{\Psi_N} + C N^{\frac23} \leq C N^{\frac23}\,.
	\end{align*}
	Hence, we deduce from the above inequality on $E^\mu(N,N)$ that
	\[
		E^\mu(N,N) \leq \left( 1+CN^{-1/3} \right) \Braket{ \mathbb{H}_N }_{\Psi_N} - 2\mu N + C N^{2/3}\,.
	\]
	Using~\eqref{eq:modified_ansatz} for $\varphi = u_0$, the GP minimizer, finishes the proof.
	\qed

\section{Convergence of ground states}\label{sec:state}
	The convergence of states in Theorem~\ref{thm:main2} follows from a simple adaption of the proof in~\cite[Sect.~4C]{NamRouSei-16} to our method of proof for the lower bound. Let us briefly explain the main steps for the reader's convenience. Let $\Psi_N$ be a normalized state in $\gH^N$ satisfying
\[
	\lim_{N\to \infty} \frac{ \pscal{ \Psi_N, H_N \Psi_N } }{N}= \GPnrg\,.
\]

	Since $\Tr h \gamma_{\Psi_N} \leq C$ and $h$ has compact resolvant, we know by the de Finetti theorem~\cite[Corollary 2.4]{LewNamRou-14} that, up to a subsequence as $N\to \infty$, there exists a Borel probability measure $\mu$ on the unit sphere $S\gH$ such that
\[
	\lim_{N\to \infty} \left| \gamma_{\Psi_N}^{(k)} - \int_{S\gH} \ketbra{ u^{\otimes k} }{ u^{\otimes k} } \d \mu(u) \right| = 0\,, \quad k=1, 2, \dots\,.
\]
Hence, it only remains to prove that the support of $\mu$ is contained in the set of minimizers of $\Egp$. To this end, it is enough prove
\begin{equation} \label{eq:prop-mu}
	\int_{S\gH} |\!\braket{v,u}\!|^{2k} \d\mu (u) \leq \sup_{u\in \GPmin} |\!\braket{v,u}\!|^{2k}
\end{equation}
for all $v\in L^2(\mathbb{R}^{3})$ and $k \in \N\setminus\{0\}$.
Indeed, one can easily verify that passing to the limit $k\to \infty$ in~\eqref{eq:prop-mu} implies that $\supp \mu \subset \GPmin$. Proving~\eqref{eq:prop-mu} is done by means of a Hellmann--Feynman argument, which we now explain. Keeping the same notations as in~\cite[Lemma 4.3]{NamRouSei-16}, for $v \in L^2(\mathbb{R}^{3})$ and $N,k \in \N\setminus\{0\}$, we define
\[
	S_{k,v}^N = \frac{k!}{N^{k-1}} \sum_{1\leq i_1 < \dots < i_k \leq N} \ketbra{v^{\otimes k}}{v^{\otimes k}}\,.
\]
The above operator is bounded uniformly in $N$ and one can easily check that we can carry out the proof of Theorem~\ref{theo:lower_bound} (lower bound) with $H_N - S_{k,v}^N$ to obtain
\[
	\liminf_{N \to \infty} \frac{\inf \sigma(H_N - S_{k,v}^N)}{N} \geq \inf_{\norm{u}_{2}=1} \left\{ \Egp (u) - |\!\braket{v,u}\!|^{2k} \right\}.
\]
Note, in particular, that the binding inequality~\eqref{eq:binding_inequality_M} is satisfied with $H_{M,N}$ replaced by $H_{M,N} - S_{k,v}^N$ and that the proof of~\cite[Lemma 2]{LieSei-06} easily adapts to yield the four\nobreakdash-body estimate~\eqref{eq:4_body_collision}. Therefore, for all $t>0$, implementing the change of parameter $v \to t^{1/(2k)} v$ and dividing by $t$, we obtain
\begin{align*}
	\int_{S\gH} |\!\braket{v,u}\!|^{2k} \d\mu (u)
		&= \lim_{N \to \infty} \frac{\< \Psi_N, H_N \Psi_N \> - \< \Psi_N, (H_N - S_{k,t^{1/(2k)} v}^N) \Psi_N \>}{t N} \\
		&\leq t^{-1} \left( \GPnrg - \inf_{\norm{u}_{2}=1} \left\{ \Egp (u) - t |\!\braket{v,u}\!|^{2k} \right\} \right).
\end{align*}
By standard compactness arguments, one can show that taking the limit $t \to 0$ above gives~\eqref{eq:prop-mu}. This concludes the sketch of the proof of Theorem~\ref{thm:main2}.

\appendix

\section*{Appendix. The two-body interaction case: lower bound}\label{app}
	In the two\nobreakdash-body interaction case, our proof of the lower bound simplifies and it gives a shorter alternative to~\cite{LieSei-06,NamRouSei-16}. Moreover, we only need the $L^1$-condition on the interaction potential, relaxing therefore some regularity assumptions in~\cite{LieSei-06,NamRouSei-16}. Since this simplification may be interesting in its own right, we give some details below.

	We consider $W \in L^1(\mathbb{R}^{3})$, a nonnegative and compactly supported potential, and we assume $A$ and $\Vext $ to be as in Theorem~\ref{thm:main1}. For $N,M \in \N\setminus\{0\}$, let us denote
\[
	E^{\mathrm{2B}}(M,N) := \inf \sigma (H_{M,N}^{\mathrm{2B}}) \quad \text{with } \quad H_{M,N}^{\mathrm{2B}} := \sum_{i=1}^M h_i + \sum_{1\leq i < j \leq M} W_N(x_i-x_j)\,,
\]
where $h_i\equiv h_{x_i}$, $h\equiv h_x := (-{\bf i}\nabla_{x} + A(x))^2 +\Vext(x)$, and $W_N = N^{2}W(N \cdot)$. Define
\[
	\cE_{\mathrm{2B}}(u) := \int_{\R^3} \left( |(-i\nabla + A(x)) u(x)|^2 + \Vext(x) |u(x)|^2 + \frac{b(V)}{2} |u(x)|^4 \right) \d x
\]
with $b(V)$ given in~\eqref{eq:def-scat-energy}. We have the following result.
\begin{theorem} \label{thm:2b} The ground state energy of $H_{N,N}^{\mathrm{2B}}$ satisfies
	\[
		\lim_{N\to\infty} \frac{E^{\mathrm{2B}}(N,N)}{N} = \inf_{\norm{u}_{2}=1} \cE_{\mathrm{2B}}(u)\,.
	\]
\end{theorem}

	This result has been proved in~\cite{LieSei-06,NamRouSei-16} (under a slightly stronger condition on~$W$). Here, we will sketch an alternative proof of the lower bound. Note that from our proof, it is also possible to derive the convergence of states similar to Theorem~\ref{thm:main2}. For a technical reason, we will work under the additional assumption~\eqref{eq:cond_A_V} on $A$, which can be removed following an argument in~\cite[Sect.~4B]{NamRouSei-16}.

	The main step in the proof of Theorem~\ref{thm:2b} is obtained via the following lemma.
\begin{lemma}[Reduction to softer potentials] \label{lem:reduc_2b}
	Let $\beta \in (0,2/3]$, $0<\eps<1<s$, and $\delta \in (0, \eps^2/2)$. We assume $\Vext $ to be as in~\eqref{eq:Vext} and $A \in L^3_{\mathrm{loc}}(\mathbb{R}^{3})$ to satisfy~\eqref{eq:cond_A_V}. Let $0\le \widetilde U\in L^\infty(\R^3)$ be radial with $\int_{\R^3} \widetilde U=1$ and $\supp \widetilde U \subset \{1/8 \le |\bx| \le 1/4\}$. Define $U$ as in~\eqref{eq:Ug} and $U_R=R^{-3}U(R^{-1}\cdot)$.
	Then, for all $N$, there exist an integer $M\in[(1-\eps)N, N]$ and $R\in [N^{-\beta}, N^{-\beta/2}]$ such that
	\begin{multline*}
		E^{\mathrm{2B}}(N,N) \ge \inf \sigma_{L^2_s (\R^{3M})} \bigg( \sum_{i=1}^{M} (h_{\eps,s})_i + \frac{b(W)}{(M-1)} \sum_{1\leq i < j \leq M} U_{R}(x_i-x_j) \bigg) \\
		- C_{\eps,s,\beta,\delta} R^{2/7} N - \eps C_{\beta} N - \delta C_{\eps} N\,,
	\end{multline*}
	where $h_{\eps,s} = h - (1-\eps) \1_{\{ |p| > s\}} p^2$.
\end{lemma}

	With Lemma~\ref{lem:reduc_2b} at hand, we can apply for instance~\cite[Theorem 2.5]{LewNamRou-16} with $0 < \beta < (21/2 + 3/\alpha)^{-1}$ in order to obtain the desired lower bound in Theorem~\ref{thm:2b}. This part is similar to the analysis in Section~\ref{sec:lower}. Thus, it remains to prove the lemma.
\begin{proof}[Proof of Lemma~\ref{lem:reduc_2b}]
	By adapting Lemma~\ref{lem:binding} to the two\nobreakdash-body interaction case, we find that there is a constant $C>0$, such that for all $N\in \N\setminus\{0\}$ and $0 < \varepsilon < 1/2$, we can find $M\equiv M(N,\varepsilon)\in\N$ such that $N(1-\varepsilon) \leq M \leq N $ and
	\[
		E^{\mathrm{2B}}(M,N) - E^{\mathrm{2B}}(M-3,N) \leq C \varepsilon^{-1}\,.
	\]
	This replaces a convexity argument in Step 1 of the proof of~\cite[Theorem 1]{LieSei-06}. Using this binding inequality and the heat kernel estimate in~\cite[Lemma 2]{LieSei-06}, we obtain for the zero temperature limit $\Gamma_{M,N}$ of the Gibbs state of $H_{M,N}^{\mathrm{2B}}$ the following analogue to Lemma~\ref{lem:4_body_collision}:
	\[
		\left\< \1_{\{|x_1-x_2|\leq R \}} \1_{\{|x_1-x_3|\leq R \}} \right\>_{\Gamma_{M,N}} \leq C_\eps R^6\,.
	\]
	Then, using
	\[
		\sum_{\substack{j=1 \\ j \neq i}}^M \1_{\{|x_j-x_{i}|\leq R \}} \prod_{\substack{\ell =1 \\ \ell \neq i,j}}^M \1_{\{|x_j-x_{\ell}|>2R \}} \leq 1\,, \quad i = 1, \dots, M\,,
	\]
	we obtain
	\begin{multline*}
		\sum_{i=1}^M \1_{\{|p_i| > s\}} p_i^2 + \sum_{1\leq i < j \leq M} W_N(x_i-x_j) \\
		\begin{aligned}
			&\geq \sum_{i=1}^M \sum_{\substack{j=1 \\ j \neq i}}^M \left(\1_{\{|p_i| > s\}} p_i\1_{\{|x_i-x_j|<R\}} p_i \1_{\{|p_i| > s\}} + \frac{1}{2} W_N(x_i-x_j) \right)\prod_{\substack{\ell =1 \\ \ell \neq i,j}}^M \1_{\{|x_j-x_{\ell}|>2R\}} \\
			&\geq \begin{multlined}[t][.9\textwidth]
				(1-\varepsilon)\sum_{i=1}^M \sum_{\substack{j=1 \\ j \neq i}}^M \left(p_i\1_{\{|x_i-x_j|<R\}} p_i + \frac{1}{2} W_N(x_i-x_j) \right)\prod_{\substack{\ell =1 \\ \ell \neq i,j}}^M \1_{\{|x_j-x_{\ell}|>2R \}} \\
				- C \varepsilon^{-1} \sum_{i=1}^M \sum_{\substack{j=1 \\ j \neq i}}^M p_i \1_{\{|p_i|<s\}}\1_{\{|x_i-x_j|<R\}} p_i\1_{\{|p_i|<s\}}
			\end{multlined} \\
			&\geq \sum_{i=1}^M \sum_{\substack{j=1 \\ j \neq i}}^M \frac{b(W)}{N} \left( 1-\frac{C}{RN} \right) U_R(x_i-x_j)\prod_{\substack{\ell =1 \\ \ell \neq i,j}}^M \1_{\{|x_j-x_{\ell}|>2R \}} - C_{\varepsilon} s^5 M^2 R^3\,.
		\end{aligned}
	\end{multline*}
	We used the Cauchy--Schwarz inequality for the second inequality and the third inequality is a consequence of Lemma~\ref{dyson_lemma_Rd_nonradial} together with the estimate
	\begin{multline*}
		\norm{ \1_{\{|p_i| \le s\}} p_i \1_{\{|x_i-x_j| \leq R\}} p_i \1_{\{|p_i| \le s\}} }_{HS} \\
		= (2\pi)^{-3} \norm{\1_{\{|p_i| \le s\}} |p_i| }_{L^2(\R^3, \d p_i)}^2 \norm{\1_{\{|x_i-x_j| \leq R\}}}_{L^2(\R^3, \d x_i)}^2 \le C s^5 R^3\,.
	\end{multline*}
	From this and the lower bound
	\[
		\prod_{\substack{\ell =1 \\ \ell \neq i,j}}^M \1_{\{|x_j-x_{\ell}|>2R \}} \ge 1- \sum_{\substack{\ell =1 \\ \ell \neq i,j}}^M (1-\1_{\{|x_j-x_{\ell}|>2R \}}) = 1- \sum_{\substack{\ell =1 \\ \ell \neq i,j}}^M \1_{\{|x_j-x_{\ell}|\leq 2R \}}\,,
	\]
	we obtain
	\begin{multline} \label{eq:E_2b_1}
		E^{\mathrm{2B}}(M,N)\\
		\geq \begin{multlined}[t]
			\bigg\< \sum_{i=1}^{M} (h_{\eps,s})_i + b(W) \sum_{1\leq i<j\leq M} M^{-1}U_R(x_i-x_j) \bigg\>_{\Gamma_{M,N}} \\
			- C \bigg\< \sum_{1\leq i<j\leq M} M^{-1} U_R(x_i-x_j) \bigg(\sum_{\substack{\ell =1 \\ \ell \neq i,j}}^M \1_{\{|x_j-x_{\ell}|\leq 2R \}} + \frac{1}{RN} + \varepsilon\bigg) \bigg\>_{\Gamma_{M,N}} \\
			- C_{\varepsilon,s} M^2 R^3\,.
		\end{multlined}
	\end{multline}

	Recall that $ N\ge M \geq (1-\varepsilon)N$. To bound the error term in~\eqref{eq:E_2b_1}, note that
	\begin{multline*}
		C R^{-3} \bigg\< \sum_{1\leq i<j\leq M} \sum_{\substack{\ell =1 \\ \ell \neq i,j}}^M M^{-1}U_R(x_i-x_j) \1_{\{|x_j-x_{\ell}|\leq 2R \}}\bigg\>_{\Gamma_{M,N}} \\
		\leq C R^{-3} M^2 \braket{\1_{|x_1-x_2|\leq 2R}\1_{|x_2-x_3|\leq 2R}}_{\Gamma_{M,N}} \leq C N\times M R^3
	\end{multline*}
	by the bosonic symmetry of $\Gamma_{M,N}$ and that, using~\eqref{eq:E_2b_1}, we also have
	\[
		\bigg\< \sum_{1\leq i<j\leq M} M^{-1}U_R(x_i-x_j) \bigg\>_{\Gamma_{M,N}} \leq C (E(M,N) + C M^2 R^3) \leq C N (1 + M R^3)\,.
	\]
	Hence,~\eqref{eq:E_2b_1} becomes
	\begin{multline*}
		E^{\mathrm{2B}}(M,N) \geq \bigg\< \sum_{i=1}^{M} (h_{\eps,s})_i + b(W) \sum_{1\leq i<j\leq M} M^{-1}U_R(x_i-x_j) \bigg\>_{\Gamma_{M,N}} \\
		- C_{\varepsilon,s} N (M^{-1} R^{-1} + M R^3) - C N \varepsilon\,.
	\end{multline*}
	To ensure that the error term $N (M^{-1} R^{-1} + M R^3)$ is $o(N)$, we need
	\[
		M^{-1/3} \gg R \gg M^{-1}\,.
	\]
	We can take for example
	\[
		R = (M^{-1} \times M^{-1/3})^{1/2} \sim N^{-2/3}\,.
	\]
	This is still too singular in order to apply the mean-field technique in~\cite[Theorem 2.5]{LewNamRou-16}.
	However, the main simplification over~\cite{LieSei-06,NamRouSei-16} comes when we use again the bosonic symmetry to rewrite~\eqref{eq:E_2b_1}, for $1\leq M_1 \leq M$, as
	\begin{multline} \label{eq:E_2b_2}
		\frac{E^{\mathrm{2B}}(M,N)}{M} \geq \frac{1}{M_1}\bigg\< \sum_{i=1}^{M_1} (h_{\eps,s})_i + \frac{b(W)}{M_1} \sum_{1\leq i<j\leq M_1}U_R(x_i-x_j) \bigg\>_{\Gamma_{M,N}} \\
		- C_{\varepsilon,s} (M^{-1} R^{-1} + M R^3) - C M \varepsilon\,.
	\end{multline}
	We will now apply Dyson's lemma to $\widetilde W = \delta^{-1}M_1^{-1}b(W)U_R$ which, by~\eqref{eq:bounds_on_w_general_2}, satisfies
	\[
		b(\widetilde W) \geq \|\widetilde W\|_{L^1(\R^3)} - C \|\widetilde W\|_{L^{6/5}(\R^3)}^2 \geq \frac{b(W)}{\delta M_1}\left(1-\frac{C}{\delta M_1R}\right).
	\]
	For $R_1 \gg R$, following the above estimates, we obtain
	\begin{multline*}
		\bigg\< \sum_{i=1}^{M_1} \delta p_i^2 + \frac{b(W)}{M_1} \sum_{1\leq i<j\leq M_1} U_R(x_i-x_j)\bigg\>_{\Gamma_{M,N}} \\
		\begin{aligned}[t]
			&\geq \begin{multlined}[t]
				\bigg\< \frac{b(W)}{M_1} \sum_{1\leq i<j\leq M_1} U_{R_1}(x_i-x_j)\bigg\>_{\Gamma_{M,N}} \\
				- C \bigg\< \frac{1}{M_1} \sum_{1\leq i<j\leq M_1} U_{R_1}(x_i-x_j)\bigg(\sum_{\substack{\ell =1 \\ \ell \neq i,j}}^{M_1} \1_{\{|x_j-x_{\ell}|\leq 2R_1 \}} + \frac{1}{\delta M_1R}\bigg)\bigg\>_{\Gamma_{M,N}}
			\end{multlined} \\
			&\geq \bigg\< \frac{b(W)}{M_1} \sum_{1\leq i<j\leq M_1} U_{R_1}(x_i-x_j)\bigg\>_{\Gamma_{M,N}} - C_\delta M_1 \left( M_1 R_1^3 + M_1^{-1} R^{-1} \right).
		\end{aligned}
	\end{multline*}
	In order for the error $M_1( M_1 R_1^3 + M_1^{-1} R^{-1})$ to be $o(M_1)$, we need
	\[
		M_1^{-1/3} \gg R_1 \gg R \gg M_1^{-1} \geq M^{-1}\,,
	\]
	which can be satisfied with for example the choices
	\begin{align*}
		M_1 &= (M \times R^{-1})^{1/2} = ( R^{-3/2} \times R^{-1})^{1/2} = R^{-5/4}
		\intertext{and}
		R_1 &= (R \times M_1^{-1/3})^{1/2} =(R \times R^{5/12})^{1/2} = R^{17/24}\,.
	\end{align*}
	Now, using~\eqref{eq:p_h} and that $\delta \in (0,\eps^2/2)$, we have $\delta p^2 \le \eps h_{\eps,s} + \delta C_\eps$. Inserting the above inequality in~\eqref{eq:E_2b_2}, we obtain
	\begin{multline*}
		\frac{E^{\mathrm{2B}}(M,N)}{(1-\varepsilon) M} \geq \frac{1}{M_1} \bigg\< \sum_{i=1}^{M_1} (h_{\eps,s})_i + \frac{b(W)}{M_1} \sum_{1\leq i<j\leq M_1}U_{R_1}(x_i-x_j) \bigg\>_{\Gamma_{M,N}} \\
		- C_{\varepsilon,s} R^{1/4} - C \varepsilon - C_{\varepsilon} \delta\,.
	\end{multline*}
	Defining $R_0 := R$ and, for $j = 1,\dots,J$, $R_j := R_{j-1}^{17/24} = R^{(17/24)^j}$ and $M_j := R_{j-1}^{-5/4}$, repeating this argument $J \geq 1$ times yields
	\begin{align*}
		\frac{E^{\mathrm{2B}}(M,N)}{(1-\varepsilon) M}
		&\begin{multlined}[t][.83\textwidth]
			\geq\frac{1}{M_J} \bigg\< \sum_{i=1}^{M_J} (h_{\eps,s})_i + \frac{b(W)}{M_J} \sum_{1\leq i<j\leq M_J}U_{R_J}(x_i-x_j) \bigg\>_{\Gamma_{M,N}} \\
			- C_{\varepsilon,s} R^{1/4^J} - C \varepsilon - C_{\varepsilon} \delta
		\end{multlined} \\
		&\begin{multlined}[t][.83\textwidth]
			\geq \frac{1}{M} \bigg\< \sum_{i=1}^{M} (h_{\eps,s})_i + \frac{b(W)}{M} \sum_{1\leq i<j\leq M}U_{R_J}(x_i-x_j) \bigg\>_{\Gamma_{M,N}} \\
			- C_{\varepsilon,s} R^{1/4^J} - C \varepsilon - C_{\varepsilon} \delta\,.
		\end{multlined}
	\end{align*}
	Note that the double sum in the first inequality has $M_J(M_J-1)/2$ terms. Therefore, a correction should arise from the approximation $(M_J-1)/M_J \geq 1 - C M_J^{-1}$ but, using $\norm{ U_{R_J} }_{L^\infty(\mathbb{R}^{3})} \leq C R_J^{-3}$, it is also bounded by $C M_J R_J^{-3} \leq C R^{1/4^J}$. For every $\beta\in (0,2/3)$, taking $J\in \N\setminus\{0\}$ such that $N^{-\beta/2} \geq R_J= R^{(17/24)^J} \geq N^{-\beta}$ (recall that $R\sim N^{-2/3}$), we conclude the proof.
\end{proof}

\end{document}